%% file: main.tex
\documentclass[aos2]{imsart}
\RequirePackage{amsthm,amsmath,amsfonts,amssymb}
\RequirePackage[authoryear]{natbib}
\usepackage{scribe}
\RequirePackage[colorlinks,citecolor=blue,urlcolor=blue]{hyperref}
\usepackage{enumerate}
\usepackage{algpseudocode}
\usepackage{bbm}
\usepackage{chngcntr}
\usepackage{apptools}
\AtAppendix{\counterwithin{theorem}{section}}
\AtAppendix{\counterwithin{lemma}{section}}
\AtAppendix{\counterwithin{corollary}{section}}
\usepackage[linesnumbered,ruled,vlined]{algorithm2e}
\SetKwInput{KwInput}{Input}                
\SetKwInput{KwOutput}{Output}              
\newcommand{\bE}{\mathbb{E}}
\newcommand{\bP}{\mathbb{P}}

\usepackage{bm, amsmath, amssymb, graphicx, multirow, bbm, xcolor, subcaption, caption, tabularx, float, hyperref}
\graphicspath{ {./images/} }

\begin{document}

\begin{frontmatter}
\title{The Generalized Elastic Net for least squares regression with network-aligned signal and correlated design}
\runtitle{The Generalized Elastic Net}

\begin{aug}
\author[A]{\fnms{Huy}~\snm{Tran}\ead[label=e1]{huydtran@uchicago.edu}},
\author[A]{\fnms{Sansen}~\snm{Wei}\ead[label=e2]{sansenw@uchicago.edu}}
\and
\author[A]{\fnms{Claire }~\snm{Donnat}\ead[label=e3]{cdonnat@uchicago.edu}}
\runauthor{H. Tran, S. Wei and C. Donnat}\\

\address[A]{Department of Statistics\\
The University of Chicago\\
Chicago, Illinois 60637\\
USA\\
\printead{e1}
\phantom{E-mail: }\printead*{e2}
\phantom{E-mail: }\printead*{e3}
}
\end{aug}


\begin{abstract}
\input{abstract.tex}
\end{abstract}

\begin{keyword}[class=MSC]
\kwd[Primary ]{62J05}
\kwd{62J07}
\kwd[; secondary ]{62F12}
\end{keyword}

\begin{keyword}
\kwd{elastic net}
\kwd{Fused Lasso}
\kwd{network alignment}
\kwd{Smooth Lasso}
\kwd{sparsity}
\kwd{total variation}
\end{keyword}

\end{frontmatter}

\section{Introduction}\label{sec:intro}
\input{intro.tex}
\section{Theoretical results}\label{sec: theory}
\input{theory.tex}
\section{Computation}\label{sec: computation}
\input{computation}
\section{Experiments}
\label{sec: experiments}

\input{experiments.tex}

\section{Conclusion} \label{sec:conclusion}
\input{conclusion}

\bibliographystyle{rss}
\bibliography{citation}       
\newpage
\begin{appendix}
\section{Proofs of theoretical results}
\input{appendix_proofs}
\section{The interior point method on the dual objective}
\input{appendix_algo}

\section{Additional details on data processing}\label{app:data}
\input{appendix_data}
\end{appendix}


\end{document}

%% file: abstract.tex
We propose a novel $\ell_1+\ell_2$-penalty, which we refer to as the \textit{Generalized Elastic Net}, for regression problems where the feature vectors are indexed by vertices of a given graph and the true signal is believed to be smooth or piecewise constant with respect to this graph. Under the assumption of correlated Gaussian design, we derive upper bounds for the prediction and estimation errors, which are graph-dependent and consist of a parametric rate for the unpenalized portion of the regression vector and another term that depends on our network alignment assumption. We also provide a coordinate descent procedure based on the Lagrange dual objective to compute this estimator for large-scale problems. Finally, we compare our proposed estimator to existing regularized estimators on a number of real and synthetic datasets and discuss its potential limitations. 

%% file: intro.tex
\subsection{Problem formulation and the proposed penalty}
Consider the usual linear regression model
\begin{equation}\label{eq:1}
    Y = X\beta^* + \epsilon
\end{equation}
where the design matrix $X\in \mathbb{R}^{n\times p}$ is random with independent and identically distributed (i.i.d.) rows, $\beta^* \in \mathbb{R}^p$ is the unknown true parameter, and $\epsilon = (\epsilon_1, \dots, \epsilon_n)^T \in \mathbb{R}^n$ are i.i.d. zero-mean Gaussian variables with (unknown) variance $\sigma^2$ and are independent of the design matrix $X$. In addition to observing the responses $y = (y_1, \dots, y_n)^T \in \mathbb{R}^n$, we also observe an undirected simple graph $G = (V, E)$ with $p$ vertices and $m$ edges. Here, the $p$ vertices index the  entries of $\beta^*$ as well as the columns of $X$ (which we can think of as feature vectors). This situation typically entails significant correlation between feature vectors, thus leading to an ill-conditioned design matrix. For simplicity, we assume throughout this paper that the rows of $X$ are i.i.d. $N(0,\Sigma)$-distributed. In our setting, the minimum eigenvalue of $\Sigma\in \mathbb{R}^{p\times p}$ may be very small and $\beta^*$ might be nearly unidentifiable. Although the addition of an unpenalized intercept should present no difficulty, we assume no intercept in our model to simplify the theoretical analysis. 

We further assume that $\beta^*$ is structured with respect to the graph $G$ so that prediction and estimation can be done with small error, even in the high-dimensional setting where $p \gg n$. As the entries of $\beta^*$ are indexed by the vertices of $G$, a natural assumption is that $\beta^*_i$ and $\beta^*_j$ should be similar if $i$ and $j$ are adjacent vertices on the graph $G$. This assumption is related to the notion of \textit{network cohesion} as discussed in Chapter 4 of \cite{kolaczyk2009statistical}: vertices may display similar characteristics because they are connected (contagion), or they may be connected because they have similar characteristics (homophily). Note that, however, many prior works such as \cite{li2019prediction} discuss network cohesion in the context where observations (the responses $y_1, \dots, y_n$ and the rows $x_1, \dots, x_n$ of $X$) are indexed by a graph's vertices and thus may no longer be i.i.d., whereas we focus on the case where the features (the columns of $X$) are indexed by the graph's vertices. Following \cite{li2019prediction}, we also use the term \textit{network cohesion} to cover both homophily and contagion, without distinguishing the difference in causal direction between them. 

More specifically, the notion of network cohesion encourages us to assume either that the number of edges $(i,j)\in E$ where $\beta_i \neq \beta_j$ is small (sparse signal jumps), or that $\beta^*$ is smooth enough so that $\Gamma\beta^*$ lies in an $\ell_q$-ball, where $\Gamma$ is the edge-incidence matrix of the graph $G$ and $0 < q \leq 1$ (note that when $q\in (0,1)$, an $\ell_q$-ball is not convex - see Figure 7.1 of \cite{wainwright2019high} for an illustration of what this "ball" looks like). Mathematically, in our theoretical analysis we assume either
\begin{equation}\label{eq:2}
    \|\Gamma\beta^*\|_0 \leq s
\end{equation}
or 
\begin{equation}\label{eq:3}
    \sum_{j=1}^m|(\Gamma\beta^*)_j|^q \leq R_q
\end{equation}
for some $s \ll m$ or some $R_q > 0$, respectively (see Section \ref{subsection:def} for the precise definition of any mathematical symbol). Assumption \eqref{eq:2} means that the number of edges with nonzero signal jumps is small and the true signal has several piecewise constant regions on the graph, whereas Assumption \eqref{eq:3} means the signal is smooth over the graph in the $\ell_q$-sense. We use the term \textit{network alignment} to refer to either Assumption \eqref{eq:2} or Assumption \eqref{eq:3}. In our experiments, we sometimes also consider the notion that $\beta^*$ is smooth over the graph in the sense that $\|\Gamma\beta^*\|_\infty$ is small. We emphasize that we allow for the possibility that the entries $\beta^*$ are all nonzero, as long as $\beta^*$ satisfies (or can be well approximated by an oracle $\beta$ that satisfies) either Assumption \eqref{eq:2} or Assumption \eqref{eq:3}. 

Under Model \eqref{eq:1} and either Assumption \eqref{eq:2} or \eqref{eq:3}, we study the prediction and estimation errors of the following estimator 
\begin{equation}\label{eq:4}
    \hat{\beta} := \arg\min_{\beta\in\mathbb{R}^p}\frac{1}{n}\|Y-X\beta\|_2^2 + \lambda_1\|\Gamma\beta\|_1 + \lambda_2 \|\Gamma\beta\|_2^2
\end{equation}
which can also be rewritten as 
\begin{equation}
    \hat{\beta} := \arg\min_{\beta\in\mathbb{R}^p}\frac{1}{n}\|Y-X\beta\|_2^2 + \lambda_1\sum_{(i,j)\in E}|\beta_i-\beta_j| + \lambda_2 \sum_{(i,j)\in E}(\beta_i-\beta_j)^2
\end{equation}

Note that our focus is mainly on the penalty $\lambda_1 \|\Gamma\beta\|_1 + \lambda_2\|\Gamma\beta\|_2^2$ where $\Gamma$ is the incidence matrix of a general graph. Following the naming conventions in \cite{zou2005regularization} and \cite{tibshirani2011solution}, we refer to this penalty as the \textit{Generalized Elastic Net (GEN)} penalty. The estimator \eqref{eq:4} can be easily extended to the generalized linear model (GLM) setting, by replacing the term $\frac{1}{n}\|Y-X\beta\|_2^2$ with another negative log-likelihood function from an exponential family distribution (see Chapter 9 of \cite{wainwright2019high} for more examples). For instance, if we have binary responses that can be modeled with the logistic GLM, then using the logistic log-likelihood function gives
\begin{equation}\label{eq:6}
\hat{\beta}_{\text{logistic}} := \frac{1}{n}\sum_{i=1}^n\log(1+e^{\langle x_i, \beta\rangle}) - \bigg\langle\frac{1}{n}\sum_{i=1}^n y_ix_i, \beta\bigg\rangle + \lambda_1\|\Gamma\beta\|_1 + \lambda_2\|\Gamma\beta\|_2^2    
\end{equation}
where $x_1, \dots, x_n$ are the rows of $X$ and $y_1, \dots, y_n$ are the entries of $Y$ which are binary. For simplicity, we only focus on analyzing the estimator \eqref{eq:4} under Model \eqref{eq:1}, but analogous theoretical results for the GLM setting should follow by adapting the theoretical framework of Chapter 6 of \cite{buhlmann2011statistics}. 
\subsection{Motivating applications}
As network-linked features are quite common and we do not restrict our attention to any particular type of graph, our proposed penalty is potentially applicable to a wide variety of settings. We provide below a non-exhaustive list of concrete examples where our penalty may be relevant.\\

\noindent{\bf Example 1: Structural MRI analysis.} We consider the use of structural magnetic resonance images (sMRI) of the brain in diagnosing Alzheimer's disease, as in \cite{xin2014efficient}. In this case, the rows $x_1, \dots, x_n$ of $X$ might represent sMRI features of $n$ human subjects and the responses $y_1, \dots, y_n$ are binary variables indicating each subject's disease status. The estimator \eqref{eq:6} can thus be applied using a 3D grid graph representing contiguous brain voxels. In \cite{xin2014efficient}, the Generalized Fused Lasso penalty $\lambda_L \|\beta\|_1 + \lambda_1 \|\Gamma\beta\|_1$ is used, and this penalty leads to a solution that is both sparse and smooth. However, it may be more reasonable to assume only that the true signal aligns with the graph, in which case the estimator \eqref{eq:6} may fare better for the purpose of predicting Alzheimer's disease.\\

\noindent{\bf Example 2: Microarray analysis with prior information.} Following \cite{segal2003regression}, we can also consider a microarray dataset with $X = [x_{ij}]$ where $x_{ij}$ is the expression level of the $j^\text{th}$ gene for the $i^\text{th}$ test subject, and $y_i$ is an outcome measure for subject $i$ which can be continuous or discrete. Often, we have prior knowledge from previous biomedical research in the form of gene regulatory pathways which can serve as our graph $G$ (see \cite{li2008network} for specific examples). We can incorporate this prior information using our GEN penalty. In \cite{li2008network}, the penalty $\lambda_L\|\beta\|_1 + \lambda_2\beta^T \tilde{L}\beta$ is used instead, where $\tilde{L}$ is the normalized Laplacian matrix. Assuming the vast majority of genes has no effect on the outcome may make it easier to interpret the estimated parameters (in terms of which genes may be responsible for the outcome). However, if many of these genes can be grouped into clusters with small (but nonzero) baseline effects on the outcome, using our penalty may lead to better predictions. \\

\noindent{\bf Example 3: Microarray analysis without prior information.} In the previous example, without any prior information about gene regulatory pathways, we can take $G$ to be the complete graph in our GEN penalty. The penalty $\lambda_L\|\beta\|_1 + \lambda_1\|\Gamma\beta\|_1$, where $\Gamma$ is the incidence matrix of a complete graph, has been studied in \cite{she2008sparse} under the name \textit{Clustered LASSO}. \\

\noindent{\bf Example 4: Temporal data.} Given a time series $\{X_t\}_{t\in \mathbb{N}}$, we consider fitting an autoregressive model of the form $X_t = \sum_{j=1}^p\beta_jX_{t-j}+\epsilon_t$. If the time points $t$ are sampled sufficiently far apart such that our data points $(X_{t}, X_{t-1},\dots, X_{t-p})$ can be considered independent across $t$, it may be reasonable to apply our method with $G$ being a $p$-vertex chain graph. 
 
\subsection{Comparison with related works}\label{sec:related}
The standalone $\ell_1$ penalty $\lambda_1\|\Gamma\beta\|_1$, which is often known as the \textit{total variation penalty} on graphs, has been studied extensively in the context of the graph trend filtering problem where the design matrix is the identity. More precisely, given the model $Y = \beta^* + \epsilon$, the trend filtering estimator for $\beta^*$ is 
\begin{equation}\label{eq: 7}
    \hat{\beta}_{\text{tf}} := \arg\min_{\beta\in \mathbb{R}^n}\frac{1}{n}\|Y-\beta\|_2^2 + \lambda_1\|\Gamma\beta\|_1
\end{equation}

This estimator is also known as the \textit{analysis} estimator, in the terminology of \cite{elad2007analysis}; see \cite{hutter2016optimal}, \cite{wang2015trend}, \cite{ortelli2021prediction} and \cite{guntuboyina2020adaptive} for results on prediction error bounds for the estimator \eqref{eq: 7} and its constrained form when $\Gamma\beta^*$ is sparse. The graph considered in the trend filtering problem is usually a chain or grid graph due to applications such as image denoising, but results for other types of graphs such as trees and star graphs are also available in the literature. The analysis matrix $\Gamma$ in \eqref{eq: 7} can be generalized to higher order total variation operators (as defined in \cite{wang2015trend}). In comparison, we focus solely on the case where $\Gamma$ is the incidence matrix defined in \eqref{eq:Gamma}, and our design matrix $X$ is random with i.i.d. rows rather than a pre-specified matrix consisting of fixed vectors from some dictionary.

When the design matrix is general, the estimator 
\begin{equation}\label{eq:8}
    \hat{\beta}_{\text{GL}} := \arg\min_{\beta\in\mathbb{R}^p}\frac{1}{n}\|Y-X\beta\|_2^2 + \lambda_1\|\Gamma\beta\|_1
\end{equation}
has been proposed by \cite{tibshirani2011solution} (under the name \textit{Generalized LASSO} estimator) as well as \cite{land1997variable} (where the penalty is called \textit{variable fusion}). These works mainly address computational techniques for the estimator \eqref{eq:8}, rather than theoretical guarantees when $\beta^*$ aligns with the graph. The idea of working with the dual objective to derive our algorithms comes from \cite{kim2009ell_1} and \cite{tibshirani2011solution}. Our analysis of the prediction and estimation errors for the estimator \eqref{eq:4} is also applicable to \eqref{eq:8}, and to our knowledge no similar analysis with random design is available in the literature. However, the error bounds for our estimator \eqref{eq:4} are better due to the improved minimum eigenvalue in the denominator of the bounds in Theorem \ref{theorem:1}. 

Two previously introduced penalties which involve the Lasso penalty to induce sparsity but are closely related to GEN have also been studied in the context where the design matrix can be non-identity; they serve as the main benchmarks in both our theoretical results and our experiments. The Smooth Lasso penalty $\lambda_L \|\beta\|_1 + \lambda_2\|\Gamma\beta\|_2^2$ was first proposed by \cite{hebiri2011smooth}, in which the theoretical analysis assumes fixed design and thus relies on a restricted eigenvalue assumption (Assumption $B(\Theta)$ in \cite{hebiri2011smooth}) on the expanded Gram matrix $n^{-1}\tilde{X}^T\tilde{X}$ (see Section \ref{subsection:def} for definition of $\tilde{X}$). The Fused Lasso penalty $\lambda_L\|\beta\|_1 + \lambda_1\|\Gamma\beta\|_1$ was first proposed by \cite{tibshirani2005sparsity} for the chain graph. These two methods implicitly assume that the true signal is both sparse and aligned with the graph. Such an assumption can be overly restrictive, and sparsity of $\beta^*$ may not always be a natural assumption in the general graph setting. When $\|\beta^*\|_0 = p$, error bounds proven for these estimators usually involve the term $\frac{p\log p}{n}$. In comparison, our penalty only assumes network alignment and should also work well in the sparse-and-smooth case when the zero entries of $\beta^*$ form large contiguous blocks on the graph. The Fused Lasso and the Smooth Lasso should only perform better than ours when sparsity holds but the network alignment assumption is significantly violated. Empirically, when $\beta^*$ aligns with the graph but is not sparse, choosing the tuning parameters by cross-validation often results in $\lambda_L$ being set to almost zero for both the Fused Lasso and the Smooth Lasso.

In \cite{li2018graph}, the penalty $\lambda_2\|\hat{\Gamma}\beta\|_2^2 + \lambda_1\|\hat{\Gamma}\beta\|_1+\lambda_L\|\beta\|_1$ is introduced and referred to as the Graph Total Variation (GTV) method, which involves three hyperparameters that require tuning. Unlike our penalty, the incidence matrix $\hat{\Gamma}$ is obtained by first estimating $\Sigma$ with $\hat{\Sigma}$ (which can depend on the design $X$ or side information) and then treating $\hat{\Sigma}$ as the adjacency matrix of a graph $\hat{G}$ with weighted edges. Note that this is a two-step process, and the graph $\hat{G}$ here also differs from our setting in that we do not consider non-binary edge weights, since in many applications only a graph structure is provided. Computationally, since we need to use 3D grid search for hyperparameter tuning and the matrix $\hat{\Gamma}$ is very dense, the estimator introduced in \cite{li2018graph} does not scale well. Furthermore, even when we use the true covariance $\Sigma$ to form $\hat{\Gamma}$, the performance of GTV in most of our synthetic experiments does not compare favorably with that of our method, Fused Lasso or Smooth Lasso. The theoretical analysis in \cite{li2018graph} does not account for the error in estimating $\Sigma$ with $\hat{\Sigma}$, which we believe cannot be overlooked.

\subsection{Organization of the paper} In Section \ref{sec: theory}, we provide upper bounds on the prediction and estimation errors of our estimator \eqref{eq:4}. We specialize these bounds for the cases when $\Gamma$ is the incidence matrix of the $r$-dimensional grid graph, the star graph or the complete graph, which are also discussed in \cite{hutter2016optimal} and may be encountered in practice. In Section \ref{sec: computation}, we describe our coordinate descent algorithm in order to compute the estimator \eqref{eq:4} for a particular choice of $(\lambda_1, \lambda_2)$. We also provide runtime comparisons between various techniques to compute \eqref{eq:4}, such as the interior point method and the Alternating Direction Method of Multipliers (ADMM). In Section \ref{sec: experiments}, we use experiments on synthetic data to provide empirical evidence of some properties of GEN that are suggested by our theoretical results, and compare its performance with estimators from prior works that we discuss in Section \ref{sec:related}. We also illustrate the usefulness of our proposed penalty in three real data examples, namely COVID-19 trend prediction, Alzheimer's disease detection and estimation of Chicago's crime patterns. 

\subsection{Notations and definitions}\label{subsection:def}
For any positive integer $n$, we denote $[n]$ as the set $\{1,\dots, n\}$. For any matrix $A$, we denote by $A^\dagger$ the Moore-Penrose inverse of $A$. For any vector $v$, $\|v\|_0$ refers to the number of nonzero entries of $v$, and $\|v\|_p$ for $1\leq p \leq \infty$ refers to the usual $\ell_p$-norm of $v$. We write $\mathbf{1}(\cdot)$ for the indicator function. For a vector $v \in \mathbb{R}^k$ and any set $S \subseteq [k]$, we denote by $v_S \in \mathbb{R}^k$ to be the vector with the $j^\text{th}$ coordinate given by $(x_S)_j = x_j\mathbf{1}(j\in S)$. For any vector $\theta \in \mathbb{R}^m$, we write $S_{\theta}$ to refer to the support $\{j\in [m]: \theta_j \neq 0\}$ of $\theta$. We use $s$ to denote $\|\Gamma\beta^*\|_0$. For any positive semi-definite matrix $M$, let $\gmax(M)$ and $\gmin(M)$ denote its maximum and minimum eigenvalues respectively, and $\text{ker}(M)$ the null space of $M$. $I_k$ denotes the identity matrix of size $k$-by-$k$.

The notation \( \lesssim \) means that the left-hand side (LHS) is bounded by the right-hand side (RHS) multiplied by an absolute constant (not dependent on any parameter of interest) that is omitted. The notation $\gtrsim$ is similarly defined. The notation $\asymp$ means that both $\lesssim$ and $\gtrsim$ hold. The constants $C$, $c$, $c_1, c_2$ are absolute constants which are allowed to change line by line.

Throughout this paper, the graph $G = (V,E)$ we consider is undirected and has no self-loops. We identify the set of vertices $V$ with $[p]$ and the set of edges $E$ with $[m]$; note that $m \leq p^2$, and $p \lesssim m$ if the graph has no isolated vertices. We also denote the maximum degree of the graph $G$ by $d$ and the number of connected components of $G$ by $n_c$ (which is also the dimension of the null space of $\Gamma$). The edge-vertex incidence matrix of the graph $G$ is denoted by $\Gamma \in \{-1,0,1\}^{m\times p}$, which is defined as follows: each edge $e = (i,j)\in E$ is represented by a row $\Gamma_{e,\cdot}\in \{-1,0,1\}^p$ of $\Gamma$ whose $k^\text{th}$ entry is given by 
\begin{equation}\label{eq:Gamma}
\Gamma_{e,k}=\left\{
\begin{array}{rl}
1 &\text{if}\ k=\min(i, j)\\
-1 &\text{if}\ k=\max(i,j)\\
0 &\text{otherwise.}
\end{array}
\right.    
\end{equation}
The unnormalized Laplacian matrix of the graph $G$ (see \cite{chung1997spectral}) is then defined by $L := \Gamma^T \Gamma$. We denote by $\Pi \in \mathbb{R}^{p\times p}$ the projection matrix onto the kernel of $\Gamma$. Note that we will use the facts $\Pi = \Pi^T$, $\Pi^2 = \Pi$ and $\Pi + \Gamma^\dagger\Gamma = I_p$ throughout the proofs.

In our theoretical analysis, we frequently make use of some definitions and conventions from \cite{hutter2016optimal}. We denote $s_1, \dots, s_m$ to be the columns of $\Gamma^\dagger \in \mathbb{R}^{p\times m}$. The \textit{inverse scaling factor} of $\Gamma$ is defined as 
\begin{equation}
\rho(\Gamma) := \max_{j\in [m]}\|s_j\|_2  
\end{equation}
while the \textit{compatibility factor} of $\Gamma$ for a nonempty set $S\subseteq [m]$ is defined as 
\begin{equation}\label{eq:compatibility-factor}
k_S := \inf_{\beta\in \mathbb{R}^p}\frac{\sqrt{|S|}\|\beta\|_2}{\|(\Gamma\beta)_S\|_1}  
\end{equation}

Following \cite{hebiri2011smooth}, we also employ the notations
\begin{equation}
\tilde{Y} := \begin{pmatrix} Y \\ 0\end{pmatrix}, \quad \tilde{X} := \begin{pmatrix}X \\ \sqrt{\lambda_2n} \Gamma\end{pmatrix}, \quad \tilde{\epsilon} := \begin{pmatrix}\epsilon \\ -\sqrt{\lambda_2 n }\Gamma \beta^*\end{pmatrix}  
\end{equation}
Note that $\tilde{Y} = \tilde{X}\beta^* + \tilde{\epsilon}$ and we can write our estimator as 
\begin{equation}
    \hat{\beta} = \arg\min_{\beta\in\mathbb{R}^p}\frac{1}{n}\|\tilde{Y}-\tilde{X}\beta\|_2^2 + \lambda_1\|\Gamma\beta\|_1
\end{equation}

%% file: theory.tex
In this section, we aim to provide non-asymptotic bounds showing that the estimator \eqref{eq:4} is consistent in prediction and estimation under a network alignment assumption, even in the high-dimensional setting where $p\gg n$. We also show that the $\ell_2$ component of the penalty helps alleviate the effects of an ill-conditioned covariance matrix $\Sigma$.  Note that the tuning parameters $\lambda_1$ and $\lambda_2$ in our theoretical analysis are dependent on unobserved quantities $\beta^*$, $\Sigma$ and $\sigma$; therefore, we cannot use the theoretical values for $\lambda_1$ and $\lambda_2$ in practice and must in general rely on cross-validation. We do not attempt to optimize the constants in our bounds, as our focus is on understanding how the performance of our estimator depends on the quantities $n$, $p$, $s$ (or $R_q$), $\Sigma$ and the graph $G$. All proofs are deferred to the Appendix.

\subsection{Main theorem}
We begin by introducing bounds for the prediction and estimation errors that are applicable to all graphs. However, these bounds may not be optimal for some graphs, especially the $p$-vertex chain graph as in that case $\rho(\Gamma)=\sqrt{p}$. The proof of Theorem \ref{theorem:1} relies on the projection argument used in \cite{hutter2016optimal} to derive error bounds for the trend filtering estimator \eqref{eq: 7}. For simplicity, in the discussion of our theoretical results, we assume that $\gmax(\Sigma)$, $n_c$ and $\sigma^2$ are of constant order as $n$ goes to infinity. Recall that $n_c$ is the dimension of $\ker(\Gamma)$, $d$ is the maximum degree of all vertices of $G$, $L := \Gamma^T\Gamma$, and $k_S$ is defined in \eqref{eq:compatibility-factor}. 
\begin{theorem}[Main theorem]\label{theorem:1}
Fix $\delta > 0$ and choose $\lambda_1 = 32\sigma \rho(\Gamma)\sqrt{\frac{\gmax(\Sigma)\log p}{n}}$, $\lambda_2 \leq \frac{\lambda_1}{8\|\Gamma\beta^*\|_\infty}$. Given any set $S$ satisfying both
\begin{equation}\label{eq:regularity}
\frac{144\gmax(\Sigma)(\sqrt{n_c}+\delta)^2}{n} +\frac{36\lambda_1^2|S|k_S^{-2}}{\sigma^2} \leq \frac{1}{2}\gmin\left(\frac{1}{64}\Sigma + \lambda_2 L\right) 
\end{equation}
and 
\begin{equation}\label{eq:regularity2}
    \lambda_1\|(\Gamma\beta^*)_{-S}\|_1 \leq \frac{\sigma^2}{18}
\end{equation}
with probability at least $1-c_1\exp(-nc_2) - \frac{2}{m} - e^{-\delta^2/2}$ we have 
\begin{equation}\label{eq:pred-error-main}
\|\Sigma^{1/2}(\hat{\beta}-\beta^*)\|_2^2 \lesssim \frac{\sigma^2\gmax(\Sigma)}{\gmin\left(\frac{1}{64}\Sigma+\lambda_2L\right)}\frac{n_c+\delta^2}{n} + \frac{\lambda_1^2|S|k_S^{-2}
}{\gmin\left(\frac{1}{64}\Sigma+\lambda_2L\right)} + \lambda_1\|(\Gamma\beta^*)_{-S}\|_1
\end{equation}
\begin{equation}\label{eq:est-error-main}
\|\hat{\beta}-\beta^*\|_2^2 \lesssim \frac{\sigma^2\gmax(\Sigma)}{\gmin^2\left(\frac{1}{64}\Sigma+\lambda_2L\right)}\frac{n_c+\delta^2}{n} + \frac{\lambda_1^2|S|k_S^{-2}
}{\gmin^2\left(\frac{1}{64}\Sigma+\lambda_2L\right)} + \frac{\lambda_1\|(\Gamma\beta^*)_{-S}\|_1}{\gmin\left(\frac{1}{64}\Sigma+\lambda_2L\right)}
\end{equation}
\end{theorem}

Note that Theorem \ref{theorem:1} is actually valid for any matrix $\Gamma$. However, if $\Gamma$ is the incidence matrix of the graph $G$, we can further bound $k_{S}^{-2}$ by applying Lemma 3 of \cite{hutter2016optimal}, which states that $k_S^{-2} \lesssim \min(d, |S|)$. 

Denote $\beta^* = \beta^*_1 + \beta^*_2$, where $\beta^*_1\in \text{ker}(\Gamma)$ and $\beta^*_2 \in \text{ker}(\Gamma)^\perp$. Note that the first term in the RHS of \eqref{eq:est-error-main} represents the error from estimating $\beta^*_1$, which is the unpenalized component of $\beta^*$. The latter two terms represent the error from estimating the penalized component $\beta_2^*$, and given a particular graph $G$ we need to further bound $\rho(\Gamma)$ and $k_{S}^{-2}$ for that graph. 

The estimation error bound \eqref{eq:est-error-main} is only different from the prediction error bound \eqref{eq:pred-error-main} by a factor of $\gmin\left(\frac{1}{64}\Sigma+\lambda_2 L\right)$ in the denominator. This means we have to make a stronger assumption about how fast $\gmin\left(\frac{1}{64}\Sigma+\lambda_2L\right)$ may decay to zero in order to ensure the estimation error, rather than just the prediction error, is also small. For example, when we specialize our bounds for the 3D grid with $p$ vertices, the prediction error bound \eqref{eq:rgrid-pred} only requires $\gmin\left(\frac{1}{64}\Sigma+\lambda_2L\right) \gg \frac{s\log p}{n}$ but the estimation error bound for this graph requires $\gmin\left(\frac{1}{64}\Sigma+\lambda_2L\right) \gg \sqrt{\frac{s\log p}{n}}$. 

The conditions \eqref{eq:regularity} and \eqref{eq:regularity2} on $S$ are the result of using Lemma \ref{lem:re}. They are equivalent to requiring that the RHS of \eqref{eq:pred-error-main} is sufficiently small (smaller than $C\sigma^2$ for some absolute constant $C > 0$). Assuming $\gmin\left(\frac{1}{64}\Sigma + \lambda_2L\right)$ is not too small, it is reasonable to expect the prediction error to converge to zero as $n$ become sufficiently large. 

Theorem \ref{eq:1} is applicable to the estimator $\hat{\beta}_{\text{GL}}$ in \eqref{eq:8} (which corresponds to setting $\lambda_2 = 0$). However, when $\gmin(\Sigma)$ is small, we may not have a meaningful error bound for $\hat{\beta}_{\text{GL}}$.  Generally, we want $\lambda_2$ to be as large as possible to improve the minimum eigenvalue term without introducing additional errors, and thus the choice $\lambda_2 = \frac{\lambda_1}{8\|\Gamma\beta^*\|_\infty}$ is appropriate. When $\|\Gamma\beta^*\|_\infty$, the maximum signal difference between adjacent vertices, is small (which is reasonable under the assumption of network cohesion on $\beta^*$) and $\gmin(\Sigma)$ is very close to zero, the improvement of the minimum eigenvalue term can be significant. In contrast, in Theorem 3 of \cite{hebiri2011smooth} and Theorem 1 of \cite{li2018graph}, similar proof ideas are used but the dependence between the $\ell_2$ and $\ell_1$ tuning parameters is such that $\lambda_2 \propto \frac{\lambda_1}{\|L\beta^*\|_\infty}$. Since $L$ is the second-order graph difference operator (see \cite{wang2015trend} for the definitions of higher-order total variation operators), the quantity $\|L\beta^*\|_\infty = \|\Gamma^T\Gamma\beta^*\|_\infty$ is not as related to Assumption \eqref{eq:2} or \eqref{eq:3} and can be much larger than $\|\Gamma\beta^*\|_\infty$ for graphs with some high-degree nodes. For example, for the star graph with $p$ nodes where the entries of $\beta^*$ are 0 at the central node and 1 at the leaves, $\|\Gamma\beta^*\|_\infty = 1$ but $\|L\beta^*\|_\infty$ is of order $p$. The choice of $\lambda_2$ in Theorem \ref{theorem:1}, however, suggests that the regularization effects of the $\ell_2$ component of the penalty may be diminished if $\|\Gamma\beta^*\|_\infty$ is large. This is consistent with what we observe in our synthetic experiments: when $\|\Gamma\beta^*\|_\infty$ is large, cross-validation often yields $\lambda_2 \approx 0$.

The proof of Theorem \ref{theorem:1} relies on the following lemma to relate the empirical quadratic form $\frac{1}{n}\|Xv\|_2^2$ to the corresponding theoretical quantity $\|\Sigma^{1/2} v\|_2^2$, uniformly for all $v \in \mathbb{R}^p$. This lemma is an extension of the main result in \cite{raskutti2010restricted} for our setting and may be of independent interest.

\begin{lemma}[Restricted eigenvalue property for random Gaussian design]
\label{lem:re}
If $X\in \mathbb{R}^{n\times p}$ has i.i.d. $N(0, \Sigma)$ rows and $m \geq 2$, $n \geq 10$, then the event 
$$\bigg\{\forall v \in \mathbb{R}^p: \frac{\|Xv\|_2}{\sqrt{n}} \geq \frac{1}{4} \|\Sigma^{1/2}v\|_2 - 3\sqrt{\frac{\gmax(\Sigma)n_c}{n}}\|v\|_2 - 6\sqrt{2}\rho(\Gamma)\sqrt{\frac{\gmax(\Sigma)\log p}{n}}\|\Gamma v\|_1\bigg\}$$
holds with probability at least $1 - c_1\exp(-nc_2)$, for some universal constants $c_1, c_2 > 0$.
\end{lemma}

By setting $S = S_{\Gamma\beta^*}$ and applying $k_{S}^{-2}\lesssim \min(d, |S|)$, we obtain the following bounds which are applicable when $\beta^*$ is piecewise constant on the graph $G$. When $\rho(\Gamma)\gtrsim 1$, the second term in \eqref{eq:16} and \eqref{eq:17} should dominate.

\begin{corollary} If $\|\Gamma\beta^*\|_0 = s$, with probability at least $1-c_1\exp(-nc_2) - \frac{2}{m} - e^{-\delta^2/2}$ we have 
\begin{equation}\label{eq:16}
   \|\Sigma^{1/2}(\hat{\beta}-\beta^*)\|_2^2 \lesssim \frac{\sigma^2\gmax(\Sigma)}{\gmin\left(\frac{1}{64}\Sigma+\lambda_2 L\right)}\left(\frac{n_c + \delta^2}{n} + \rho^2(\Gamma)\min(d, s)\frac{s\log p}{n}\right)
\end{equation}
\begin{equation}\label{eq:17}
    \|\hat{\beta}-\beta^*\|_2^2 \lesssim \frac{\sigma^2\gmax(\Sigma)}{\gmin^2\left(\frac{1}{64}\Sigma+\lambda_2L\right)}\left(\frac{n_c + \delta^2}{n} + \rho^2(\Gamma)\min(d, s)\frac{s\log p}{n}\right)
\end{equation}
provided that the RHS of \eqref{eq:16} is smaller than $C\sigma^2$.
\end{corollary}

On the other hand, if we set $S = \emptyset$, we obtain the following bounds that are applicable when $\|\Gamma\beta^*\|_1$ is small. When $\beta^*$ is smoothly varying over $G$ and $\|\Gamma\beta^*\|_0$ is large, these bounds are more helpful in explaining our estimator's good performance. 

\begin{corollary}
With probability at least $1-c_1\exp(-nc_2) - \frac{2}{m} - e^{-\delta^2/2}$,
\begin{equation}\label{eq:l1-sparse-pred}
\|\Sigma^{1/2}(\hat{\beta}-\beta^*)\|_2^2 \lesssim \frac{\sigma^2\gmax(\Sigma)}{\gmin\left(\frac{1}{64}\Sigma+\lambda_2 L\right)}\frac{n_c+\delta^2}{n} + \sigma\rho(\Gamma)\sqrt{\frac{\gmax(\Sigma)\log p}{n}} \|\Gamma\beta^*\|_1
\end{equation}
\begin{equation}\label{eq:l1-sparse-est}
\|\hat{\beta}-\beta^*\|_2^2 \lesssim \frac{\sigma^2\gmax(\Sigma)}{\gmin^2\left(\frac{1}{64}\Sigma+\lambda_2 L\right)}\frac{n_c+\delta^2}{n} + \frac{\sigma\rho(\Gamma)}{\gmin\left(\frac{1}{64}\Sigma+\lambda_2L\right)}\sqrt{\frac{\gmax(\Sigma)\log p}{n}} \|\Gamma\beta^*\|_1
\end{equation}
provided that the RHS of \eqref{eq:l1-sparse-pred} is smaller than $C\sigma^2$.
\end{corollary}

We can also consider the notion that $\Gamma\beta^*$ is $\ell_q$-sparse ($0<q<1$), in the sense that $\sum_{j=1}^m|(\Gamma\beta^*)_j|^q \leq R_q$ (Assumption \eqref{eq:3}). This notion of weak sparsity has been considered in \cite{raskutti2011minimax} (where $\beta^*$ is assumed to lie in an $\ell_q$-ball) and \cite{cai2012optimal} (where, in the context of covariance estimation, the columns of the covariance matrix are assumed to lie in an $\ell_q$-ball). In contrast, \cite{hebiri2011smooth} defines the smoothness of the true signal using $\ell_2$-norm, in the sense that $\sum_{j=1}^m|(\Gamma\beta^*)_j|^2 \leq R_2$
for some $R_2>0$. If there exists an edge with a large signal difference, $R_2$ can be very large. For smaller values of $q$, we can more easily accommodate the occasional large signal jump with a reasonably small $R_q$, which appears in the bound \eqref{cor:lq}. 

By choosing $S$ to trade off the last two terms in the RHS of \eqref{eq:pred-error-main}, we obtain the following bound for the prediction error. The proof is routine and is thus omitted.  

\begin{corollary}\label{cor:5}
With probability at least $1-c_1\exp(-nc_2) - \frac{2}{m} - e^{-\delta^2/2}$, if Assumption \eqref{eq:3} holds for some $q \in (0,1)$, we have
\begin{equation}\label{cor:lq}
\begin{split}
    &\|\Sigma^{1/2}(\hat{\beta}-\beta^*)\|_2^2 \lesssim 
    \frac{\sigma^2\gmax(\Sigma)}{\gmin\left(\frac{1}{64}\Sigma+\lambda_2 L\right)}\frac{n_c+\delta^2}{n} \\
    &+\min\left(\frac{[\sigma\rho(\Gamma)]^{2-q}\left(\frac{\gmax(\Sigma)\log p}{n}\right)^{1-q/2}R_qd^{1-q}}{[\gmin\left(\frac{1}{64}\Sigma+\lambda_2L\right)]^{1-q}},
    \frac{[\sigma\rho(\Gamma)]^{\frac{2}{1+q}}\left(\frac{\gmax(\Sigma)\log p}{n}\right)^{\frac{1}{1+q}}R_q^{\frac{2}{1+q}}}{\left[\gmin\left(\frac{1}{64}\Sigma+\lambda_2L\right)\right]^{\frac{1-q}{1+q}}}\right)
\end{split}
\end{equation}
provided that the RHS of \eqref{cor:lq} is smaller than $C\sigma^2$.
\end{corollary}
\subsection{Discussion of the quantity \texorpdfstring{$\gmin\left(\frac{1}{64}\Sigma + \lambda_2 L\right)$}{a}}\label{theory:mineigen} When $\Gamma = I_p$, our penalty is just the original Elastic Net penalty. In that case, since $\rho(\Gamma) = 1$ and $k_S^{-2}\leq 1$, the corresponding estimator $\hat{\beta}_{\text{EN}}$ satisfies with high probability
\begin{equation}
    \|\Sigma^{1/2}(\hat{\beta}_{\text{EN}}-\beta^*)\|_2^2 \lesssim \frac{\sigma^2\gmax(\Sigma)}{\gmin\left(\frac{1}{64}\Sigma+\lambda_2I_p\right)}\frac{\|\beta^*\|_0\log p}{n}
\end{equation}
Here, it is clear the minimum eigenvalue term is bounded below by $\lambda_2$. When $\Gamma$ is an incidence matrix of a graph, however, $\Gamma$ has a nontrivial kernel and so the behavior of the quantity $\gmin\left(\frac{1}{64}\Sigma+\lambda_2L\right)$ is less clear. 

We conjecture that under reasonable assumptions about $(\Sigma, L)$, $\gmin\left(\frac{1}{64}\Sigma+\lambda_2L\right)$ is bounded below by $c\lambda_2$ for some absolute constant $c$, at least when $\lambda_2$ is in a neighborhood of zero. \textit{We emphasize that the proof of Theorem \ref{theorem:1} makes no assumption about how $\Sigma$ is related to the graph $G$ or its Laplacian $L$}.  When we only have $\gmin\left(\frac{1}{64}\Sigma+\lambda_2L\right) \geq c\lambda_2$, \eqref{eq:16} yields 
\begin{equation}\label{eq:lambda2}
    \|\Sigma^{1/2}(\hat{\beta}-\beta^*)\|_2^2 \lesssim \sigma\sqrt{\gmax(\Sigma)}\|\Gamma\beta^*\|_\infty\left(\frac{n_c+\delta^2}{\rho(\Gamma)\sqrt{n\log p}}+\rho(\Gamma)\min(d,s)s\sqrt{\frac{\log p}{n}}\right)
\end{equation}
but we fail to obtain any theoretical guarantee of consistency in estimation when $\Gamma\beta^*$ is sparse. If we can assume $\gmin\left(\frac{1}{64}\Sigma+\lambda_2L\right) \geq c\sqrt{\lambda_2}$, however, we obtain from \eqref{eq:17} that
\begin{equation}\label{eq:lambda2-est}
    \|\hat{\beta}-\beta^*\|_2^2 \lesssim \sigma\sqrt{\gmax(\Sigma)}\|\Gamma\beta^*\|_\infty\left(\frac{n_c+\delta^2}{\rho(\Gamma)\sqrt{n\log p}}+\rho(\Gamma)\min(d,s)s\sqrt{\frac{\log p}{n}}\right)
\end{equation}

The bounds \eqref{eq:lambda2} and \eqref{eq:lambda2-est} may be more applicable when $\Sigma$ is ill-conditioned and $\gmin(\Sigma)$ cannot be assumed to be bounded away from zero. Unfortunately, characterizing the spectrum of the sum of two symmetric matrices in terms of the spectra of the summands is known to be a difficult problem, and we leave as an open problem the question of identifying a reasonable assumption on $(\Sigma, L)$ (which may both have nontrivial kernels) under which $\gmin\left(\frac{1}{64}\Sigma+\lambda_2L\right) \geq c\lambda_2$ holds. In comparison to our work, Corollary 1 of \cite{hebiri2011smooth} (which assumes fixed design) assumes that its restricted eigenvalue constant $\phi_{\mu_n}$, defined with respect to the matrix $\tilde{X}^T\tilde{X}/n$, may be greater than $\mu_n$ or $\sqrt{\mu_n}$ without further justification; here $\mu_n$ plays a similar role as our $\lambda_2$. In order to prove $\gmin\left(\frac{1}{64}\Sigma+\lambda_2L\right) \geq c\lambda_2$, Lemma 1 of \cite{li2018graph} assumes that, for some absolute constant $c_l > 0$, $\min_{j\in [p]}\sum_{k=1}^p|\Sigma_{jk}| \geq c_l$ and $\max_{j\in [p]}\sum_{k=1}^p|\hat{\Sigma}_{jk}-\Sigma_{jk}| \leq c_l/4$; again, note that $\hat{\Sigma}$ acts as the adjacency matrix of the graph considered in \cite{li2018graph}. Such assumptions may be too restrictive as the same absolute constant $c_l$ is used in both assumptions.

In Section \ref{sec:min_eig}, we provide empirical evidence to show that, in many situations where the true covariance matrix $\Sigma$ reflects the structure of the graph $G$ (that is, features indexed by adjacent or nearby nodes are more correlated) and $\Sigma$ is degenerate, the improvement of the minimum eigenvalue term is significant and can be better than $c\lambda_2$ (or even $c\sqrt{\lambda}_2$).

\subsection{Error bounds for specific types of graphs}\label{sec:theory-specific} In this section, we apply our results to some specific graph structures that are also explored in \cite{hutter2016optimal}. Throughout this section, $s$ denotes $\|\Gamma\beta^*\|_0$ and $R_q$ denotes the bound on $\sum_{j=1}^m |(\Gamma\beta^*)_j|^q$. We only present prediction error bounds here as the estimation error bounds are different only by a factor of $\gmin\left(\frac{1}{64}\Sigma + \lambda_2L\right)$ in the denominator. We mainly assume $\frac{\sigma^2\gmax(\Sigma)}{\gmin(\frac{1}{64}\Sigma + \lambda_2 L}$ is of constant order, but we also specialize the bound \eqref{eq:lambda2} assuming $\gmin\left(\frac{1}{64}\Sigma + \lambda_2 L\right) \gtrsim \lambda_2$ for the case when $\Gamma\beta^*$ is sparse to illustrate the effects of the $\ell_2$ component in our penalty when $\gmin(\Sigma)$ is very small. In that situation, the bounds for the standalone $\ell_1$ penalty provide no control on the errors.\\

\noindent{\bf The 2D grid.} From Proposition 4 of \cite{hutter2016optimal} as well as our lower bound result on $\rho(\Gamma)$ for the 2D grid (proven in the Appendix), we have the following lemma. 

\begin{lemma}\label{lem:2-dim-rho}
    If $\Gamma$ is the incidence matrix of the 2D grid with $p$ vertices, then $$1\lesssim \rho(\Gamma) \lesssim \sqrt{\log p}$$
\end{lemma}

We therefore obtain the following corollary for the 2D grid. 
\begin{corollary}
Let $\Gamma$ be the incidence matrix of the 2D grid with $p$ vertices. With the same choice of $\delta$, $\lambda_1$ and $\lambda_2$ as in Theorem \ref{theorem:1}, with high probability we have
\begin{equation}\label{eq:grid-pred}
\|\Sigma^{1/2}(\hat{\beta}-\beta^*)\|_2^2 \lesssim \frac{\sigma^2\gmax(\Sigma)}{\gmin\left(\frac{1}{64}\Sigma+\lambda_2 L\right)}\left(\frac{1 + \delta^2}{n} + \frac{s(\log p)^2}{n}\right)
\end{equation}
\begin{equation}\label{eq:grid-predl1}
    \|\Sigma^{1/2}(\hat{\beta}-\beta^*)\|_2^2 \lesssim \frac{\sigma^2\gmax(\Sigma)}{\gmin\left(\frac{1}{64}\Sigma+\lambda_2 L\right)}\frac{1 + \delta^2}{n} + \sigma \sqrt{\frac{\gmax(\Sigma)(\log p)^2}{n}}\|\Gamma\beta^*\|_1
\end{equation}
\begin{equation}\label{eq:grid-pred3}
    \|\Sigma^{1/2}(\hat{\beta}-\beta^*)\|_2^2 \lesssim \frac{\sigma^2\gmax(\Sigma)}{\gmin\left(\frac{1}{64}\Sigma+\lambda_2 L\right)}\frac{1 + \delta^2}{n} + \frac{\sigma^{2-q}R_q\left(\frac{\gmax(\Sigma)(\log p)^2}{n}\right)^{1-q/2}}{[\gmin\left(\frac{1}{64}\Sigma+\lambda_2L\right)]^{1-q}}
\end{equation}
provided that the RHS of the bounds above are smaller than $C\sigma^2$. If $\gmin\left(\frac{1}{64}\Sigma+\lambda_2L\right) \gtrsim\lambda_2$, we also have 
\begin{equation}\label{eq:grid-pred2}
    \|\Sigma^{1/2}(\hat{\beta}-\beta^*)\|_2^2 \lesssim \sigma\sqrt{\gmax(\Sigma)}\|\Gamma\beta^*\|_\infty\left(\frac{1+\delta^2}{\sqrt{n\log p}} + \frac{s\log p}{\sqrt{n}}\right)      
\end{equation}
\end{corollary}

The rates obtained in \eqref{eq:grid-pred} and \eqref{eq:grid-pred2} are good if $s$ is of small order relative to $n$. For example, if there is a small island of size $k$-by-$k$ where $\beta^*$ attains a value distinct from its background value outside that island (this situation can correspond to finding abnormal spots on an MRI scan), then \eqref{eq:grid-pred} gives us the rate $\frac{k(\log p)^2}{n}$, provided that $\frac{\sigma^2\gmax(\Sigma)}{\gmin\left(\frac{1}{64}\Sigma + \lambda_2 L\right)}$ is of constant order. This can be compared with the rate obtained by the Lasso estimator, which is $\frac{k^2\log p}{n}$ if the background value outside the island is zero (but it fails to achieve this rate if the background value is nonzero). However, in the situation where the 2D grid can be divided in the middle into a left island and a right island and $\beta^*$ is constant on each of these islands, then $s\asymp \sqrt{p}$ and our rates are meaningful only in the $p \ll n$ setting.\\

\noindent{\bf The $r$-dimensional grid ($r\geq 3$).}
From Proposition 6 of \cite{hutter2016optimal} as well as our lower bound result on $\rho(\Gamma)$ for the $r$-dimensional grid, we can conclude that $\rho(\Gamma)$ in this case is of constant order, assuming $r$ is fixed.
\begin{lemma}\label{lem:r-dim-rho}
    If $\Gamma$ is the incidence matrix of the $r$-dimensional grid with $p$ vertices and $r\geq 3$, then 
    $$c(r) \leq \rho(\Gamma) \leq C(r) $$
    for some constants $c(r), C(r)$ that only depend on $r$. 
\end{lemma}

We obtain the following corollary for the $r$-dimensional grid. 
\begin{corollary}
Let $\Gamma$ be the incidence matrix of the $r$-dimensional grid with $p$ vertices, where $r\geq 3$ is fixed. With the same choice of $\delta$, $\lambda_1$ and $\lambda_2$ as in Theorem \ref{theorem:1}, with high probability we have
\begin{equation}\label{eq:rgrid-pred}
    \|\Sigma^{1/2}(\hat{\beta}-\beta^*)\|_2^2 \lesssim \frac{\sigma^2\gmax(\Sigma)}{\gmin\left(\frac{1}{64}\Sigma+\lambda_2 L\right)}\left(\frac{1 + \delta^2}{n} + \frac{s\log p}{n}\right)
\end{equation} 
\begin{equation}
    \|\Sigma^{1/2}(\hat{\beta}-\beta^*)\|_2^2 \lesssim \frac{\sigma^2\gmax(\Sigma)}{\gmin\left(\frac{1}{64}\Sigma+\lambda_2 L\right)}\frac{1 + \delta^2}{n}+ \sigma\sqrt{\frac{\gmax(\Sigma)\log p}{n}}\|\Gamma\beta^*\|_1
\end{equation}
\begin{equation}\label{eq:3D-Rq}
    \|\Sigma^{1/2}(\hat{\beta}-\beta^*)\|_2^2 \lesssim \frac{\sigma^2\gmax(\Sigma)}{\gmin\left(\frac{1}{64}\Sigma+\lambda_2 L\right)}\frac{1 + \delta^2}{n} + \frac{\sigma^{2-q}R_q\left(\frac{\gmax(\Sigma)\log p}{n}\right)^{1-q/2}}{[\gmin\left(\frac{1}{64}\Sigma+\lambda_2L\right)]^{1-q}}
\end{equation}
provided the RHS of the bounds above are smaller than $C\sigma^2$.
If $\gmin\left(\frac{1}{64}\Sigma+\lambda_2L\right) \gtrsim\lambda_2$, we also have
\begin{equation}
        \|\Sigma^{1/2}(\hat{\beta}-\beta^*)\|_2^2 \lesssim\sigma\sqrt{\gmax(\Sigma)}\|\Gamma\beta^*\|_\infty\left(\frac{1+\delta^2}{\sqrt{n\log p}} + s\sqrt{\frac{\log p}{n}}\right)
\end{equation}
\end{corollary}

If we consider $r = 3$ and there is a small island of size $k$-by-$k$-by-$k$ where $\beta^*$ attains a value distinct from its background value outside that island, then \eqref{eq:rgrid-pred} gives us the rate $\frac{k^2\log p}{n}$, whereas the Lasso gives us the rate $\frac{k^3\log p}{n}$ if we further assume the background value is zero. This suggests that if the signal is both sparse and smooth over the graph, in some situations using our estimator is preferable to using the Lasso. More generally, if the island is not cubic but rather has an arbitrary shape, $\|\Gamma\beta^*\|_0$ should be the island's surface area, whereas $\|\beta^*\|_0$ should be the island's volume. \\

\noindent{\bf The complete graph.} As previously mentioned, we can consider regularization with the complete graph when there is no prior structural information available. 

\begin{lemma}
If $\Gamma$ is the incidence matrix of the complete graph with $p$ vertices, $\rho(\Gamma)\asymp \frac{1}{p}$.
\end{lemma}
\begin{proof}
    In Proposition 10 of \cite{hutter2016optimal}, replace any `$\leq$' sign with `$=$'.
\end{proof}

If we replace the term $\min(d,s)$ by $p$ (since $d\asymp p$), we obtain the following corollary:
\begin{corollary}
Let $\Gamma$ be the incidence matrix of the complete graph with $p$ vertices. With the same choice of $\delta$, $\lambda_1$ and $\lambda_2$ as in Theorem \ref{theorem:1}, with high probability we have
\begin{equation}\label{eq:21}
\|\Sigma^{1/2}(\hat{\beta}-\beta^*)\|_2^2 \lesssim \frac{\sigma^2\gmax(\Sigma)}{\gmin\left(\frac{1}{64}\Sigma+\lambda_2 L\right)}\left(\frac{1 + \delta^2}{n} + \frac{s\log p}{pn}\right)
\end{equation}
\begin{equation}
    \|\Sigma^{1/2}(\hat{\beta}-\beta^*)\|_2^2 \lesssim \frac{\sigma^2\gmax(\Sigma)}{\gmin\left(\frac{1}{64}\Sigma+\lambda_2 L\right)}\frac{1 + \delta^2}{n} + \frac{\sigma}{p}\sqrt{\frac{\gmax(\Sigma)\log p}{n}}\|\Gamma\beta^*\|_1
\end{equation}
\begin{equation}\label{eq:Rq-complete}
    \|\Sigma^{1/2}(\hat{\beta}-\beta^*)\|_2^2 \lesssim \frac{\sigma^2\gmax(\Sigma)}{\gmin\left(\frac{1}{64}\Sigma+\lambda_2 L\right)}\frac{1 + \delta^2}{n} + \frac{\sigma^{2-q}\frac{R_q}{p}\left(\frac{\gmax(\Sigma)\log p}{n}\right)^{1-q/2}}{[\gmin\left(\frac{1}{64}\Sigma+\lambda_2L\right)]^{1-q}}
\end{equation}
provided that the RHS of the above bounds are smaller than $C\sigma^2$. If $\gmin\left(\frac{1}{64}\Sigma+\lambda_2L\right) \gtrsim\lambda_2$, we also have 
\begin{equation}\label{eq:23}
    \|\Sigma^{1/2}(\hat{\beta}-\beta^*)\|_2^2 \lesssim \sigma\sqrt{\gmax(\Sigma)}\|\Gamma\beta^*\|_\infty\left(\frac{p(1+\delta^2)}{\sqrt{n\log p}} + s\sqrt{\frac{\log p}{n}}\right)
\end{equation}
\end{corollary}
In the case when the signal takes $k \ll p$ different values, with $k-1$ of those attained on small islands of size $l \ll p$, $s$ is of order $klp$, and \eqref{eq:21} yields the rate $\frac{kl\log p}{n}$, provided that $\frac{\sigma^2\gmax(\Sigma)}{\gmin\left(\frac{1}{64}\Sigma+\lambda_2 L\right)}$ is of constant order. This is the same as the rate we obtain for the Lasso if the complement of the small islands has value zero. However, if there are two large components with two different values, $s$ is of order $p^2$ and so \eqref{eq:21} only guarantees some control when $p\ll n$. If $\gmin\left(\frac{1}{64}\Sigma + \lambda_2 L\right)$ is of order $\lambda_2$, then \eqref{eq:23} 
only gives us a meaningful bound when $p\ll \sqrt{n}$, provided that $\|\Gamma\beta^*\|_\infty$ is of constant order. \\

\noindent{\bf The star graph.} Here, we consider the graph with $p$ nodes and with one center node connected to $p-1$ leaves. A similar penalty has been considered by \cite{ollier2017regression} to model stratified data, and this penalty is useful particularly when most outer nodes share the same value as the central node. 

\begin{lemma}
    If $\Gamma$ is the incidence matrix of the star graph with $p$ vertices, then $\rho(\Gamma) \asymp 1$. 
\end{lemma}
\begin{proof}
From Proposition 12 in \cite{hutter2016optimal}, any column $s_j$ of $\Gamma^\dagger$ has $\|s_j\|_2^2 = 1-\frac{1}{p}$.
\end{proof} 

\begin{corollary}
Let $\Gamma$ be the incidence matrix of the star graph with $p$ vertices. With the same choice of $\delta$, $\lambda_1$ and $\lambda_2$ as in Theorem \ref{theorem:1}, with high probability we have
\begin{equation}\label{eq:24}
    \|\Sigma^{1/2}(\hat{\beta}-\beta^*)\|_2^2 \lesssim \frac{\sigma^2\gmax(\Sigma)}{\gmin\left(\frac{1}{64}\Sigma+\lambda_2 L\right)}\left(\frac{1 + \delta^2}{n} + \frac{s^2\log p}{n}\right)
\end{equation}\label{eq:25}
\begin{equation}
    \|\Sigma^{1/2}(\hat{\beta}-\beta^*)\|_2^2 \lesssim \frac{\sigma^2\gmax(\Sigma)}{\gmin\left(\frac{1}{64}\Sigma+\lambda_2 L\right)}\frac{1 + \delta^2}{n} + \sigma\sqrt{\frac{\gmax(\Sigma)\log p}{n}}\|\Gamma\beta^*\|_1
\end{equation}
\begin{equation}\label{eq:Rq-star}
    \|\Sigma^{1/2}(\hat{\beta}-\beta^*)\|_2^2 \lesssim \frac{\sigma^2\gmax(\Sigma)}{\gmin\left(\frac{1}{64}\Sigma+\lambda_2 L\right)}\frac{1 + \delta^2}{n} + \frac{\left(\frac{\sigma^2\gmax(\Sigma)R_q^2\log p}{n}\right)^{\frac{1}{1+q}}}{[\gmin\left(\frac{1}{64}\Sigma+\lambda_2L\right)]^{\frac{1-q}{1+q}}}
\end{equation}
provided that the RHS of the above bounds are smaller than $C\sigma^2$. If $\gmin\left(\frac{1}{64}\Sigma+\lambda_2L\right) \gtrsim\lambda_2$, we also have 
\begin{equation}\label{eq:26}
    \|\Sigma^{1/2}(\hat{\beta}-\beta^*)\|_2^2 \lesssim \sigma\sqrt{\gmax(\Sigma)}\|\Gamma\beta^*\|_\infty\left(\frac{1+\delta^2}{\sqrt{n\log p}} + s^2\sqrt{\frac{\log p}{n}}\right)
\end{equation} 
\end{corollary}

For the star graph, we obtain meaningful bounds for the prediction error, even in the high-dimensional setting where $p\gg n$.\\

\noindent{\bf The chain graph.} When $\Gamma$ is the $p$-vertex chain graph (1D grid graph), $\rho(\Gamma) = \sqrt{p}$ and Theorem \ref{theorem:1} does not yield an error bound that is meaningful in the $p \ll n$ setting. We modify the proof of Theorem \ref{theorem:1} using an idea in Theorem 6 of \cite{wang2015trend} to obtain the following bound when $\|\Gamma\beta^*\|_1$ is small. 

\begin{theorem}\label{theorem:chain}
Let $\Gamma$ be the incidence matrix of the $p$-vertex chain graph, and fix $\delta > 0$. With an appropriate choice of $\lambda_1$ and $\lambda_2 \leq \frac{\lambda_1}{8\|\Gamma\beta^*\|_\infty}$, with high probability we have 
\begin{equation}\label{eq:chain-pred}
    \|\Sigma^{1/2}(\hat{\beta}-\beta^*)\|_2^2 \lesssim \frac{\sigma^2\gmax(\Sigma)}{\gmin(\frac{1}{64}\Sigma+\lambda_2L)}\frac{1+\delta^2}{n} + \frac{(\sigma^2\gmax(\Sigma)\|\Gamma\beta^*\|_1)^{2/3}}{\gmin^{1/3}\left(\frac{1}{64}\Sigma+\lambda_2L\right)}\sqrt[3]{\frac{p\log p}{n^2}}
\end{equation}
provided that the bound above is smaller than $C\sigma^2$.    
\end{theorem}

The bound above is meaningful when $n \gg \sqrt{p \log p}$ and thus sufficient to justify the use of our estimator when $\Gamma$ is the chain graph. Optimal error bounds under the assumption of hard sparsity on $\Gamma\beta^*$ are available in the literature if $X$ is identity (see for example \cite{ortelli2021prediction} and \cite{guntuboyina2020adaptive}). However, such bounds are often derived under a ``minimum length'' condition, which requires that the distances between jumps for the true signal are roughly of the same order. The bound \eqref{eq:chain-pred}, on the other hand, requires minimal assumptions. We leave open for future work the analysis of our estimator \eqref{eq:4} under the assumption of hard sparsity on $\Gamma\beta^*$.

%% file: computation.tex
In this section, we describe our coordinate descent procedure to compute the estimator \eqref{eq:4}. For convenience, we will work with the following definition of $\hat{\beta}$ where we replace the loss $\frac{1}{n}\|Y-X\beta\|_2^2$ in \eqref{eq:4} by $\frac{1}{2}\|Y-X\beta\|_2^2$. Note that this simply corresponds to a different scaling of $\lambda_1$ and $\lambda_2$.
\begin{equation}\label{eq:28}
    \hat{\beta} := \arg\min_{\beta\in\mathbb{R}^p}\frac{1}{2}\|Y-X\beta\|_2^2 + \lambda_1\|\Gamma\beta\|_1+\lambda_2\|\Gamma\beta\|_2^2
\end{equation}

Again, let $\tilde{Y}:=\begin{pmatrix} Y \\ 0\end{pmatrix} \in \mathbb{R}^{n+m}$, $\tilde{X} := \begin{pmatrix}X \\ \sqrt{2\lambda_2}\Gamma\end{pmatrix}\in \mathbb{R}^{(n+m)\times p}$ so that we can write 
\begin{equation}\label{eq:29}
    \hat{\beta}:= \arg\min_{\beta\in\mathbb{R}^p}\frac{1}{2}\|\tilde{Y}-\tilde{X}\beta\|_2^2 + \lambda_1\|\Gamma\beta\|_1
\end{equation}
 
If we fix $\lambda_2$, the solution path in terms of $\lambda_1$ for \eqref{eq:29} as well as its dual objective is piecewise linear, and a path-finding algorithm for the dual objective yielding the entire solution path in terms of $\lambda_1$ has been proposed in \cite{tibshirani2011solution}. However, for the purpose of selecting tuning parameters, this is of limited usefulness since $\lambda_2$ needs to be fixed. The solution path in terms of $\lambda_2$ is not piecewise linear, and so we cannot use a LARS-like algorithm to get the entire path in terms of both $(\lambda_1, \lambda_2)$. Also, as mentioned in \cite{tibshirani2011solution}, the set of knots in the solution path becomes very large as the problem size increases, and at each knot we must solve a large least squares problem (especially at the regularized end of the path, which is typically the region of interest in this paper) in order to compute the whole path. If we want to compute \eqref{eq:29} for a small set of candidate $(\lambda_1, \lambda_2)$ values, then the path algorithm in \cite{tibshirani2011solution} is unlikely to be the most efficient. 

\subsection{Coordinate descent on the dual objective} Our coordinate descent algorithm builds upon the dual problem derived in \cite{tibshirani2011solution} for the Generalized Lasso. 
\cite{tibshirani2011solution} suggests, without explicit derivations, that we can use coordinate descent on the dual problem to compute the solution of \eqref{eq:29} for a fixed value of $(\lambda_1, \lambda_2)$. Coordinate descent cannot be directly applied to the primal objective \eqref{eq:29} as the $\ell_1$-penalty here is not separable in terms of $\beta$; in such a situation, coordinate descent does not necessarily converge. However, the dual objective \eqref{eq:dual} has a non-smooth component that is separable, and thus convergence is guaranteed (since conditions (A1), (B1)-(B3) and (C2) from \cite{tseng2001convergence} hold). For completeness, we fully derive this coordinate descent algorithm on the dual and provide experiments to convince the reader that our estimator can be efficiently computed. 

Define $\check{Y} := \tilde{X}\tilde{X}^\dagger\tilde{Y}\in \mathbb{R}^{m+n}$, $\check{\Gamma} := \Gamma\tilde{X}^\dagger \in \mathbb{R}^{m\times (m+n)}$. From Equation (36) of \cite{tibshirani2011solution}, the dual problem is:
\begin{equation}\label{eq:dual_objective}
    \hat{u} = \arg\min_{u\in\mathbb{R}^m}\frac{1}{2}\|\check{Y}-\check{\Gamma}^Tu\|_2^2 \quad\text{subject to } \|u\|_\infty \leq \lambda_1, \Gamma^Tu\in \text{row}(\tilde{X})
\end{equation}
and the primal-dual relation, as in Equation (37) of \cite{tibshirani2011solution}, is:
\begin{equation}
    \hat{\beta} = \tilde{X}^\dagger(\check{Y} - \check{\Gamma}^T\hat{u}) + z
\end{equation}
where $z \in \text{ker}(\tilde{X})$. In most situations, the augmented matrix $\tilde{X} := \begin{pmatrix}X \\ \sqrt{2\lambda_2}\Gamma\end{pmatrix}$ has a trivial kernel, in which case $\text{row}(\tilde{X}) = \mathbb{R}^p$ and we can ignore $z$ as well as the constraint $\Gamma^Tu\in \text{row}(\tilde{X})$. Now if we let $Q := \check{\Gamma}\check{\Gamma}^T \in \mathbb{R}^{m\times m}$ and $b := \check{\Gamma}\check{Y}\in \mathbb{R}^m$, then we can write the dual objective as: 
\begin{equation}\label{eq:dual}
    \hat{u} = \arg\min_{u\in \mathbb{R}^m}\frac{1}{2}u^TQu - b^Tu \quad \text{subject to } \|u\|_\infty \leq \lambda_1
\end{equation}

We denote the projection map from $\mathbb{R}$ onto $[-\lambda, \lambda]$ by $T_\lambda (\cdot)$: 
\begin{equation}
T_{\lambda}(x):=\left\{
\begin{array}{rl}
\lambda &\text{if}\ x > \lambda\\
x &\text{if}\ -\lambda \leq x \leq \lambda\\
-\lambda &\text{if}\ x < -\lambda
\end{array}
\right.    
\end{equation}

Our coordinate descent algorithm is presented below. \\
\begin{algorithm}[H]\label{alg:1}
\DontPrintSemicolon
  \KwInput{$\lambda_1, \lambda_2, \Gamma, Y, X$, tolerance $\epsilon$}
  \KwOutput{$\hat{\beta}$ as defined in \eqref{eq:28}}
  Compute $ Q = \check{\Gamma}\check{\Gamma}^T = (\Gamma\tilde{X}^\dagger)(\Gamma\tilde{X}^\dagger)^T$ and  $b = \check{\Gamma}\check{Y} = \Gamma\tilde{X}^\dagger\tilde{Y}$\\
  Initialize $\hat{u}_i^{(0)} \leftarrow 0 $ for all $i\in [m]$ \Comment{0 is an arbitrary choice in $[-\lambda_1, \lambda_1]$}\\
  \While{$\|\hat{u}^{(k)} - \hat{u}^{(k-1)}\|_2 > \epsilon$ }
   {
   		$\hat{u}_i^{(k+1)} \leftarrow T_{\lambda_1}\left(\frac{b_i-\sum_{j<i}Q_{ij}\hat{u}_j^{(k+1)} - \sum_{j>i}Q_{ij}\hat{u}_j^{(k)}}{Q_{ii}}\right)$
   }
   Compute $\hat{\beta} \leftarrow \tilde{X}^\dagger(\check{Y} - \check{\Gamma}^T \hat{u})$\\
   Return $\hat{\beta}$
   
\caption{Coordinate descent on the dual objective}
\end{algorithm}

For general GLM loss functions, we can also derive the dual problem with a separable non-smooth constraint; however, we may not be able to write the coordinate descent updates in closed form (we can only do so in Algorithm \ref{alg:1} because the dual objective \eqref{eq:dual} is quadratic). In this case, we can use \textit{coordinate proximal gradient descent}, in which we apply the projection operator to the gradient descent update for each coordinate.

In the Appendix, we also provide an alternative algorithm to compute \eqref{eq:28}, based on the interior point method applied to the dual objective (as in \cite{kim2009ell_1}). This algorithm will be denoted as IP in the following section.

\subsection{Runtime comparisons}\label{subsec:runtime experiments}
We compare the runtimes for computing the estimator \eqref{eq:4} using  Algorithm \ref{alg:1} (CD), IP, ADMM, and the Embedded Conic Solver (ECOS) from \cite{domahidi2013ecos} applied to the primal objective. ECOS is a generic solver for second-order cone programs (SOCP) that performs well for small or medium-sized problems. We use the highly optimized ECOS implementation in the Python package CVXPY to serve as a benchmark for comparing the runtimes of our algorithms. Figure~\ref{fig:runtime} shows the growth of empirical runtimes as $n$ or $p$ increases for signals over the chain graph (where $m = p-1$) with $\|\Gamma\beta^*\|_\infty = 0.3$ fixed; here, the hyperparameters $\lambda_1, \lambda_2$ are chosen according to our theory so as to satisfy $\lambda_2 = \frac{\lambda_1}{8\|\Gamma\beta^*\|_\infty}$. 

As we can see from Figure \ref{fig:runtime}, our coordinate descent algorithm scales well as $n$ and $p$ increase, and its runtime does not exceed 10 seconds if $n$ and $p$ are both smaller than 1,000. More generally, when $\lambda_2$ is not too close to zero, the matrix $Q = \check{\Gamma}\check{\Gamma}^T$ is not ill-conditioned and our coordinate descent algorithm performs quite well. We note that this is the setting where our estimator \eqref{eq:4} should be preferred over the Generalized Lasso estimator $\hat{\beta}_{\text{GL}}$ in \eqref{eq:8}, whose accuracy is impeded by the ill-conditioned nature of the matrix $Q$ when $\lambda_2$ is equal to zero. While our estimator requires a two-dimensional grid search to choose $(\lambda_1, \lambda_2)$, Algorithm \ref{alg:1} can significantly reduce the time it takes to perform hyperparameter tuning, even for large-scale problems where $p$ and $n$ are both in the thousands. Note that when both $n$ and $p$ are not too large, the generic SOCP solver ECOS can also be competitive.

As for our interior point method, the main computational bottleneck is the cost of solving a linear equation involving the Hessian matrix; in other words, we need to solve the problem $Ax=b$ for each iteration, where $A$ is an $m$-by-$m$ matrix. Solving it requires $O(m^3)$ operations, and thus IP can do well only if the number of iterations required is small. Figure \ref{fig:runtime}(b) shows that in the case of the chain graph, when we fix $n$ and increase $p$, IP still performs better than the generic solver ECOS and scales well with $p$. 

\begin{figure}[H]
    \centering
\begin{subfigure}{0.49\textwidth}
    \centering
\includegraphics[width = \textwidth]{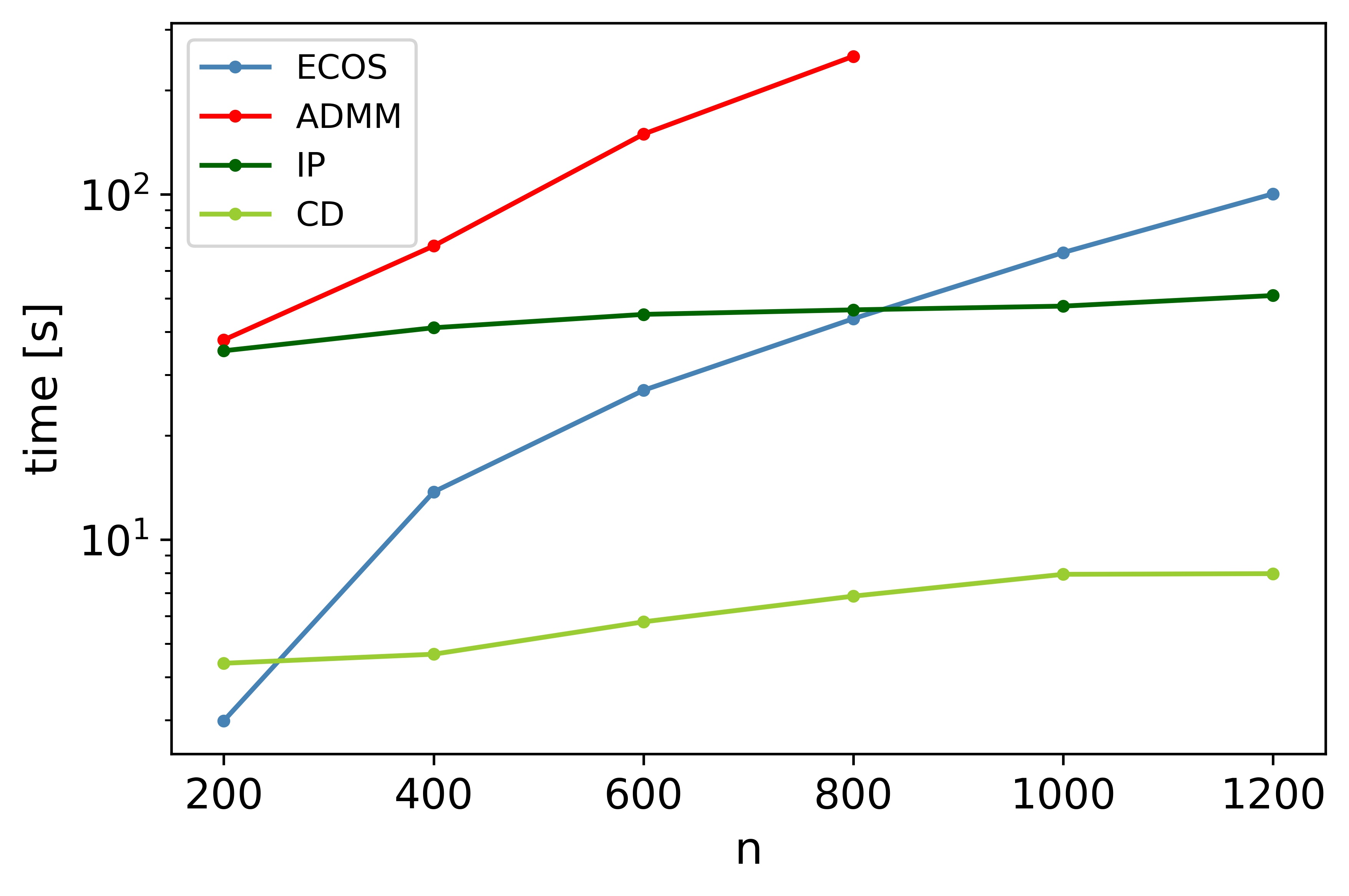} 
\caption{$p = 2000$}\label{fig:runtime1}
\end{subfigure}
    \hfill
\begin{subfigure}{0.49\textwidth}
    \centering
\includegraphics[width = \textwidth]{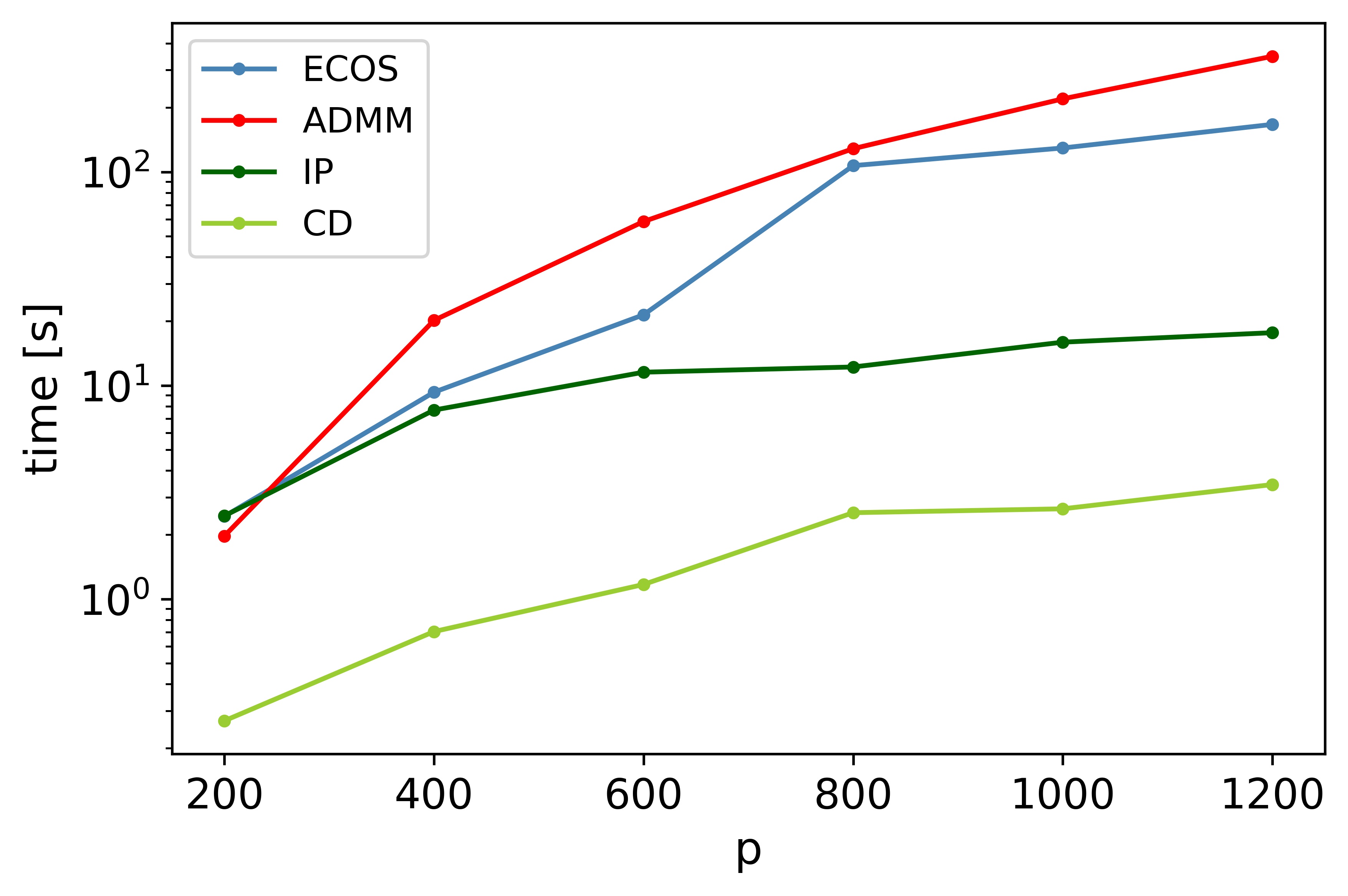} 
\caption{$n = 2000$}\label{fig:runtime2}
\end{subfigure}
\caption{\textit{Runtimes of different algorithms (reported on the log scale) when (a) $p$ is fixed but $n$ increases, or (b) $n$ is fixed but $p$ increases. The tolerance levels for IP, CD, and ECOS are set at $10^{-4}$. The tolerance level for ADMM is $10^{-3}$. Signals are defined on a 1D chain graph with $p$ vertices. In both situations, CD has the best runtime scaling, and IP scales better than ECOS.}}
    \label{fig:runtime}
\end{figure}
We also examine the runtimes for the 2D grid as well as the star graph when $n$ is fixed but $p$ increases. For these graphs, IP no longer scales well with $p$ whereas CD still has the best scaling, and ECOS is also competitive for small problem sizes. 

\begin{figure}[H]
    \centering
\begin{subfigure}{0.47\textwidth}
    \centering
\includegraphics[width = \textwidth]{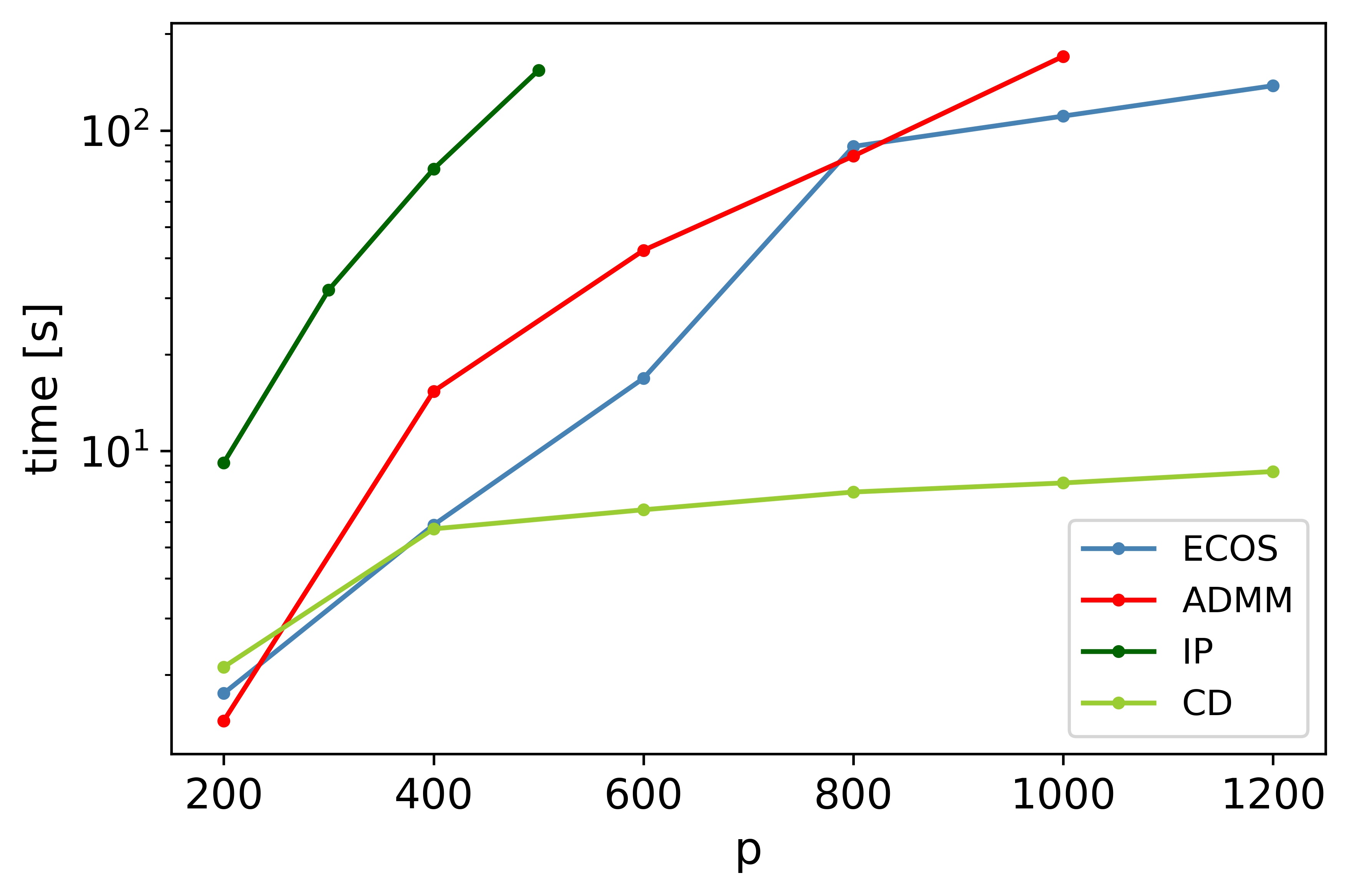} 
\caption{$n = 2000$, 2D grid}\label{fig:2D_runtime}
\end{subfigure}
    \hfill
\begin{subfigure}{0.49\textwidth}
    \centering
\includegraphics[width = \textwidth]{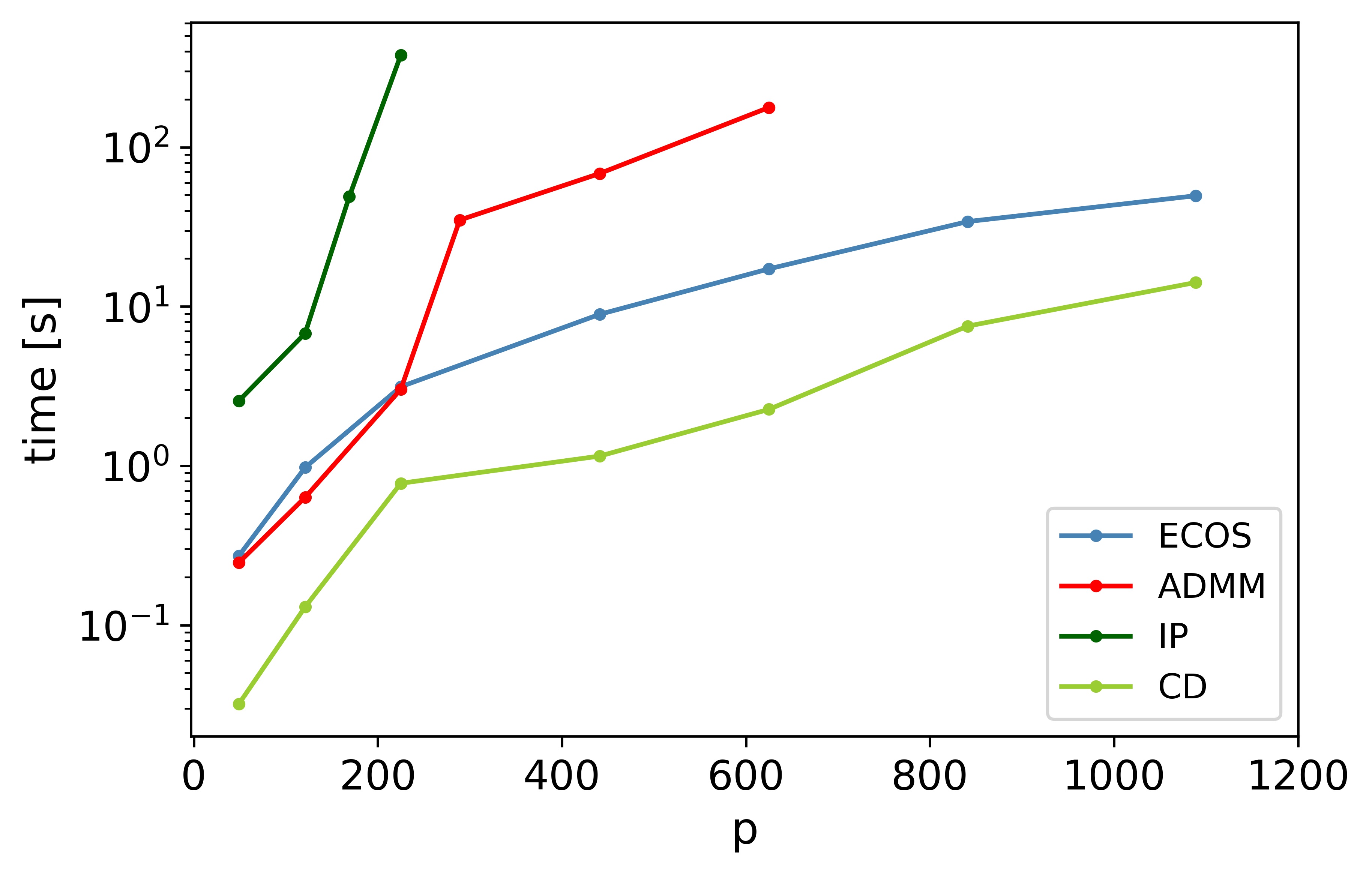} 
\caption{$n = 2000$, star graph}\label{fig:star_runtime}
\end{subfigure}
\caption{\textit{Runtimes of different algorithms (reported on the log scale) when $n$ is fixed but $p$ increases. (a) Signals are defined on a $p$-vertex 2D grid graph ($m = 2p - 2\sqrt{p}$) with $\|\Gamma\beta^*\|_\infty = 0.66$. (b) Signals are defined on a $p$-vertex star graph ($m = p-1$) with $\|\Gamma\beta^*\|_\infty = 0.5$. The tolerance levels for IP, CD, and ECOS are set at $10^{-4}$. The tolerance level for ADMM is $10^{-3}$. As before, $(\lambda_1, \lambda_2)$ are chosen according to theory. In both situations, CD has the best runtime scaling.}}
    \label{fig:other_graph_runtime}
\end{figure}

%% file: experiments.tex
In this section, we present the empirical performance of our penalty $\lambda_1\|\Gamma\beta\|_1 + \lambda_2\|\Gamma\beta\|_2^2$ and compare with some existing penalized M-estimators in the literature, under several synthetic settings where we vary the true signal structure and graph topology. Particularly, we focus on the case where $\beta^*$ is not sparse but aligns with the graph $G$. The design matrix is also allowed to be correlated in a way such that two vertices have more correlated feature vectors if they are adjacent or nearby on the graph $G$. Such a covariance structure is natural for node-indexed feature vectors and is in line with the notion of \textit{network cohesion} discussed in Section \ref{sec:intro}. We list the methods to which we compare our estimator below.\\

\noindent{\bf Graph-independent methods} that do not take into account the graph provided. These methods usually do not perform well in the setting we describe above, and they mainly serve as benchmarks for comparison.
\begin{enumerate}
    \item The {\it ordinary least squares} (OLS) estimator, which is a standard method in the setting when $p < n$ and the underlying signal is dense. It often does not perform well when we are in the high-dimensional setting ($p > n$) or the design is highly correlated and $\gmin(\Sigma)$ is close to zero. 
    
    \item The {\it Lasso} (L) penalty  $\lambda_L\|\beta\|_1$ from \cite{tibshirani1996regression}, which can perform well in the $p\gg n$ setting if the true signal is known to be sparse. In the $p>n$ case, however, it has been shown to select at most $n$ variables before it saturates. As discussed in \cite{zou2005regularization}, the Lasso lacks the ability to select groups of correlated variables, and it is empirically observed to suffer from unstable selections in the presence of high correlation between features.
    \item The {\it Elastic Net} (EN) penalty $\lambda_L\|\beta\|_1 + \lambda_E\|\beta\|_2^2$, which was developed in \cite{zou2005regularization} to deal with highly correlated predictors. The Elastic Net tends to encourage strongly correlated predictors to be in or out of the model together while also preserving sparsity of representation like the Lasso. It is a suitable candidate in our setting due to our assumption of highly correlated design. 
    \end{enumerate}
\noindent{\bf  Graph-based methods} that utilizes information from the given graph $G$ (except for possibly the GTV method). We have described these methods in Section \ref{sec:related}. 
\begin{enumerate}
 \setcounter{enumi}{3} 
\item The {\it Fused Lasso} (FL) penalty $\lambda_1\|\Gamma\beta\|_1 + \lambda_L\|\beta\|_1$ proposed in \cite{tibshirani2005sparsity} encourages the resulting estimate to be both sparse and piecewise constant with respect to $G$. This penalty may be suitable if we believe the true signal is sparse and also forms clusters on $G$ (that is, in each cluster the true signal attains a single value). When the true signal is not sparse, the tuning parameter $\lambda_L$ is often set to zero if we use cross-validation (CV) for hyperparameter selection, and FL degenerates into our GEN penalty with $\lambda_2 = 0$. 

\item The \textit{Smooth Lasso} (SL) penalty $ \lambda_2\|\Gamma\beta\|_2^2 + \lambda_L\|\beta\|_1 $ in  \cite{hebiri2011smooth} results in an estimate that is smooth, in the sense that $\|\Gamma\hat{\beta}_{\text{SL}}\|_\infty$ is small. It is useful when $\beta^*$ is sparse and we also believe $\|\Gamma\beta^*\|_\infty$ is small. When the true signal is not sparse and we use CV for hyperparameter selection, $\lambda_L$ for SL is often set to zero, in which case SL also degenerates into our GEN penalty with $\lambda_1 = 0$. 

\item The {\it Graph Total Variation} (GTV) penalty $\lambda_1\|\hat{\Gamma}\beta\|_1 + \lambda_2\|\hat{\Gamma}\beta\|_2^2 + \lambda_L\|\beta\|_1$ in \cite{li2018graph} estimates $\Sigma$ with some covariance estimator $\hat{\Sigma}$ and then treats $\hat{\Sigma}$ as the weighted adjacency matrix of some graph $\hat{G}$ with incidence matrix $\hat{\Gamma}$. In our experiments, as suggested by \cite{li2018graph}, the estimator $\hat{\Sigma}$ is obtained by hard-thresholding the sample covariance matrix (see \cite{bickel2008covariance} for details). This choice of $\hat{\Sigma}$ means that we also need to tune the covariance threshold in addition to the 3 hyperparameters that appear in the GTV penalty. In general, however, $\hat{\Sigma}$ can be any covariance estimator and can also incorporate side information such as the graph $G$ provided in our setting. 

\item We also denote by \textit{GTV-oracle} the GTV penalty based on using the unobserved covariance matrix $\Sigma$ (rather than $\hat{\Sigma}$) to construct the corresponding incidence matrix $\hat{\Gamma}_{\text{oracle}}$. Using the true covariance matrix should eliminate any error from covariance estimation. However, if all entries of $\Sigma$ are nonzero, computation of the GTV-oracle estimator can be especially challenging, since the graph used in the GTV penalty here is a weighted complete graph. 
\end{enumerate}

\subsection{Experiments on synthetic data}\label{sec:synthesis}

We repeatedly generate training and testing data from the model $y = X\beta^* + \epsilon$, where the rows of $X$ are generated i.i.d. from $N(0,\Sigma)$ and independent of $\epsilon$ which is generated from $N(0, \sigma^2I_n)$. Hyperparameter selection via CV is performed using a separate validation data set. We report the estimation error $\|\hat{\beta}-\beta^*\|_2$, as well as the prediction error $\frac{1}{n}\|X_{\text{test}}(\hat{\beta} - \beta^*)\|_2^2$ computed using the testing data. 

\subsubsection{Choices of $\Sigma$ and the graph $G$} \label{sec:params_choice}

We consider the chain graph, the 2D grid graph and the barbell graph in our experiments. The first two graphs allow for easier visualization of the true and estimated signals defined on them. The barbell graph, which consists of two non-overlapping cliques connected by a single path that has an endpoint in each clique, allows us to test the performance of our method on a denser graph with a less homogenous degree distribution. 

As previously mentioned, $\Sigma$ is constructed so that nearby nodes have more correlated feature vectors. For the chain graph, we use the Toeplitz covariance structure with $\Sigma_{ij} = \rho^{|i-j|}$ where, if not stated otherwise, we typically choose $\rho = 0.5$ (moderate correlation). For the 2D grid and barbell graph, we construct $\Sigma$ by inverting the matrix $L + 0.5I_p$ (recall that $L$ denotes the Laplacian of the graph $G$) and then normalize $\Sigma$ so that all covariates have unit variance. The resulting covariance matrices obtained from this process are illustrated in Figure \ref{fig:covariance}.

\begin{figure}[H]
    \centering
\begin{subfigure}[b]{0.45\textwidth}
    \centering
 \includegraphics[width=\textwidth]{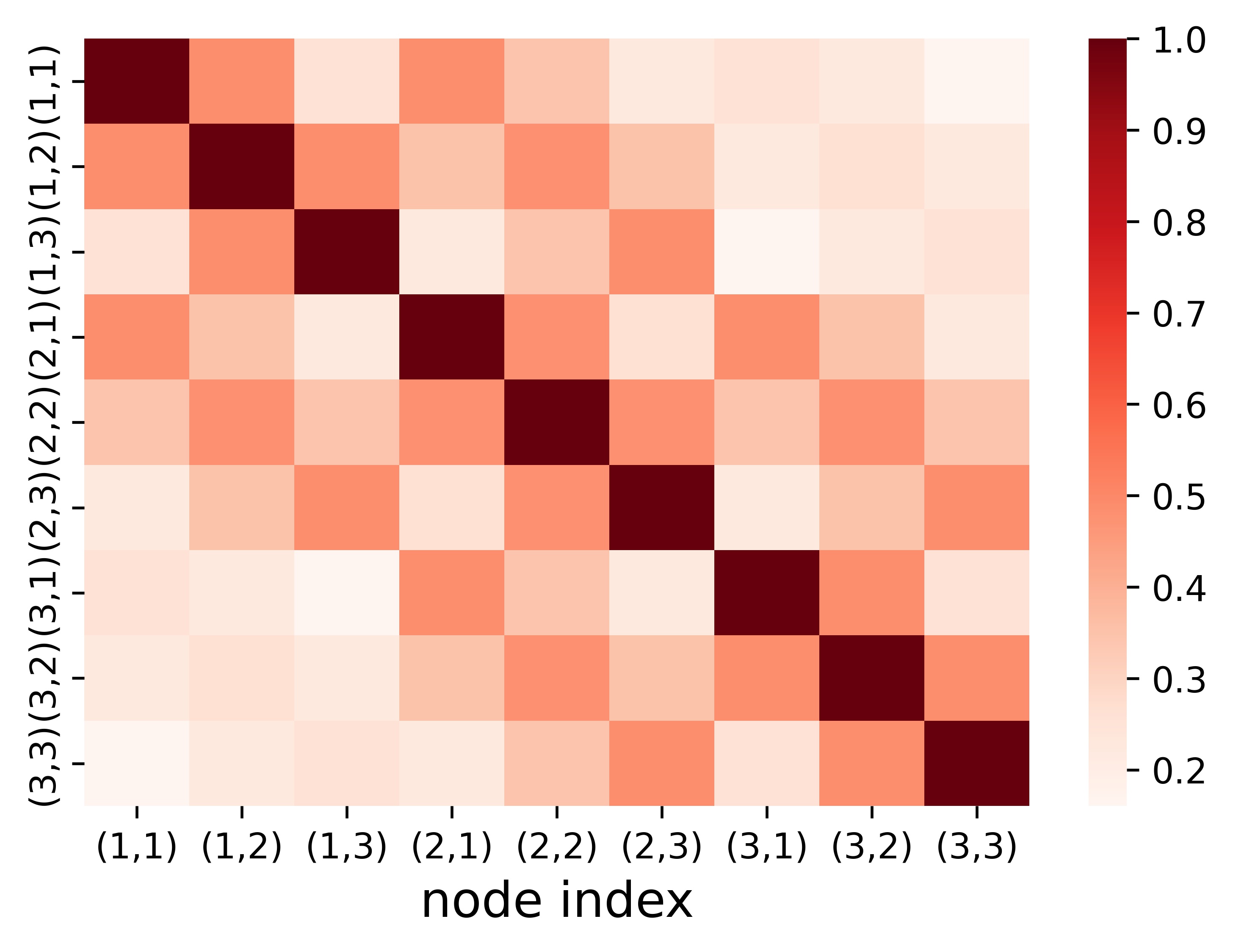}   
\end{subfigure}
   \hfill
\begin{subfigure}[b]{0.45\textwidth}
    \centering
 \includegraphics[width=\textwidth]{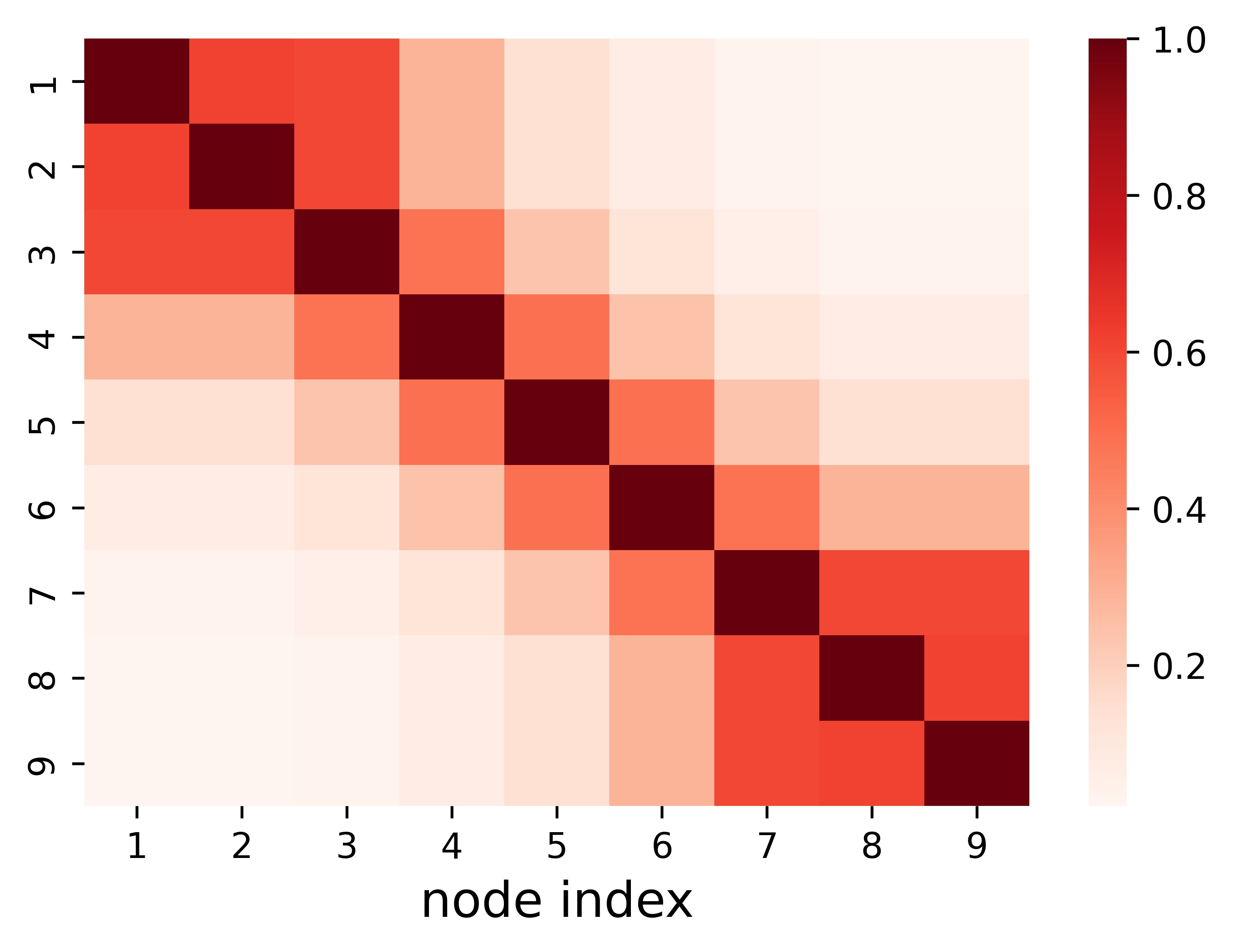}  
\end{subfigure}
    \caption{\textit{Left: the covariance matrix obtained for a 2D grid graph with $p = 3 \times 3$ vertices. Right: the covariance matrix obtained for a barbell graph with two cliques $\{1,2,3\}$ and $\{7,8,9\}$ connected by the path $\{3,4,5,6,7\}$. Note that correlation is higher for adjacent or nearby vertices.}}
    \label{fig:covariance}
\end{figure}

\subsubsection{Hyperparameter selection and tuning time} \label{sec: tuning}

We select hyperparameters based on 5-fold CV using a fine grid search, where each hyperparameter is chosen from a list of at least 20 values. The scorer for CV is the negative mean squared error (MSE) $-\frac{1}{n}\|Y-X\hat{\beta}\|_2^2$, which tends to select for hyperparameters with better prediction performance. 

Hyperparameter tuning for GEN is computationally manageable. When $G$ is a chain graph, the tuning time for GEN is comparable to that of other methods with two hyperparameters, namely the Elastic Net, the Fused Lasso and the Smooth Lasso. If we disregard the covariance thresholding parameter, the GTV penalty still involves three hyperparameters, and the graph $\hat{\Gamma}$ used in the GTV penalty, computed using the covariance estimate $\hat{\Sigma}$, has more nonzero weighted edges compared to the given graph $\Gamma$. These factors contribute to longer tuning time for the GTV penalty. The tuning time is much worse for the GTV-oracle penalty since the true covariance matrix $\Sigma$ is denser than $\hat{\Sigma}$, and so $\hat{\Gamma}_{\text{oracle}}$ has many more nonzero weighted edges than $\hat{\Gamma}$ does. We present the tuning times for a toy example where each hyperparameter is selected from a small grid search in Table \ref{tab:cv_runtime}. Here, since $n$ and $p$ are both not too large, we use the SOCP solver ECOS (see Section \ref{subsec:runtime experiments}) for all methods.

\begin{table}[H]
    \caption{Tuning times with ECOS when $G$ is the chain graph. $p = 110, m=109, n = 210, \sigma = 1$, and $\Sigma$ is constructed as in Section \ref{sec:params_choice}. The GTV penalty is based on $\hat{\Gamma}$ which has around 200 nonzero weighted edges. The GTV-oracle penalty is based on $\hat{\Gamma}_{\text{oracle}}$ which has almost 6000 nonzero weighted edges. 5-fold CV is performed for each method on a small grid with 5 candidate values $[0, 0.1, 1, 10, 100]$ for each hyperparameter.}
    \centering
    \begin{tabularx}{\textwidth}{c X X X X X X X}
    \hline
         & \textbf{L} & \textbf{EN} & \textbf{FL} & \textbf{SL} & \textbf{GTV}& \textbf{GTV-oracle} & \textbf{GEN}\\[1mm]
    \hline \hline  
    \textbf{\# hyperparameters} & 1 & 2 & 2& 2 & 3 & 3 & 2\\[1mm]
    \hline
    \textbf{time [seconds]} & 0.45 & 3.22& 2.85& 2.30 & 14.31 & 134.08 & 2.46\\[1mm]
    \hline
\end{tabularx}
    \label{tab:cv_runtime}
\end{table}

When the graph $G$ contains more edges, we can expect the tuning time for GEN to increase relative to other two-hypterparameter methods, as both the $\ell_1$ and $\ell_2$ components of the GEN penalty depend on $\Gamma$. Table \ref{tab:cv_runtime_barbell} repeats the above experiment but with $G$ being the barbell graph and $\Sigma$ reflecting the structure of this graph. The tuning time with ECOS for GEN is roughly double that of FL or SL, whose penalties contain only one component depending on $\Gamma$.

\begin{table}[H]
    \caption{Tuning times with ECOS when $G$ is the barbell graph. $p = 110, m = 2461, n = 210, \sigma = 1$, and $\Sigma$ is constructed as in Section \ref{sec:params_choice}. The GTV penalty is based on $\hat{\Gamma}$ which has around 2500 nonzero weighted edges. As $\Sigma$ for the barbell graph is denser than $\Sigma$ in Table \ref{tab:cv_runtime}, $\hat{\Sigma}$ here is also denser than $\hat{\Sigma}$ in Table \ref{tab:cv_runtime}.}
    \centering
    \begin{tabularx}{\textwidth}{c X X X X X X X}
    \hline
         & \textbf{L} & \textbf{EN} & \textbf{FL} & \textbf{SL} & \textbf{GTV}& \textbf{GTV-oracle} & \textbf{GEN}\\[1mm]
    \hline \hline  
    \textbf{\# hyperparameters} & 1 & 2 & 2& 2 & 3 & 3 & 2\\[1mm]
    \hline
    \textbf{time [seconds]} & 0.44 & 2.51& 4.62& 4.57 & 52.47& 131.91 & 9.68\\[1mm]
    \hline
\end{tabularx}
    \label{tab:cv_runtime_barbell}
\end{table}

\subsubsection{Comparisons between GEN, FL and SL when $\beta^*$ is dense but aligns with $G$} \label{sec:adaptivity}

In this section, we focus on the case when $\beta^*$ is not sparse but $\Gamma\beta^*$ is sparse or $\beta^*$ is otherwise smooth with respect to $G$. As all parts of the true signals constructed in this section are far from zero, the component $\lambda_L\|\beta\|_1$ in the FL and SL penalties is of little use, and setting $\lambda_L > 0$ worsens both prediction and estimation errors in this setting. Consequently, CV yields $\lambda_L$ values that are almost identically zero for both FL and SL. Essentially, in this section, FL refers to the standalone $\lambda_1\|\Gamma\beta\|_1$ penalty, whereas SL refers to the standalone $\lambda_2\|\Gamma\beta\|_2^2$ penalty.   

We observe that FL performs well when $\beta^*$ has few signal jumps on $G$, regardless of whether there exists large jumps (i.e. $\|\Gamma\beta^*\|_\infty$ is large). SL, on the other hand, tends to perform well when the signal is smooth with respect to $G$, in the sense that $\|\Gamma\beta^*\|_\infty$ is small, even if the number of signal jumps $\|\Gamma\beta^*\|_0$ might be large. To demonstrate these observations, we construct signals with varying smoothness ($\|\Gamma\beta^*\|_\infty$) and numbers of jumps ($\|\Gamma\beta^*\|_0$). Figure \ref{fig:1D_signals} illustrates the true signals on the 1D chain graph, whereas Figure \ref{fig:2D_signals} illustrates the true signals on the 2D grid graph (note that $p$ is fixed for these graphs). For the barbell graph, we let the signal values be constant (at 5 and 20 respectively) on each clique. The lengths of the path connecting the two cliques are chosen from $\{1,4,7,10,13,16\}$ (and so $p$ has to vary), and we let the signal decrease from 20 to 5 gradually on the connecting path, so that $\|\Gamma\beta^*\|_\infty$ decreases from 15 to 1.46 while $\|\Gamma\beta^*\|_0$ increases from 1 to 16.

\begin{figure}[H]
    \centering
\begin{subfigure}{0.3\textwidth}
    \centering
 \includegraphics[width=\textwidth]{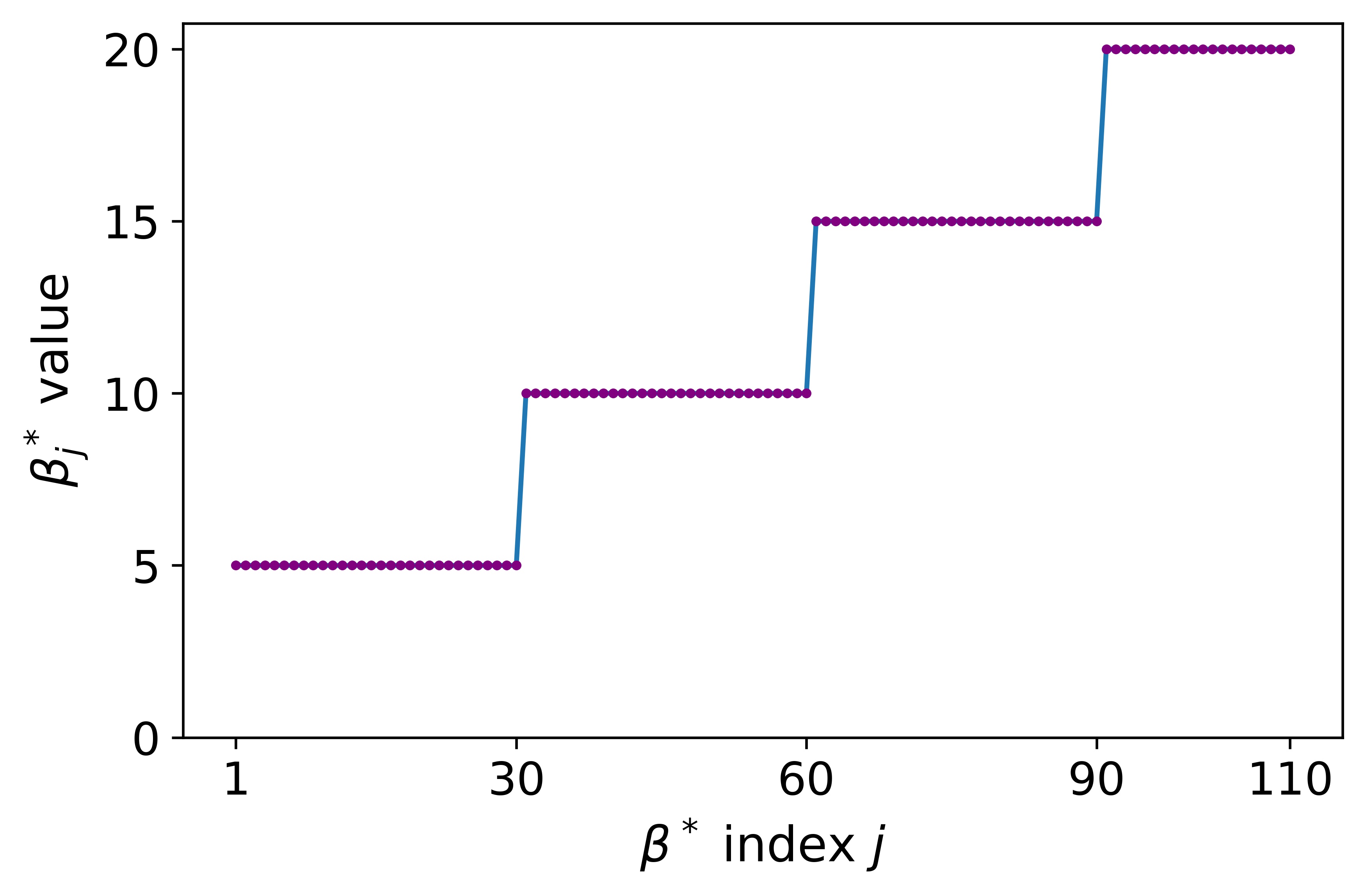}   
\end{subfigure}
   \hfill
\begin{subfigure}{0.3\textwidth}
    \centering
 \includegraphics[width=\textwidth]{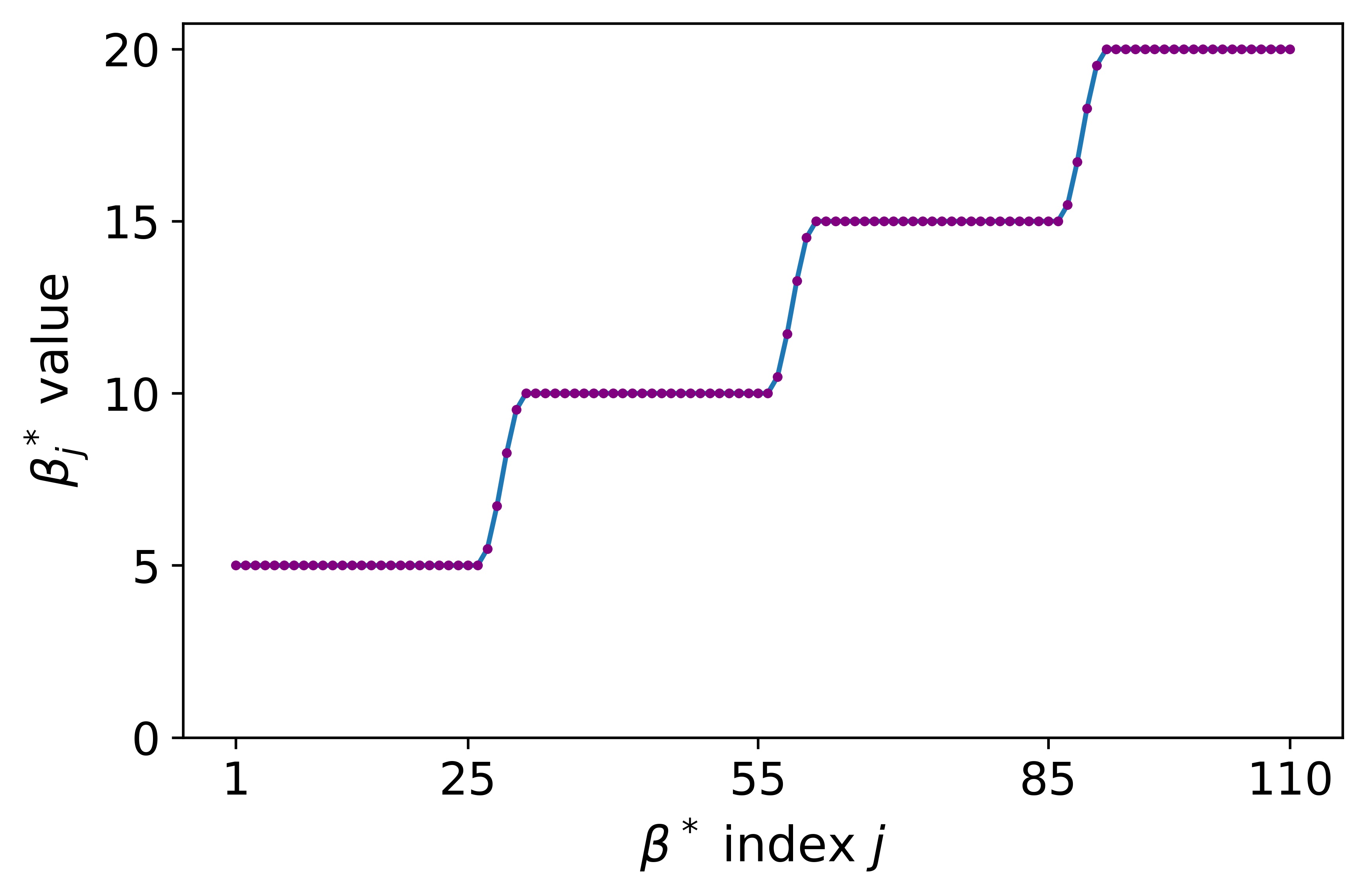}  
\end{subfigure}
    \hfill
\begin{subfigure}{0.3\textwidth}
    \centering
 \includegraphics[width=\textwidth]{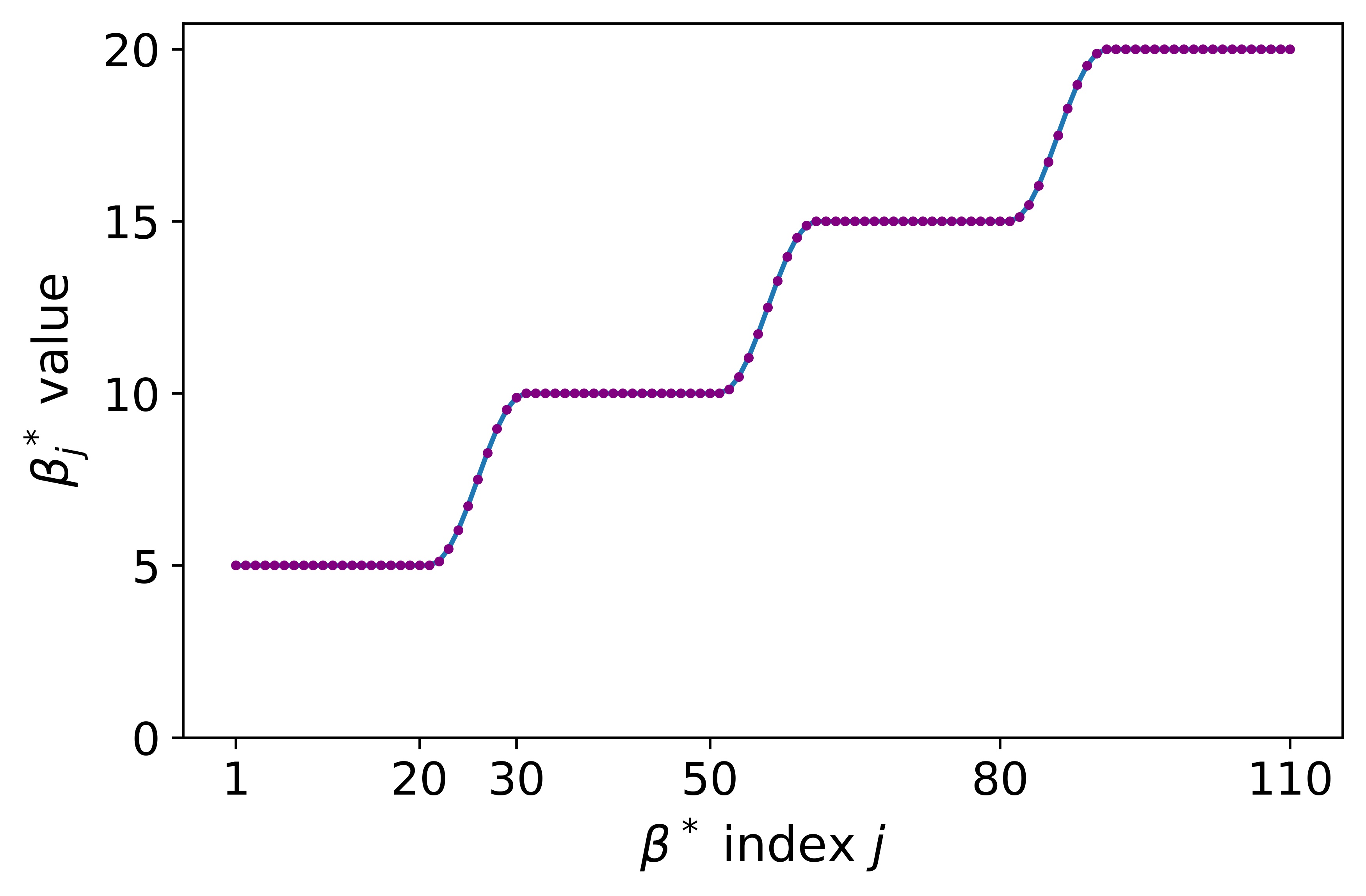}  
\end{subfigure}

\begin{subfigure}{0.3\textwidth}
    \centering
 \includegraphics[width=\textwidth]{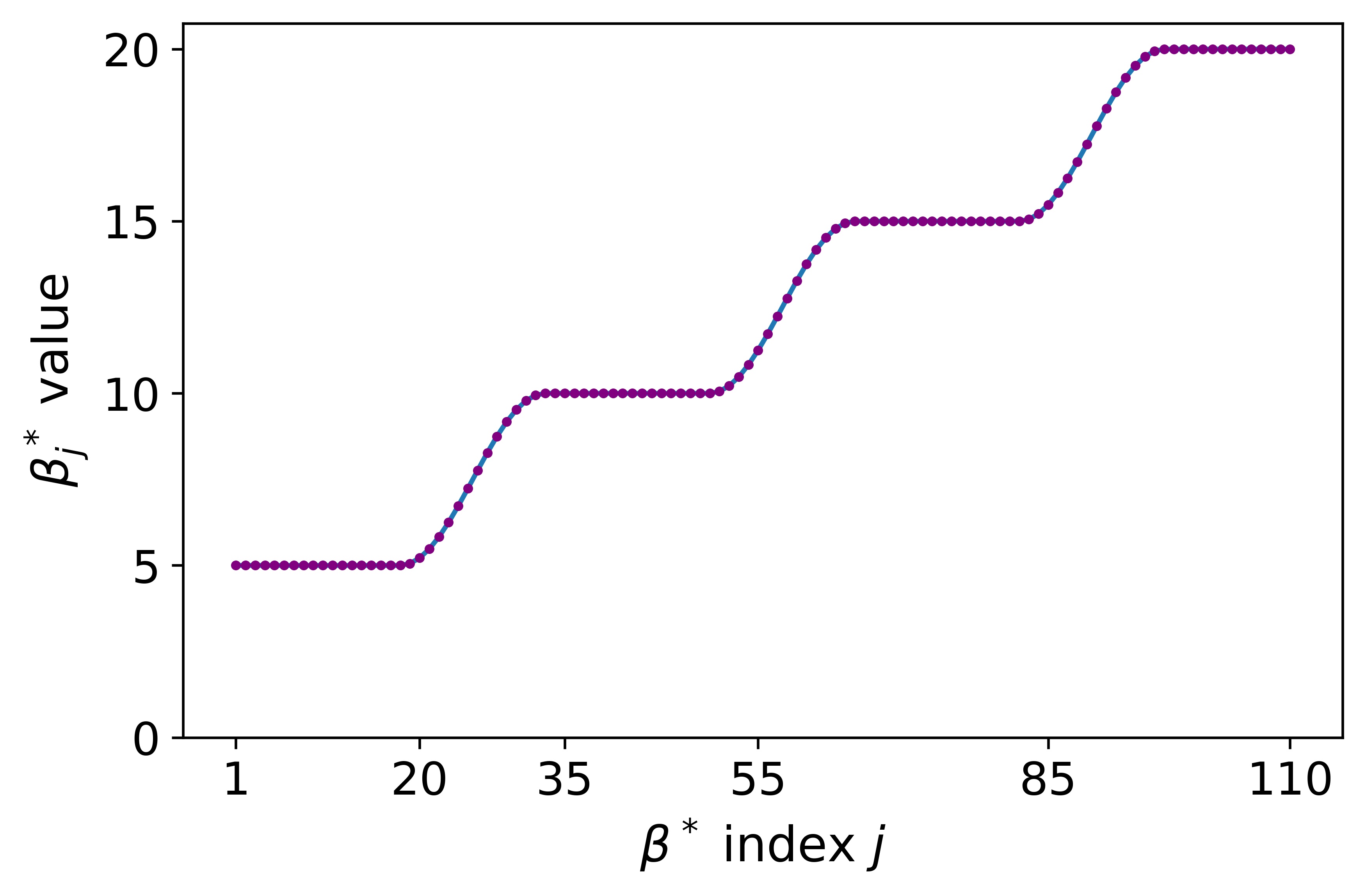}  
\end{subfigure}
    \hfill
\begin{subfigure}{0.3\textwidth}
    \centering
 \includegraphics[width=\textwidth]{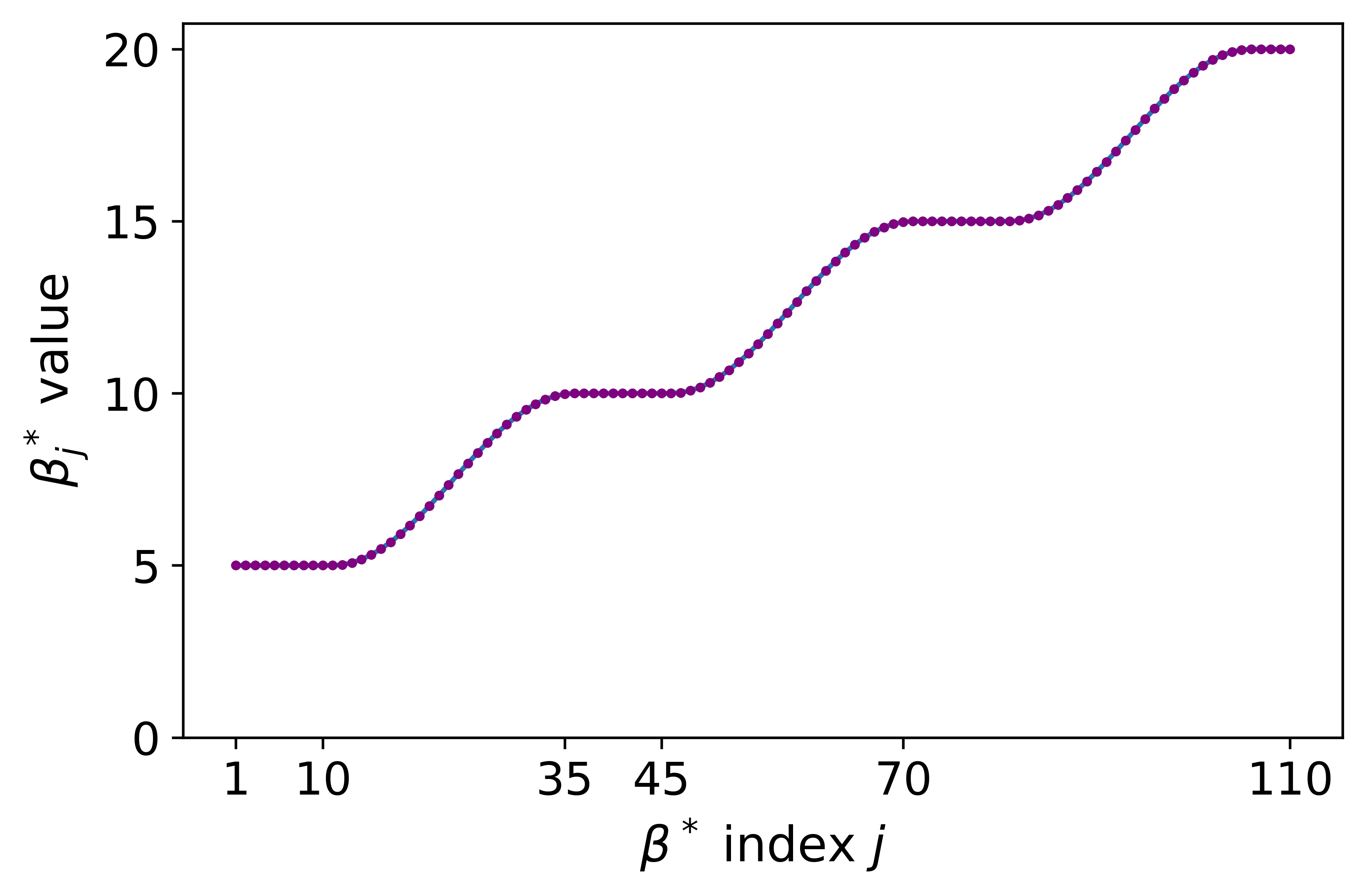}  
\end{subfigure}
    \hfill
\begin{subfigure}{0.3\textwidth}
    \centering
 \includegraphics[width=\textwidth]{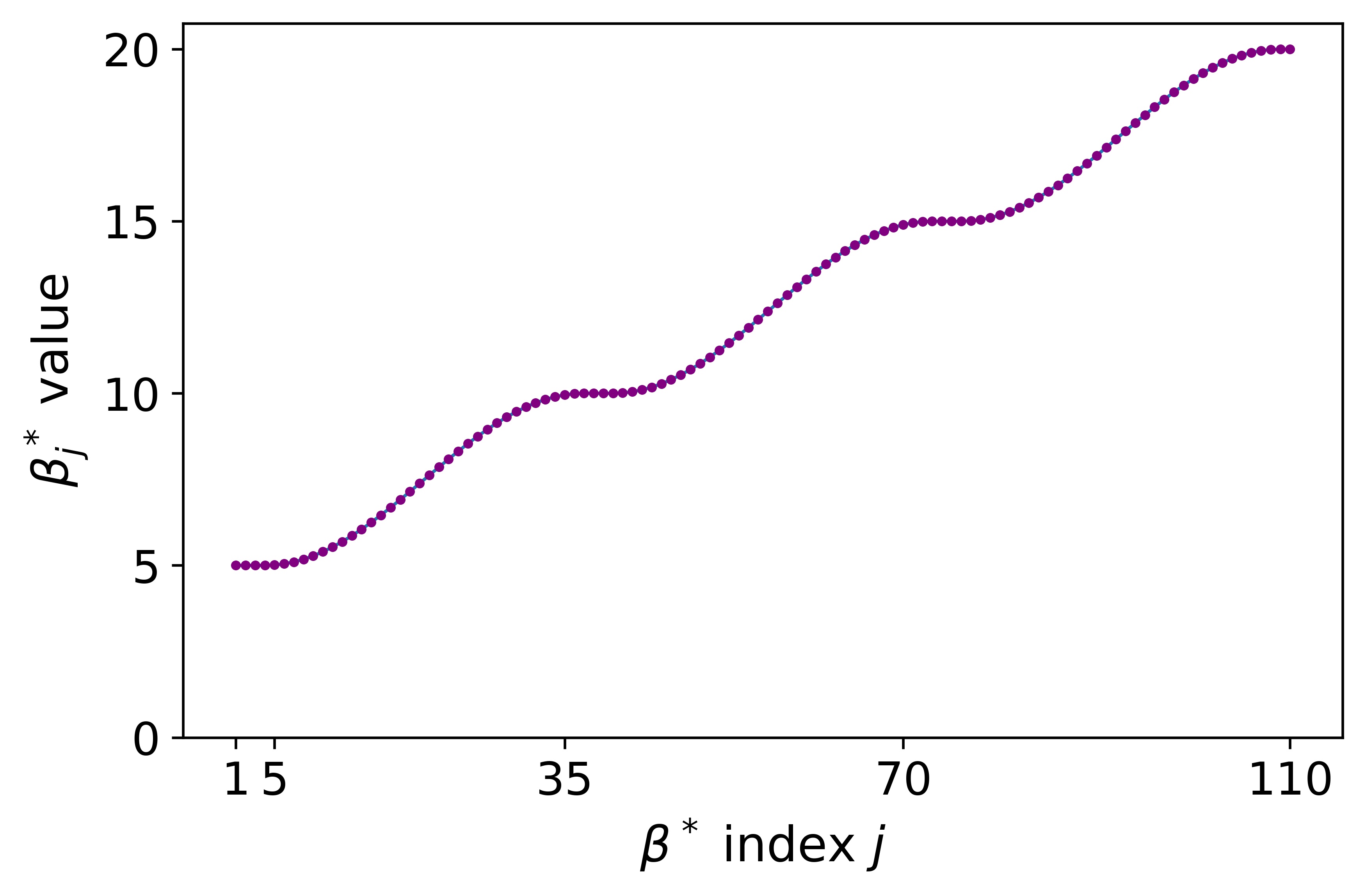}  
\end{subfigure}

    \caption{\textit{True signals defined on the chain graph with $p = 110$. The top left signal is piecewise constant and has the smallest $\|\Gamma\beta^*\|_0 = 3$ but the largest $\|\Gamma\beta^*\|_\infty = 5$. The bottom right signal is the smoothest with the largest $\|\Gamma\beta^*\|_0 = 99$ and the smallest $\|\Gamma\beta^*\|_\infty = 0.24$. The intermediate signals are constructed such that $\|\Gamma\beta^*\|_0$ decreases but $\|\Gamma\beta^*\|_\infty$ increases gradually. All 6 signals have $\|\Gamma\beta^*\|_1 = 15$.}}
    \label{fig:1D_signals}
\end{figure}

\begin{figure}[H]
    \centering
\begin{subfigure}{0.31\textwidth}
    \centering
 \includegraphics[width=\textwidth]{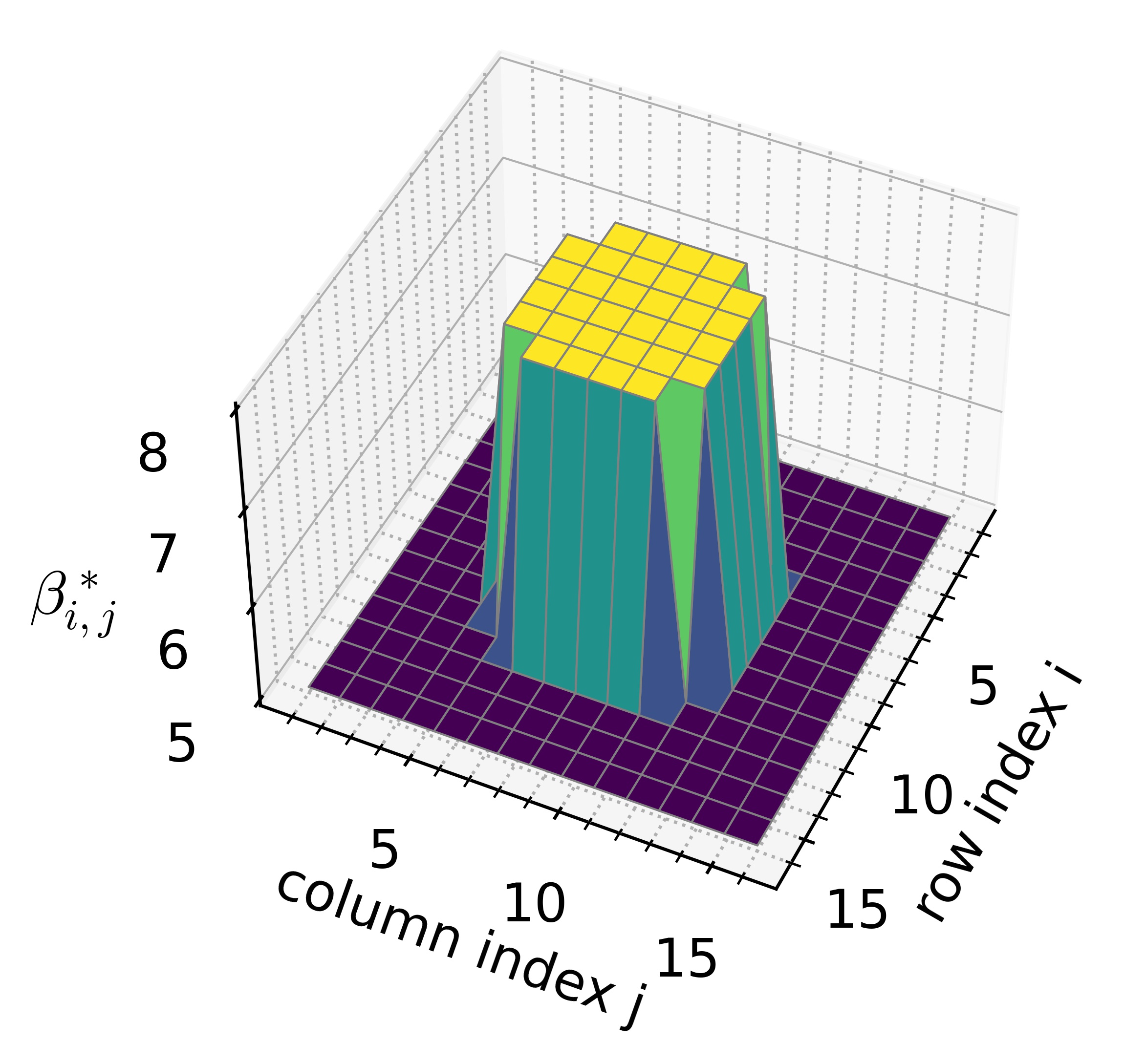}   
\end{subfigure}
   \hfill
\begin{subfigure}{0.31\textwidth}
    \centering
 \includegraphics[width=\textwidth]{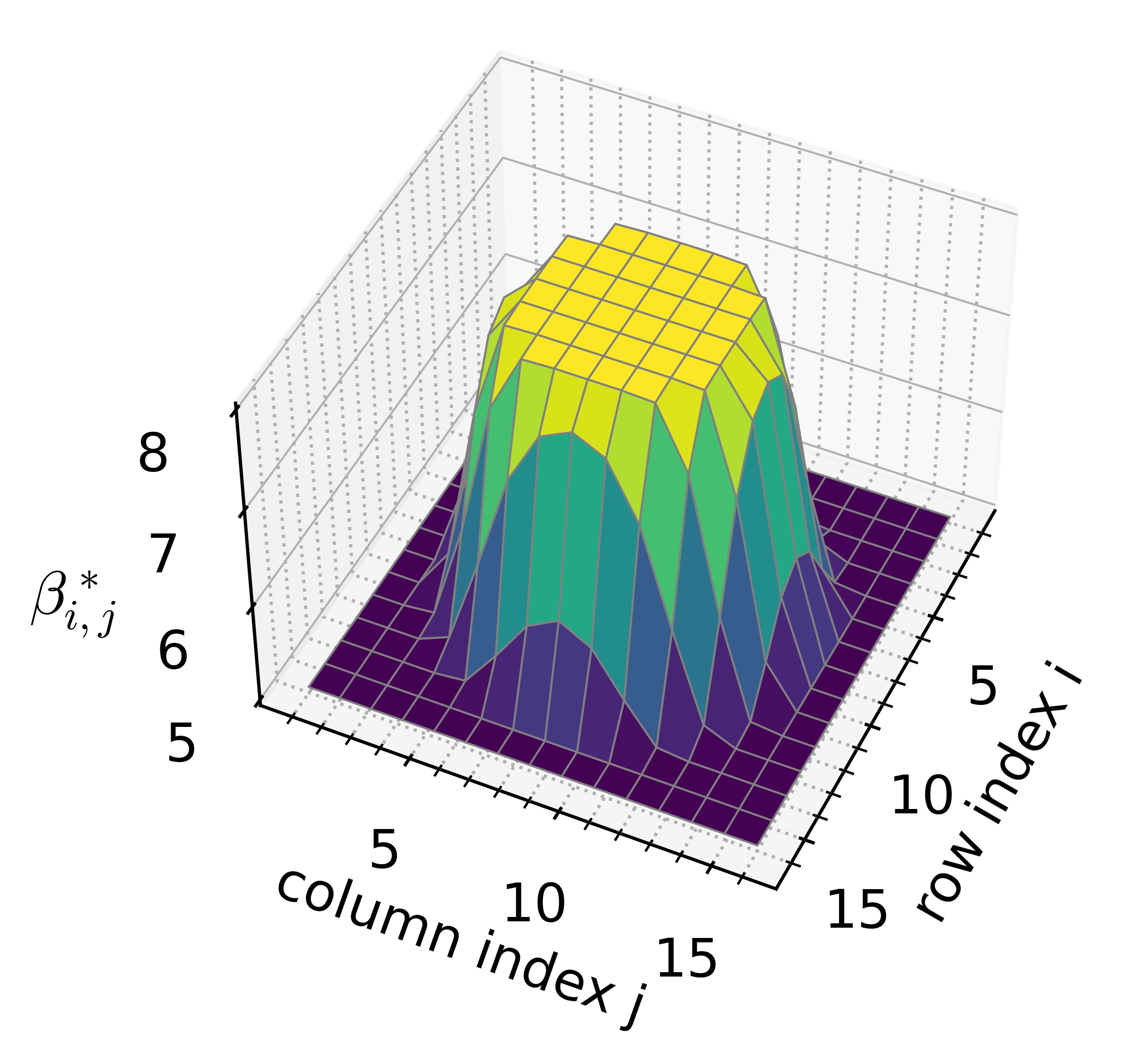}  
\end{subfigure}
    \hfill
\begin{subfigure}{0.31\textwidth}
    \centering
 \includegraphics[width=\textwidth]{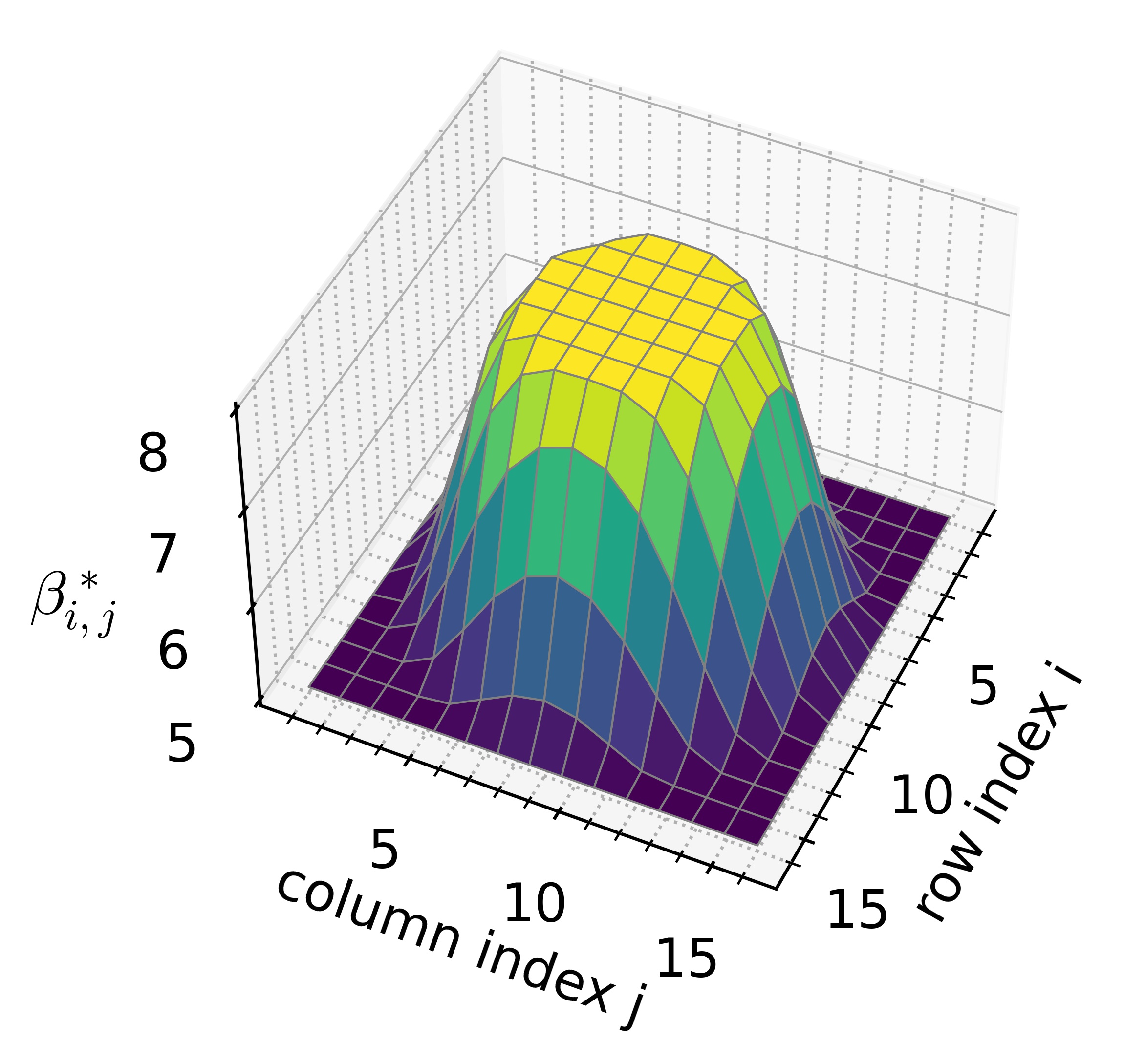}  
\end{subfigure}

\begin{subfigure}{0.31\textwidth}
    \centering
 \includegraphics[width=\textwidth]{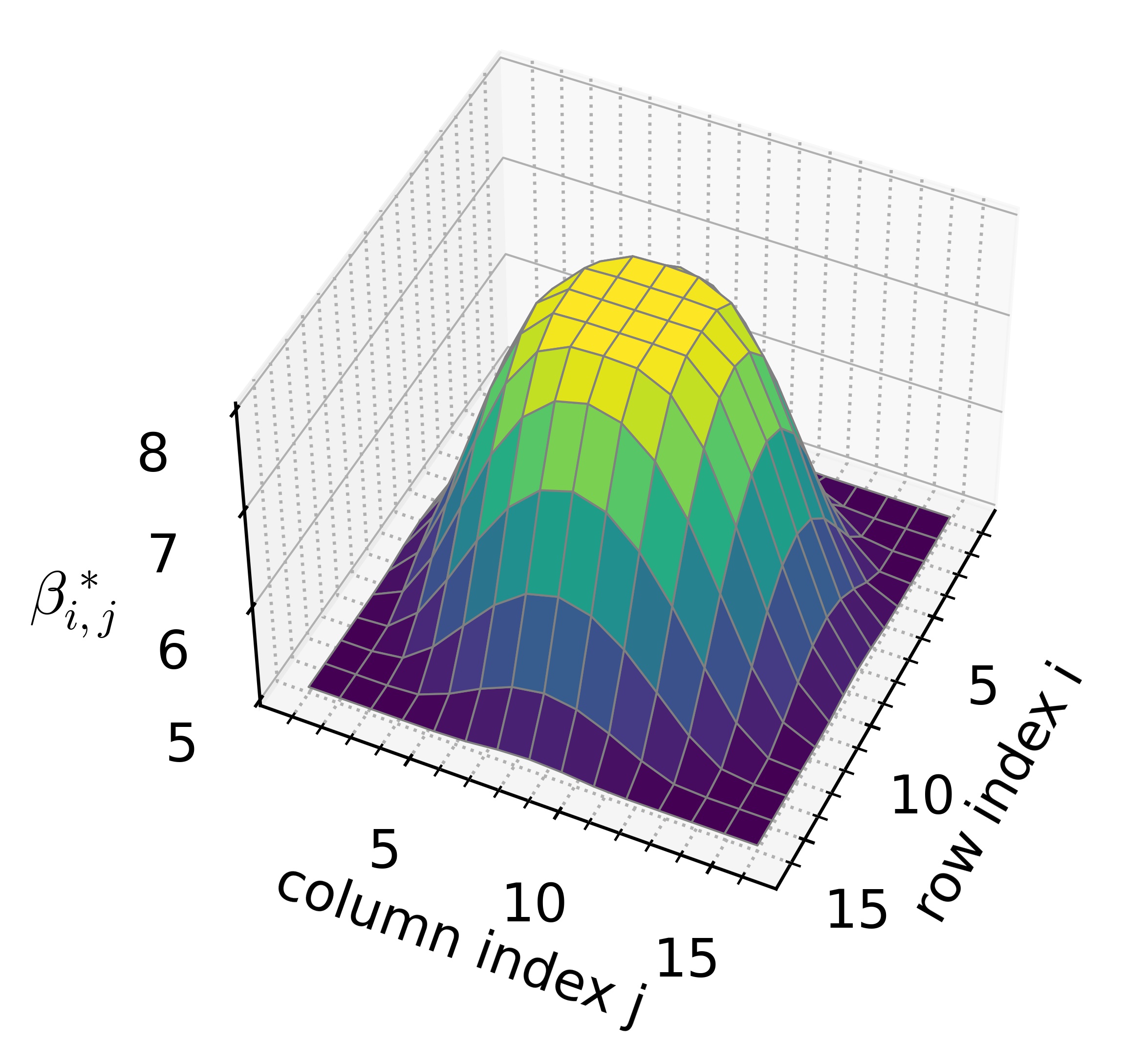}  
\end{subfigure}
    \hfill
\begin{subfigure}{0.31\textwidth}
    \centering
 \includegraphics[width=\textwidth]{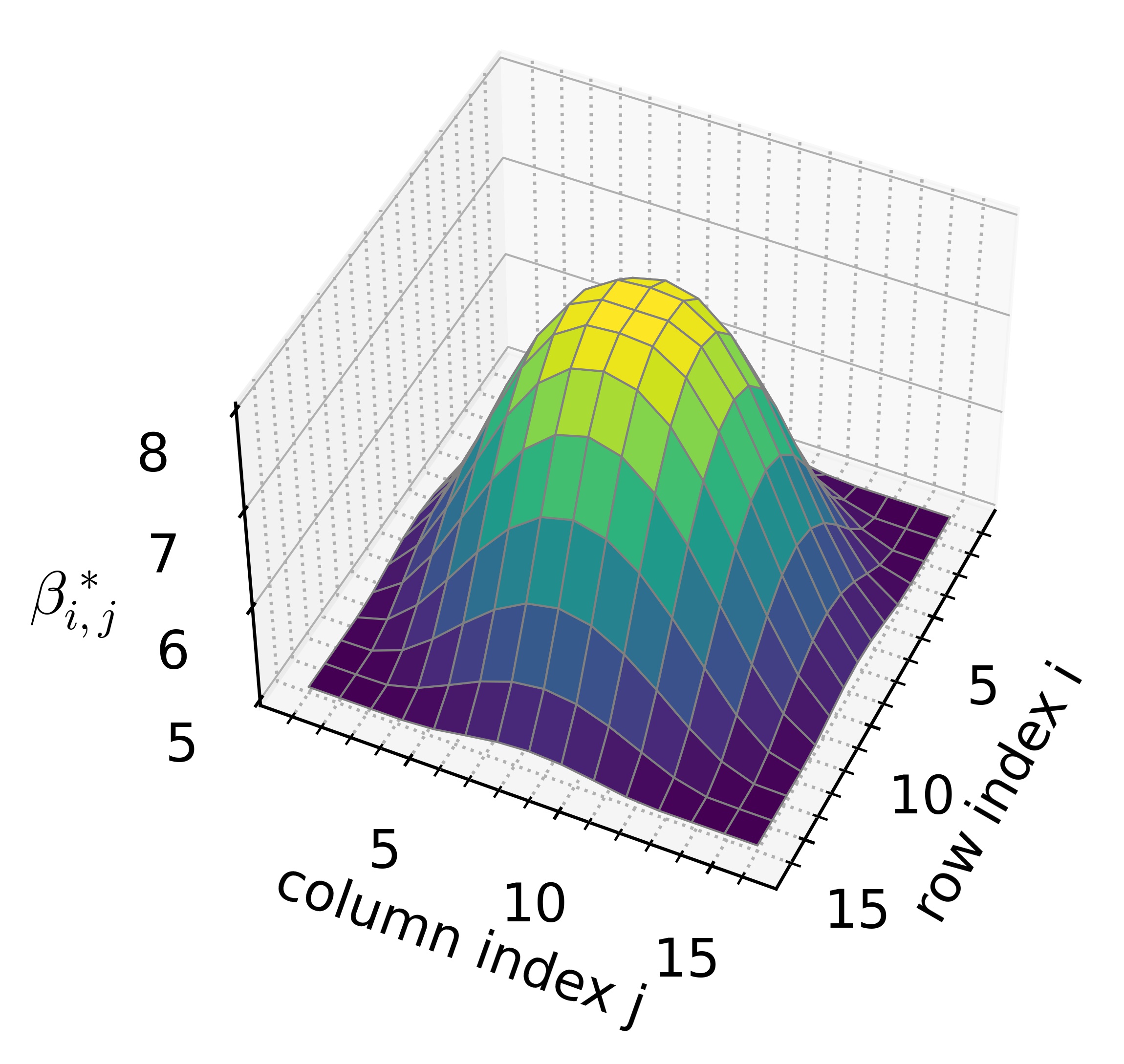}  
\end{subfigure}
    \hfill
\begin{subfigure}{0.31\textwidth}
    \centering
 \includegraphics[width=\textwidth]{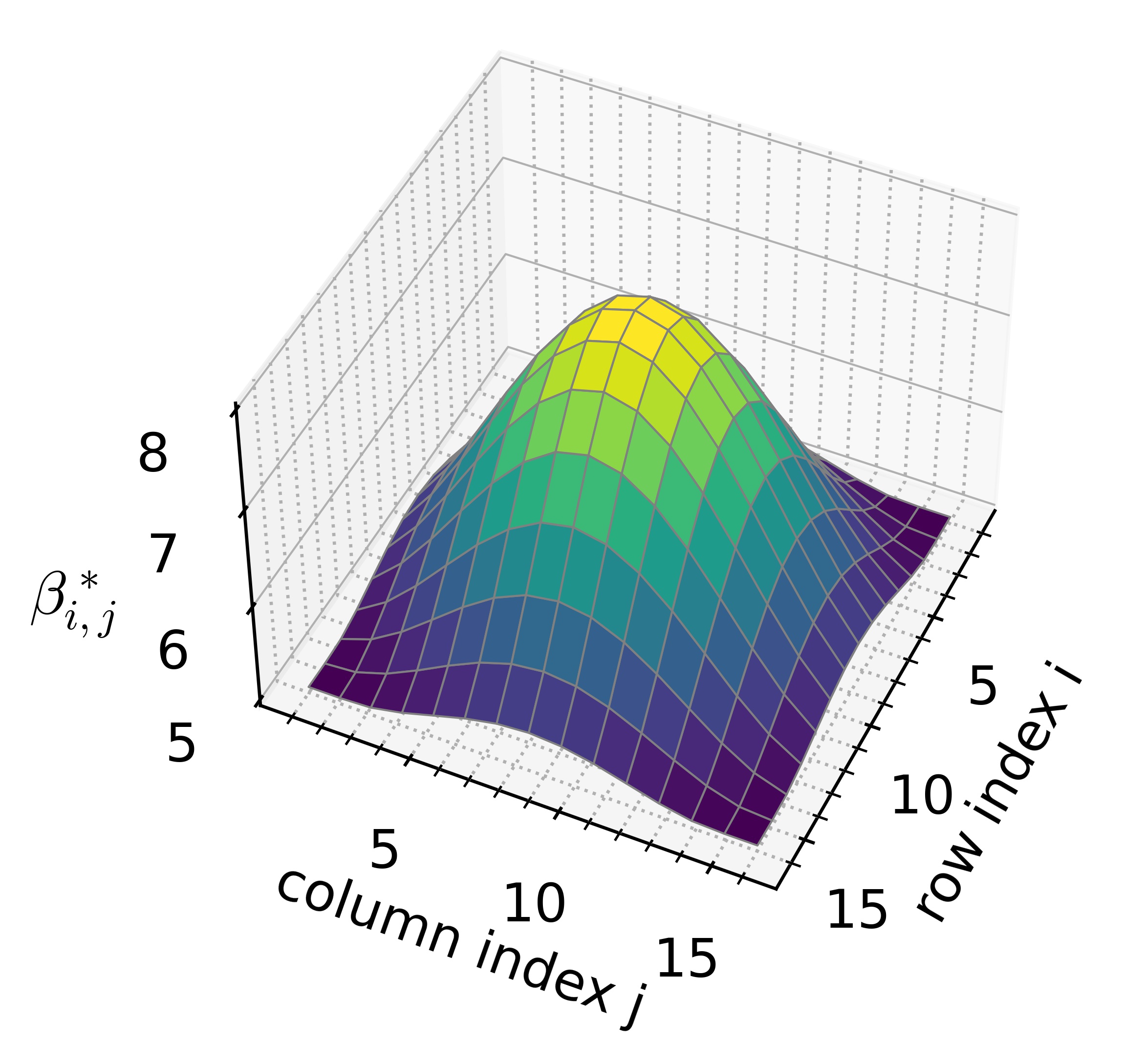}  
\end{subfigure}
    \caption{\textit{True signals defined on the 2D grid with $p = 15 \times 15$. The top left signal is piecewise constant and has the smallest $\|\Gamma\beta^*\|_0 = 28$ but the largest $\|\Gamma\beta^*\|_\infty = 3$. The bottom right signal is the smoothest with the largest $\|\Gamma\beta^*\|_0 = 412$ and the smallest $\|\Gamma\beta^*\|_\infty = 0.24$. All 6 signals have $\|\Gamma\beta^*\|_1$ between 84 and 120.}}
    \label{fig:2D_signals}
\end{figure}

Figure \ref{fig:adaptivitys} illustrates the performances of FL, SL and GEN in terms of estimation and prediction errors. When $\|\Gamma\beta^*\|_\infty$ is small, CV yields $\lambda_2$ values that are larger relative to $\lambda_1$, which is consistent with our theory in that $\lambda_2$ can be chosen up to $\frac{C\lambda_1}{\|\Gamma\beta^*\|_\infty}$ without incurring additional errors. As can be seen from Figure \ref{fig:adaptivitys}, our GEN penalty adapts well to true signals of various smoothness levels, thus demonstrating the importance of having both the $\lambda_1\|\Gamma\beta\|_1$ and $\lambda_2\|\Gamma\beta\|_2^2$ components in our penalty. From the performances of FL and SL, we can see that the $\lambda_1\|\Gamma\beta\|_1$ penalty ensures good performance when the true signal is piecewise constant, whereas the $\lambda_2\|\Gamma\beta\|_2^2$ penalty ensures good performance when the true signal is smooth over $G$.

Note again that FL and SL in this setting correspond to the GEN penalty with $\lambda_2$ or $\lambda_1$ set to zero, respectively. Therefore, GEN's superior performance over FL and SL in terms of prediction error is not surprising, given that the scorer used for CV $-\frac{1}{n}\|Y-X\hat{\beta}\|_2^2$ selects for hyperparameters with stronger prediction performance. However, GEN is also consistently better than FL or SL in terms of estimation error. This can be understood better by examining the signal estimates obtained from the three procedures. Figure \ref{fig:1D_slice} compare the estimated signals with the true signals defined on the 1D chain graph. FL recovers the constant regions well but struggles with the smoothly increasing region, whereas the SL estimate is better in the smoothly increasing region but cannot reproduce the constant regions of the true signal. GEN, on the other hand, is able to recover the true signal in all regions. We can also make the same observations when $G$ is the 2D grid graph (but they are harder to visualize).

\begin{figure}[H]
    \centering
\begin{subfigure}[b]{\textwidth}
    \centering
 \includegraphics[width=0.49\linewidth]{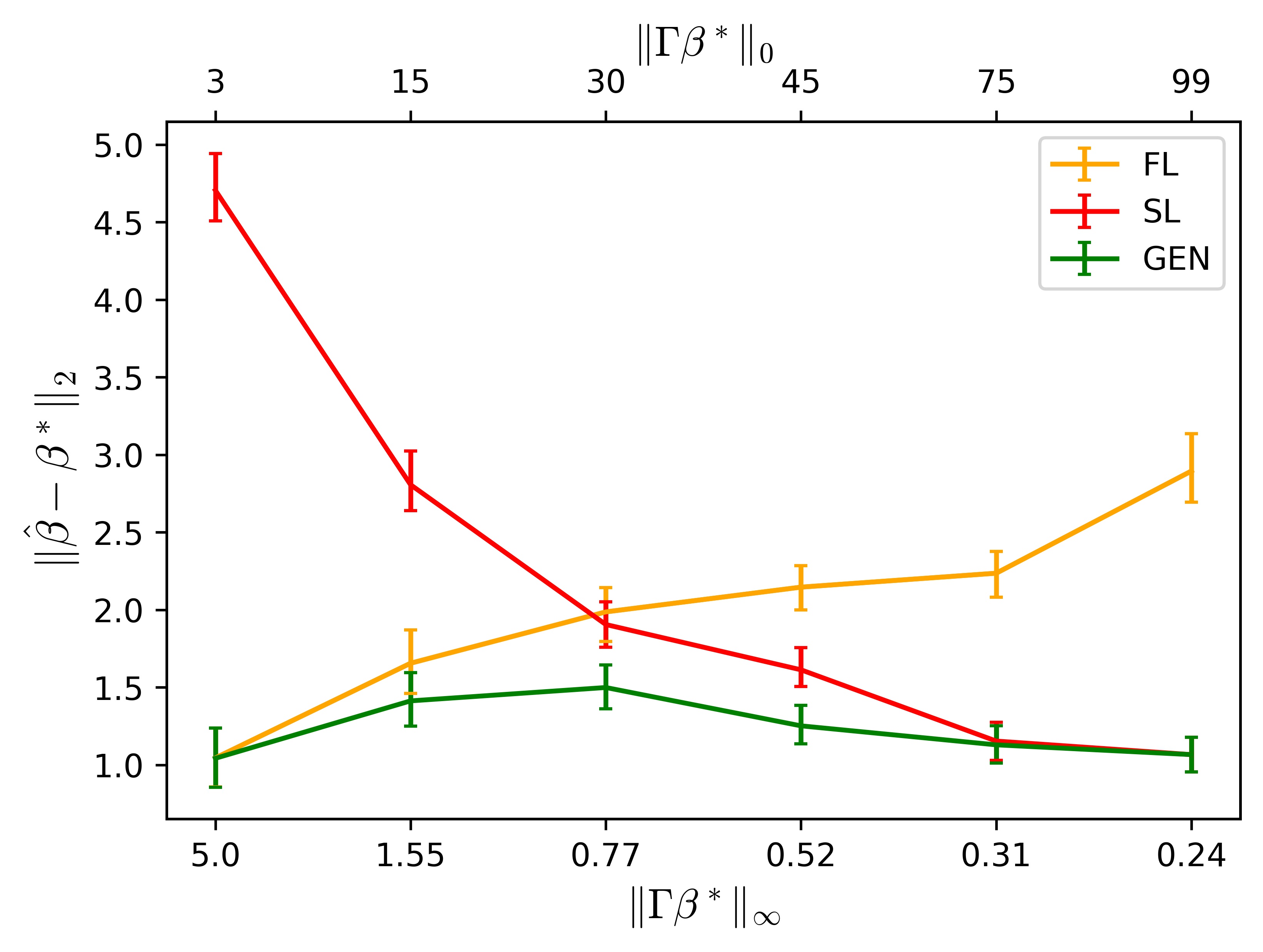}   
   \hfill
 \includegraphics[width=0.49\linewidth]{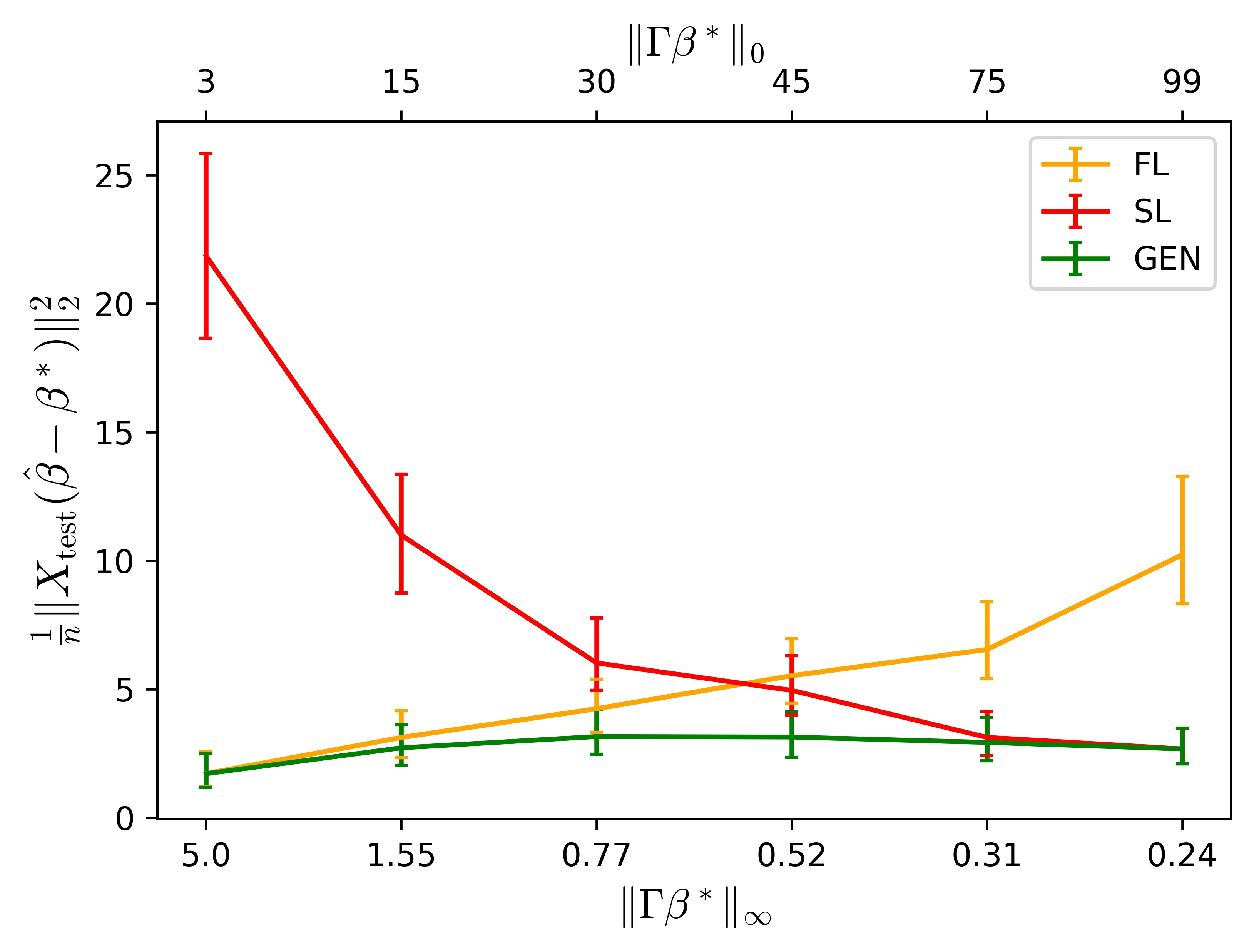} 
 \caption{1D chain}
\end{subfigure}

\begin{subfigure}[b]{\textwidth}
    \centering
\includegraphics[width = 0.49\linewidth]{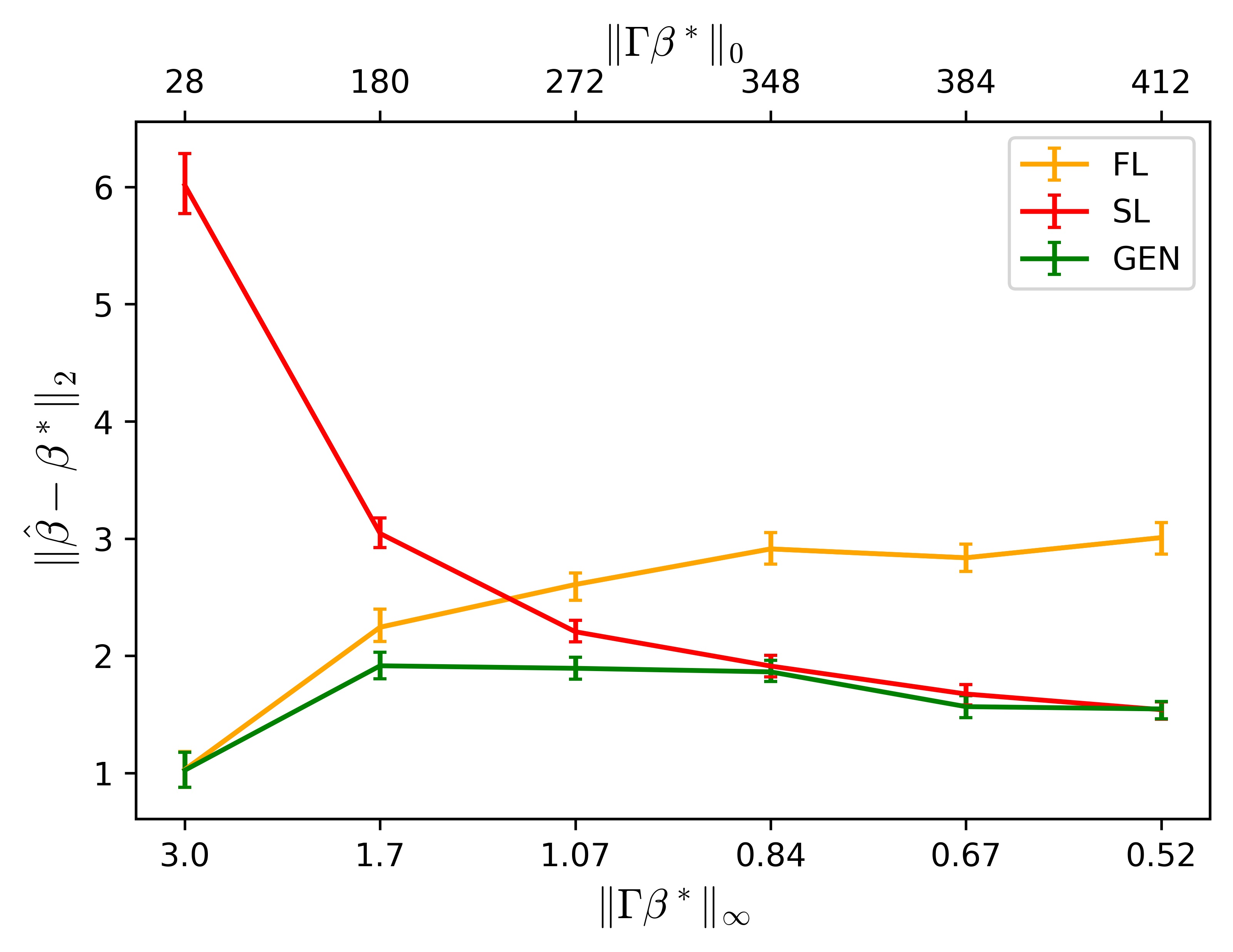}
    \hfill
\includegraphics[width = 0.49\linewidth]{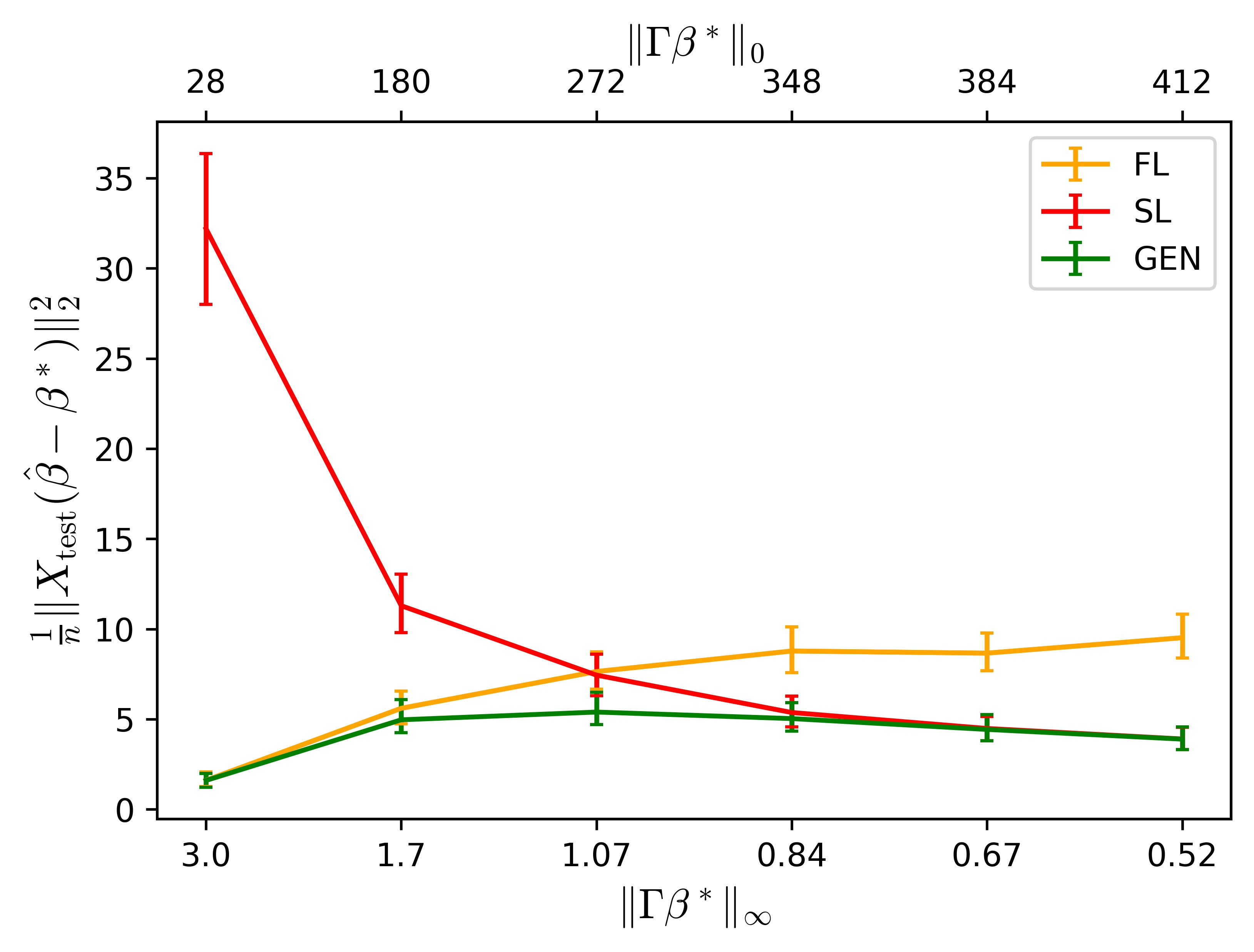}
\caption{2D grid}
\end{subfigure}

\begin{subfigure}[b]{\textwidth}
    \centering
\includegraphics[width = 0.49\linewidth]{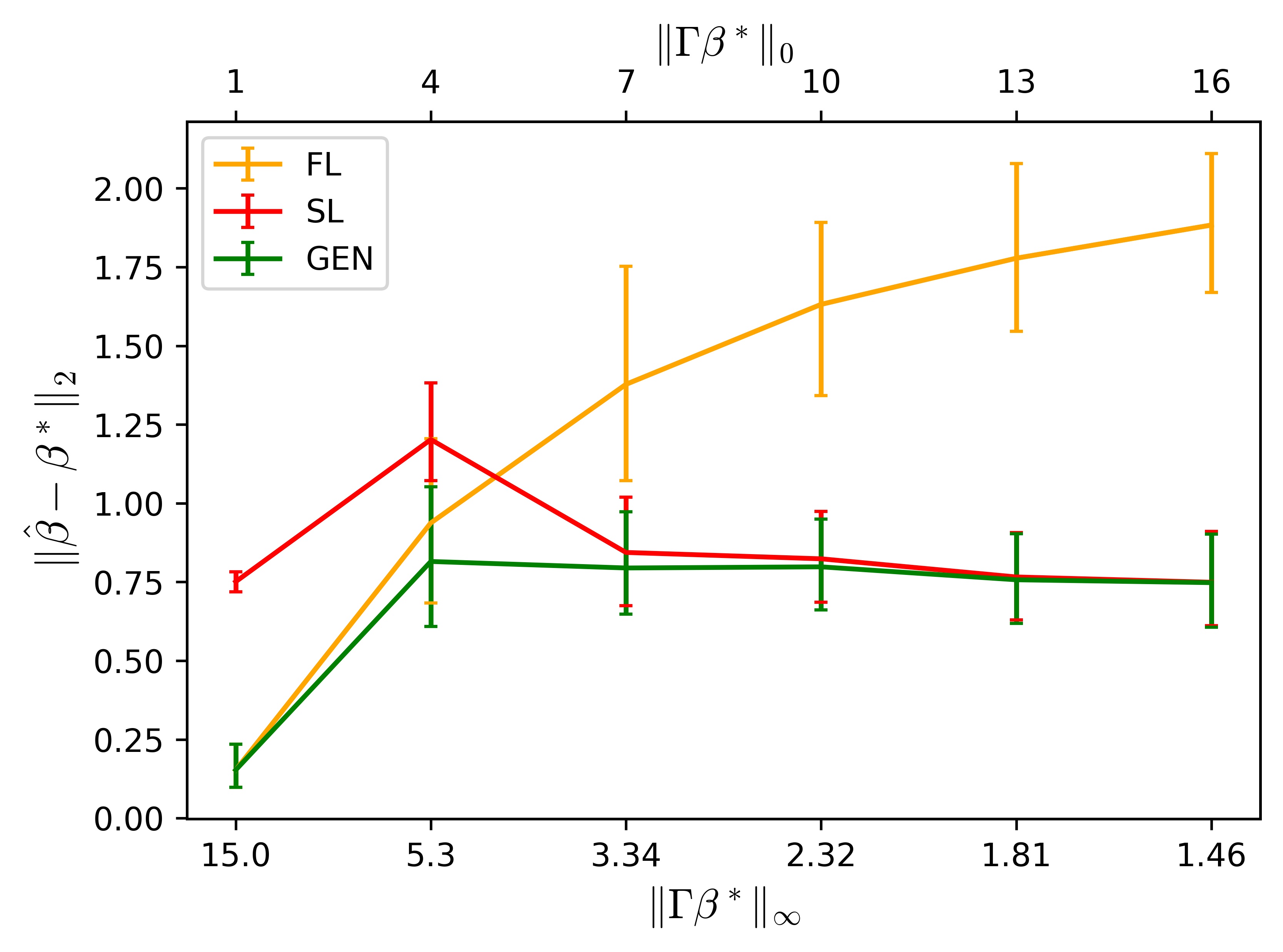}
    \hfill
\includegraphics[width = 0.49\linewidth]{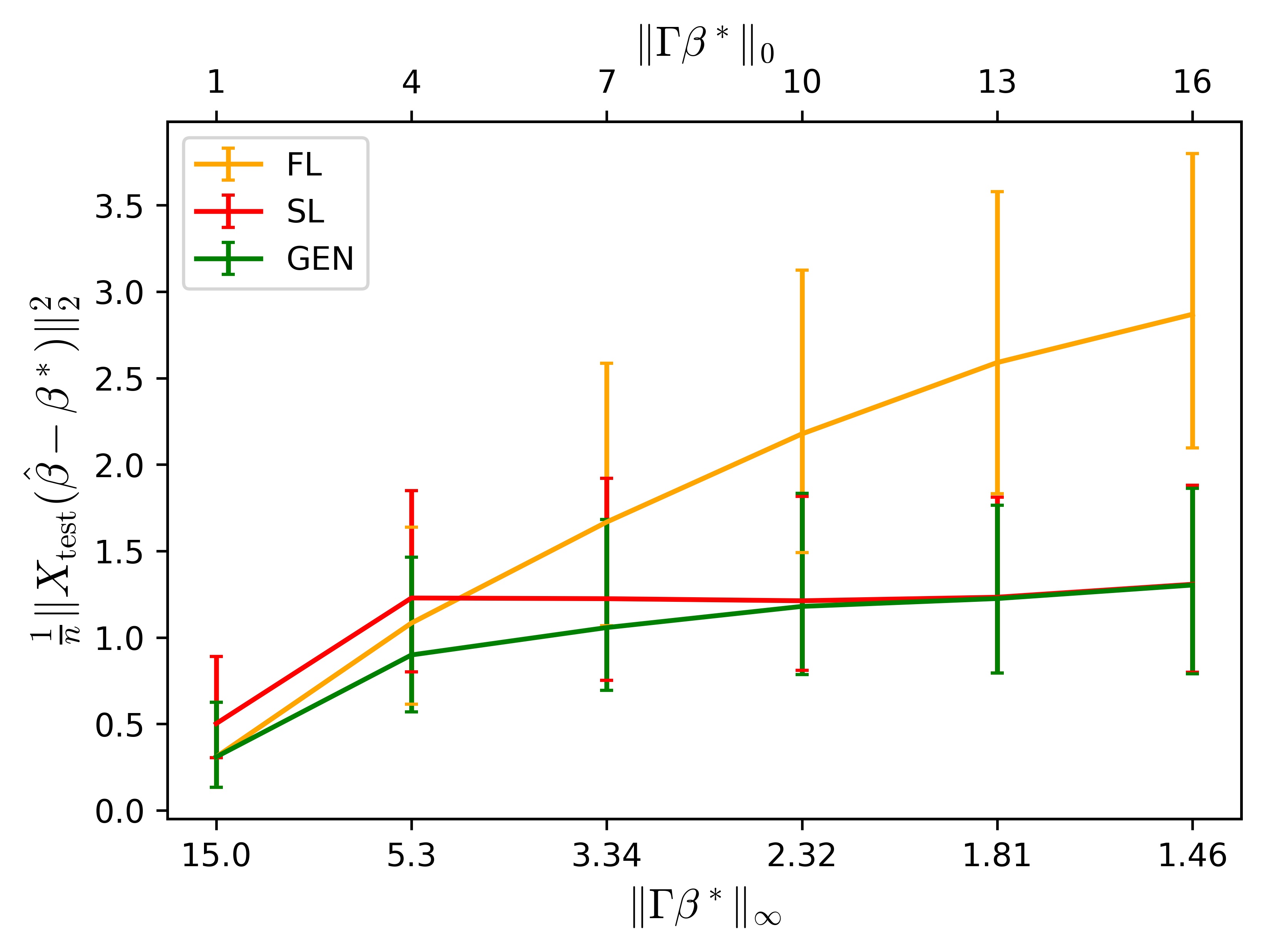}
\caption{Barbell}
\end{subfigure}

    \caption{\textit{Prediction and estimation errors for three graphs as $\|\Gamma\beta^*\|_\infty$ and $\|\Gamma\beta^*\|_0$ vary. Results are based on 500 resamplings. Vertical bars for each true signal connect the $25^\text{th}$ and $75^\text{th}$ percentiles. The lines labeled by FL, SL and GEN connect the medians of errors. 1D chain: $n = 70, p = 110$. 2D grid: $n = 150, p = 225$. Barbell graph: $n = \frac{2}{3}p$, and $p = 60, 63, 66, 69, 72, 75$.}}
    \label{fig:adaptivitys}
\end{figure}

\begin{figure}[H]
    \centering
\begin{subfigure}[b]{0.49\textwidth}
    \centering
 \includegraphics[width=\textwidth]{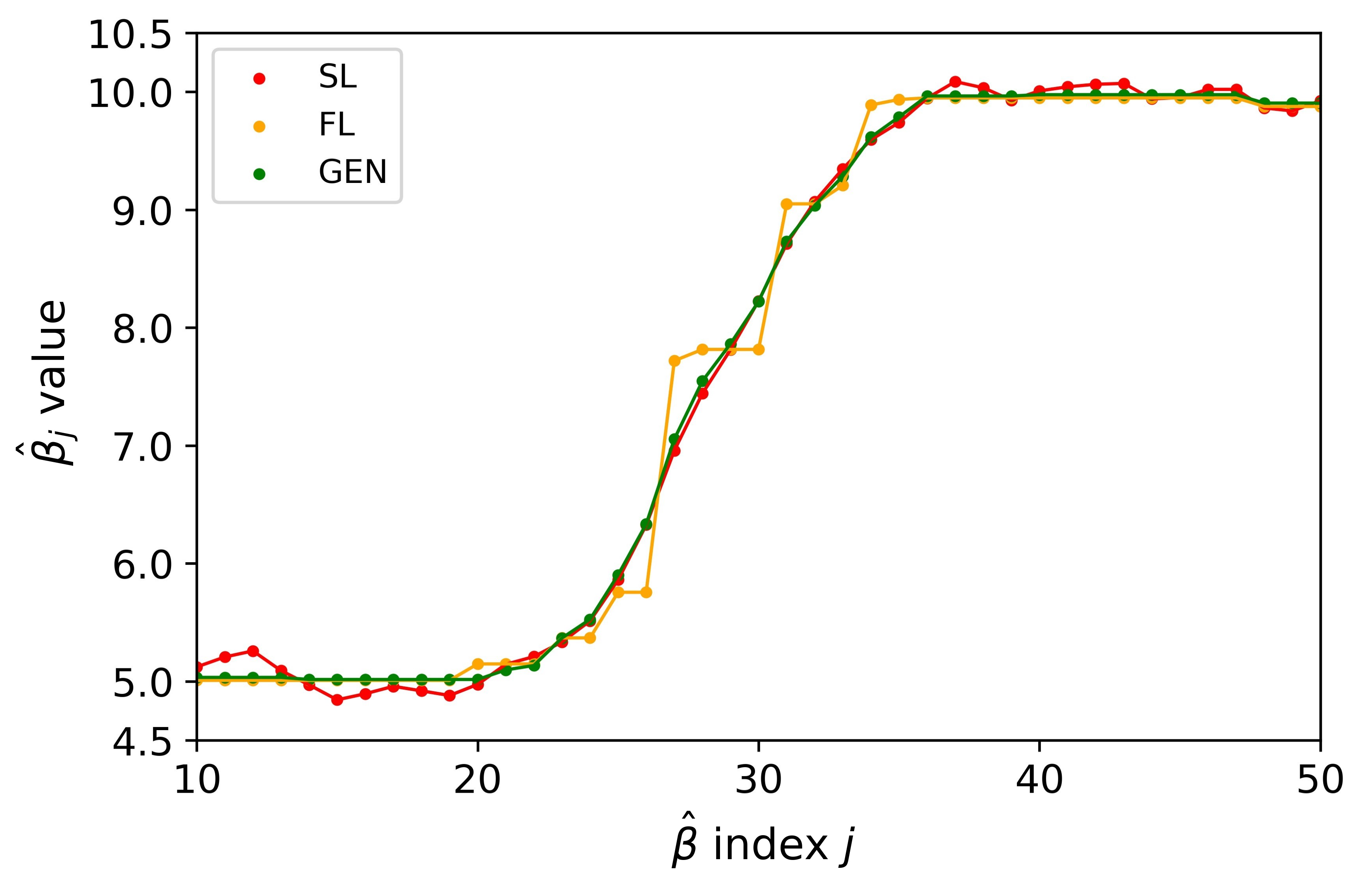}   
\end{subfigure}
   \hfill
\begin{subfigure}[b]{0.49\textwidth}
    \centering
 \includegraphics[width=\textwidth]{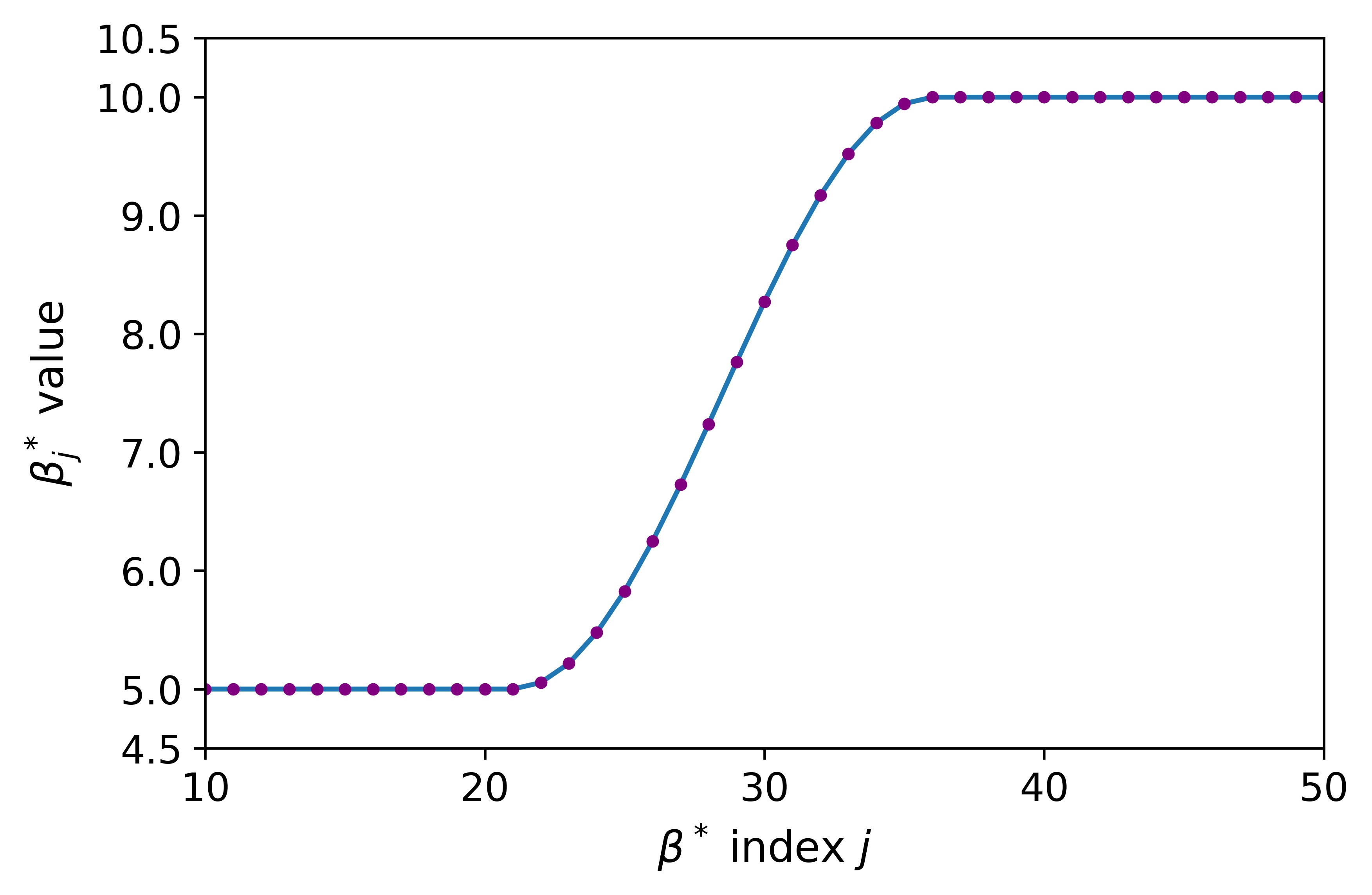}  
\end{subfigure}
    \caption{\textit{Left: estimated signals obtained from FL, SL and GEN. Right: true signal. GEN recovers the true signal well in both the constant and the smoothly increasing regions.}}
    \label{fig:1D_slice}
\end{figure}



We also examine the performances of FL, SL and GEN as features become more correlated and thus $\Sigma$ becomes more ill-conditioned. Figure \ref{fig:toeplitz_rho} shows the estimation errors for the chain graph when $\Sigma$ is the identity matrix (which is the limit of the Toeplitz covariance matrix as $\rho \to 0$) and when $\Sigma$ is the Toeplitz covariance matrix with $\rho = 0.95$. From our theoretical results, we expect that when the features are highly correlated, the performance of the standalone $\lambda_1\|\Gamma\beta\|_1$ penalty (FL) should be negatively affected by the ill-conditioned nature of $\Sigma$. However, the $\lambda_2\|\Gamma\beta\|_2^2$ penalty should improve the minimum eigenvalue term in the denominators of our error bounds, especially when $\|\Gamma\beta^*\|_\infty$ is small and $\lambda_2$ can be chosen to be larger. Such an improvement is not as noticeable when $\Sigma$ is the identity, which is already well-conditioned. 

\begin{figure}[H]
    \centering
\begin{subfigure}{0.49\textwidth}
    \centering
 \includegraphics[width=\textwidth]{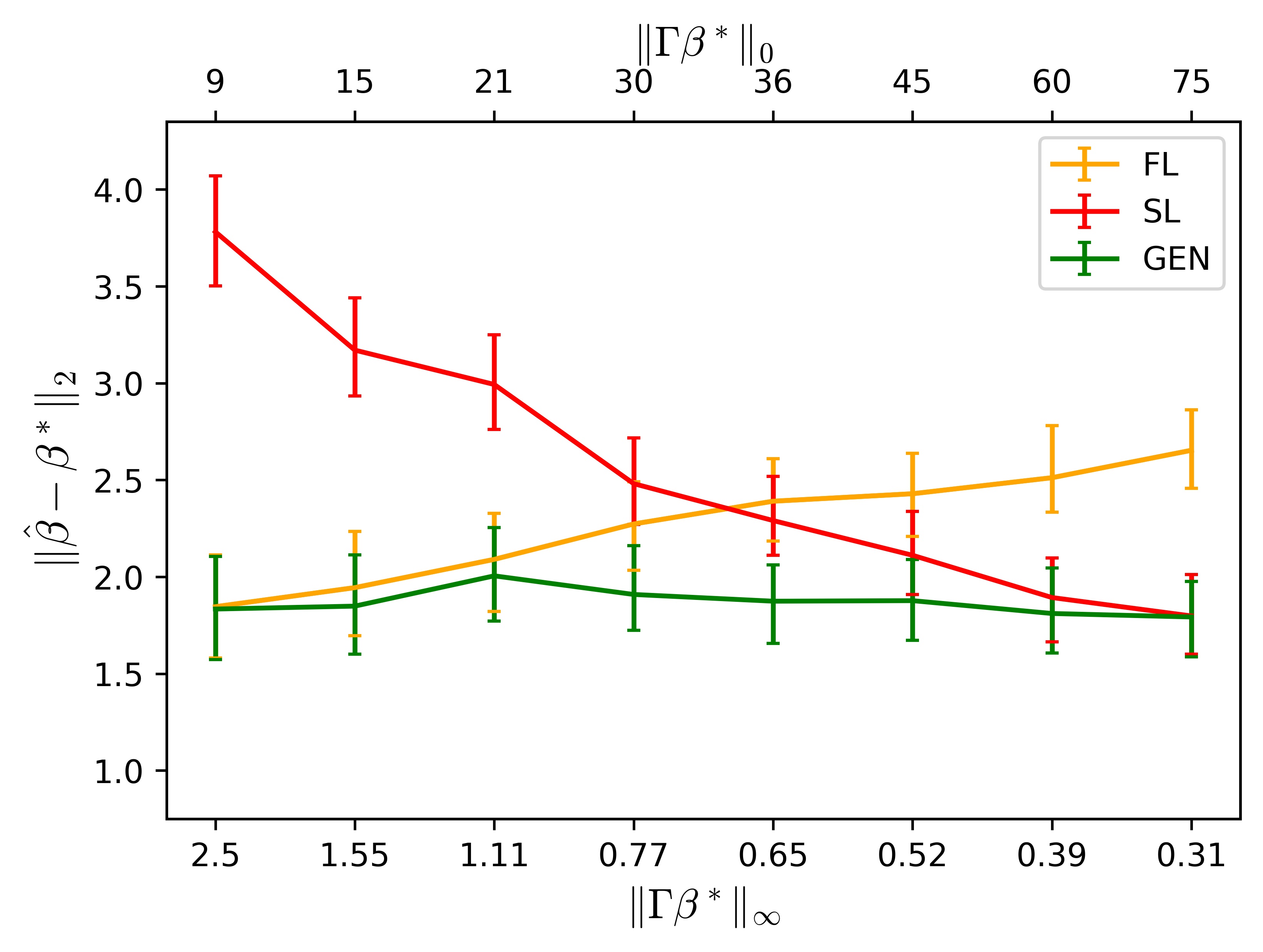}   
\end{subfigure}
   \hfill
\begin{subfigure}{0.49\textwidth}
    \centering
 \includegraphics[width=\textwidth]{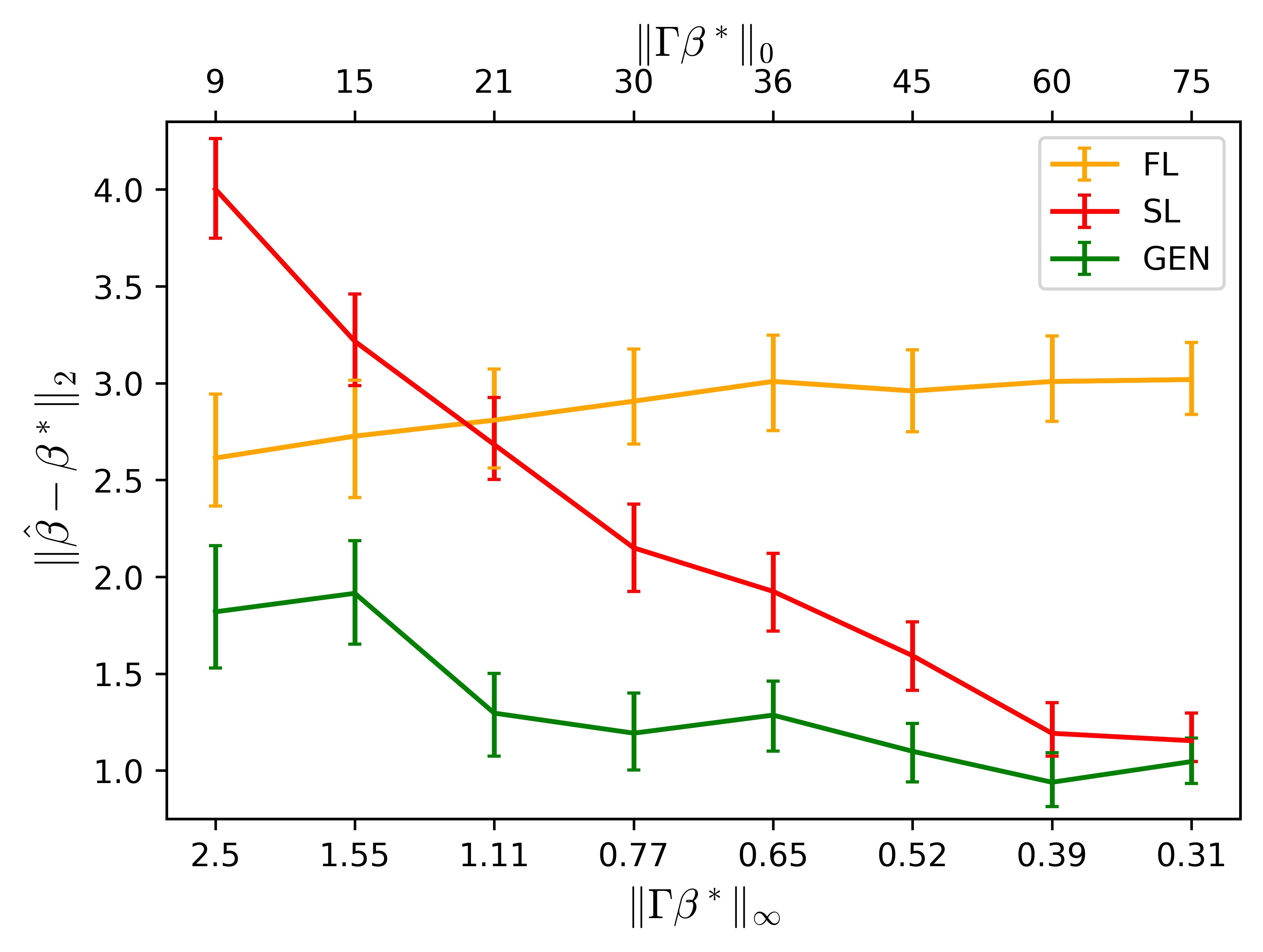}
 \end{subfigure}
    \caption{\textit{Side-by-side comparison of the estimation errors for the chain graph when $\Sigma$ is the identity matrix (left) and when $\Sigma$ has the Toeplitz structure with $\rho = 0.95$ (right). $\|\Gamma\beta^*\|_1$ is fixed at $15$. Note the greater divergence between the estimation errors of FL and GEN when there is higher correlation.}}
    \label{fig:toeplitz_rho}
\end{figure}   

\subsubsection{Performance comparisons as $n$ and $p$ vary} \label{sec:consistency}

This section examines the performance of GEN relative to all other methods as $n$ is fixed and $p$ increases, or as $p$ is fixed and $n$ increases. The covariance matrix $\Sigma$ is constructed as in Section \ref{sec:params_choice}, and the graphs we use are again the chain graph, the 2D grid and the barbell graph. The true signal $\beta^*$ is again not sparse, but contains a mix of piecewise constant regions and smoothly varying regions on the graph $G$ (similar to the true signals with intermediate values of $\|\Gamma\beta^*\|_\infty$ and $\|\Gamma\beta^*\|_0$ in Figure \ref{fig:1D_signals} and Figure \ref{fig:2D_signals}). Note that we do include the high-dimensional setting when $n$ is smaller than $p$. As described in Section \ref{sec: tuning}, hyperparameters for all methods are chosen based on the best prediction scores in cross-validation. We only report the estimation errors in Figure \ref{fig:consistency}, since the prediction errors for all methods show the same trends.  

\begin{figure}[H]
    \centering
\begin{subfigure}[b]{\textwidth}
    \centering
 \includegraphics[width=0.49\linewidth]{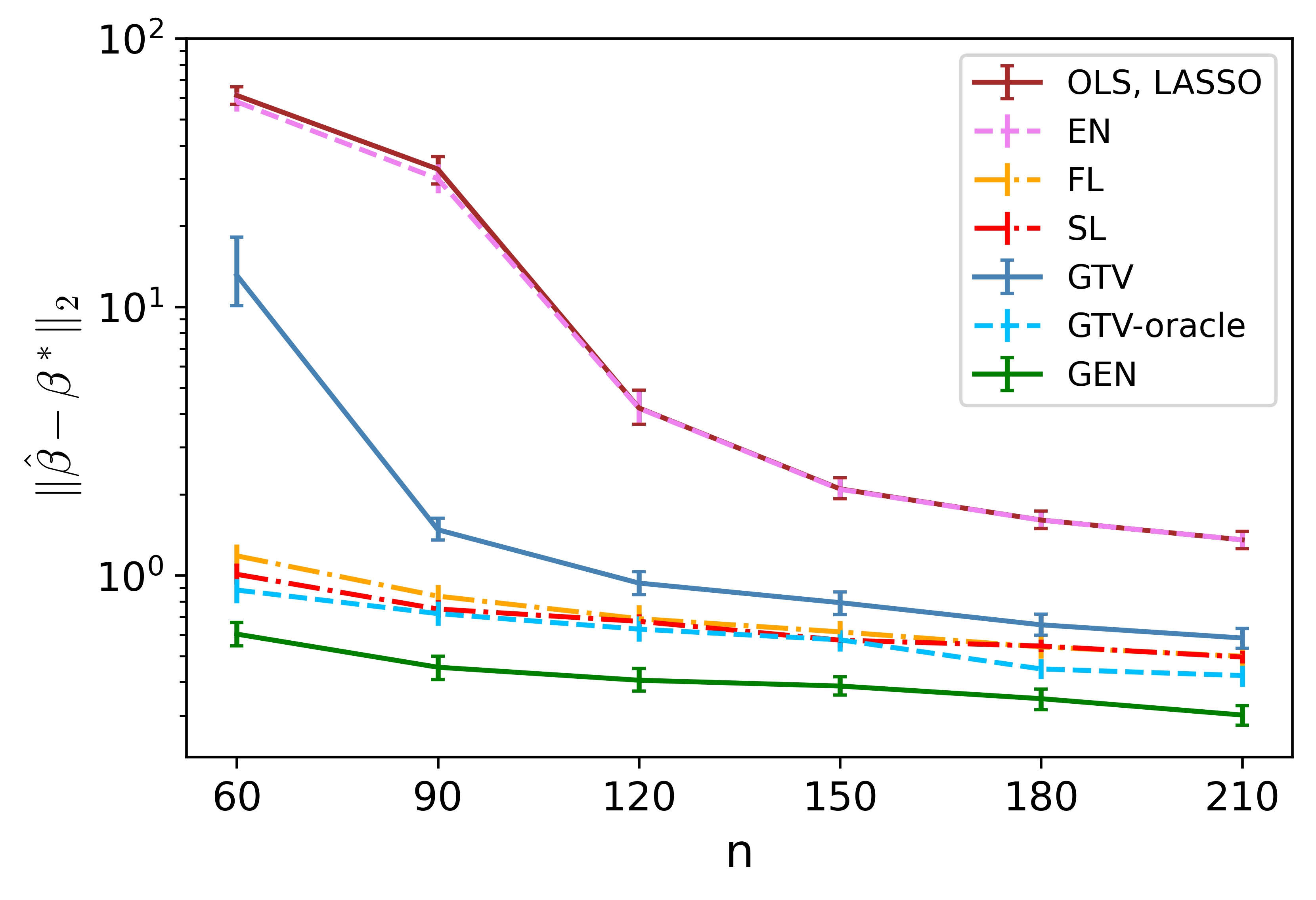}   
   \hfill
 \includegraphics[width=0.49\linewidth]{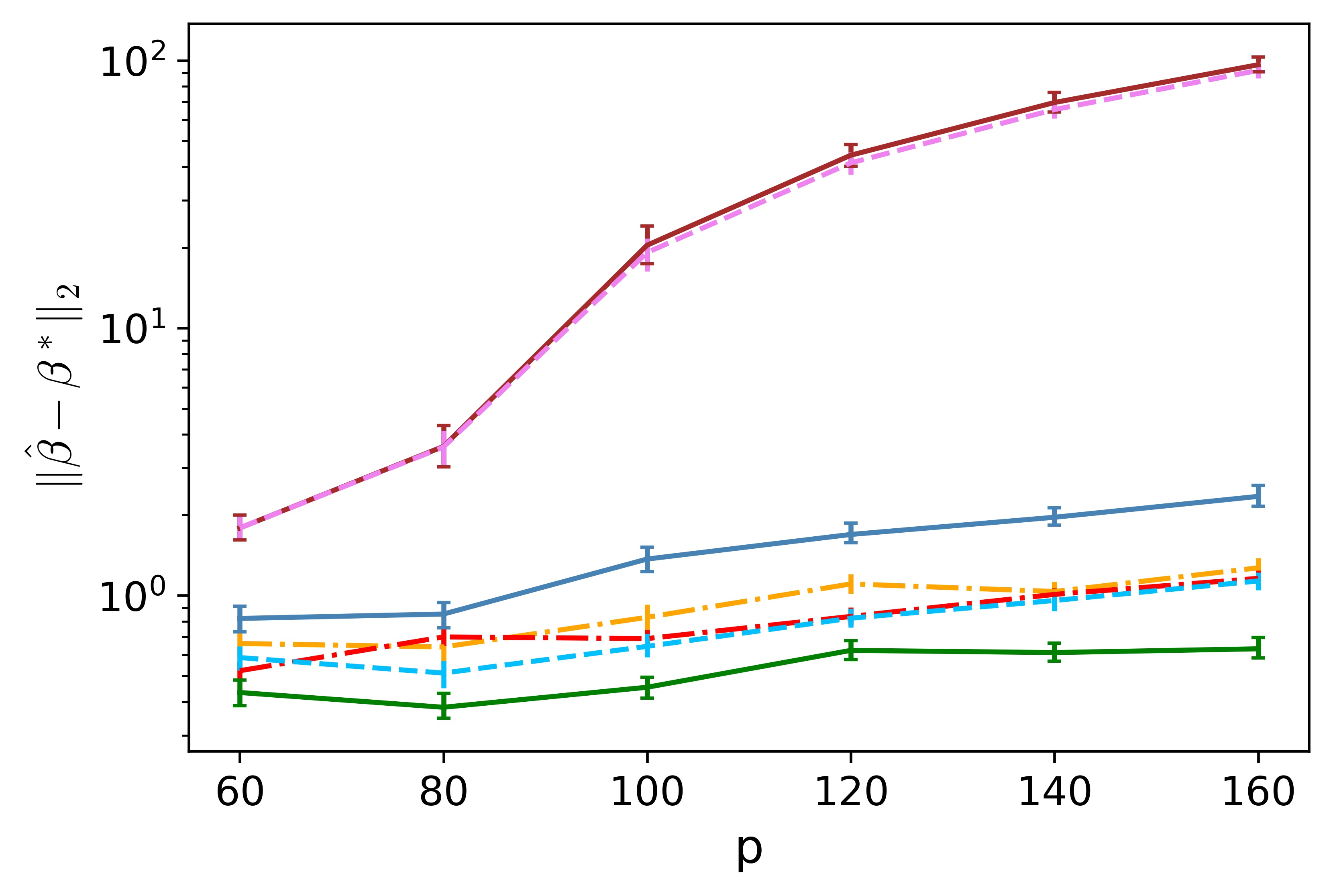}
 \caption{1D chain}
\end{subfigure}
\begin{subfigure}[b]{\textwidth}
    \centering
 \includegraphics[width=0.49\linewidth]{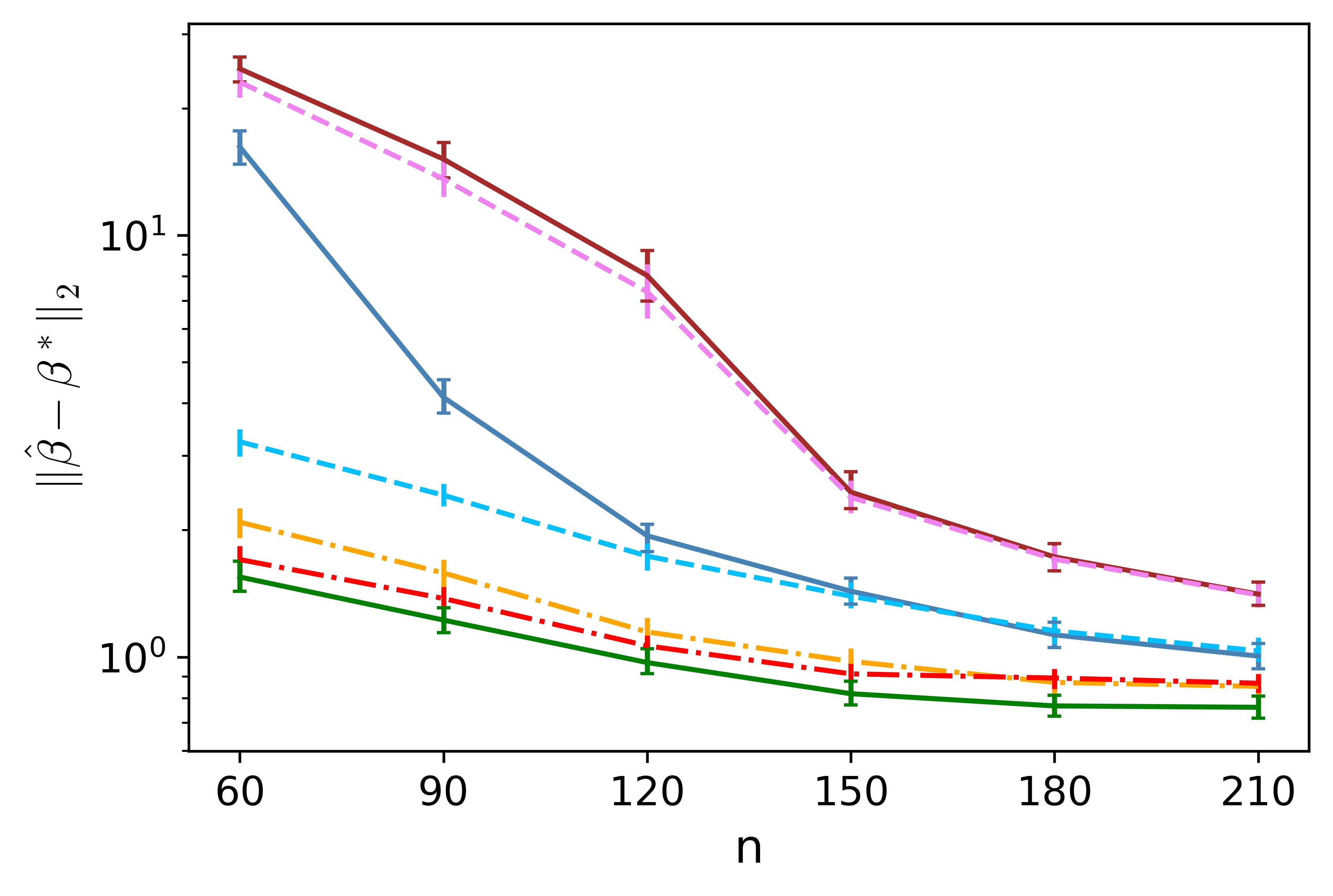}   
   \hfill
 \includegraphics[width=0.49\linewidth]{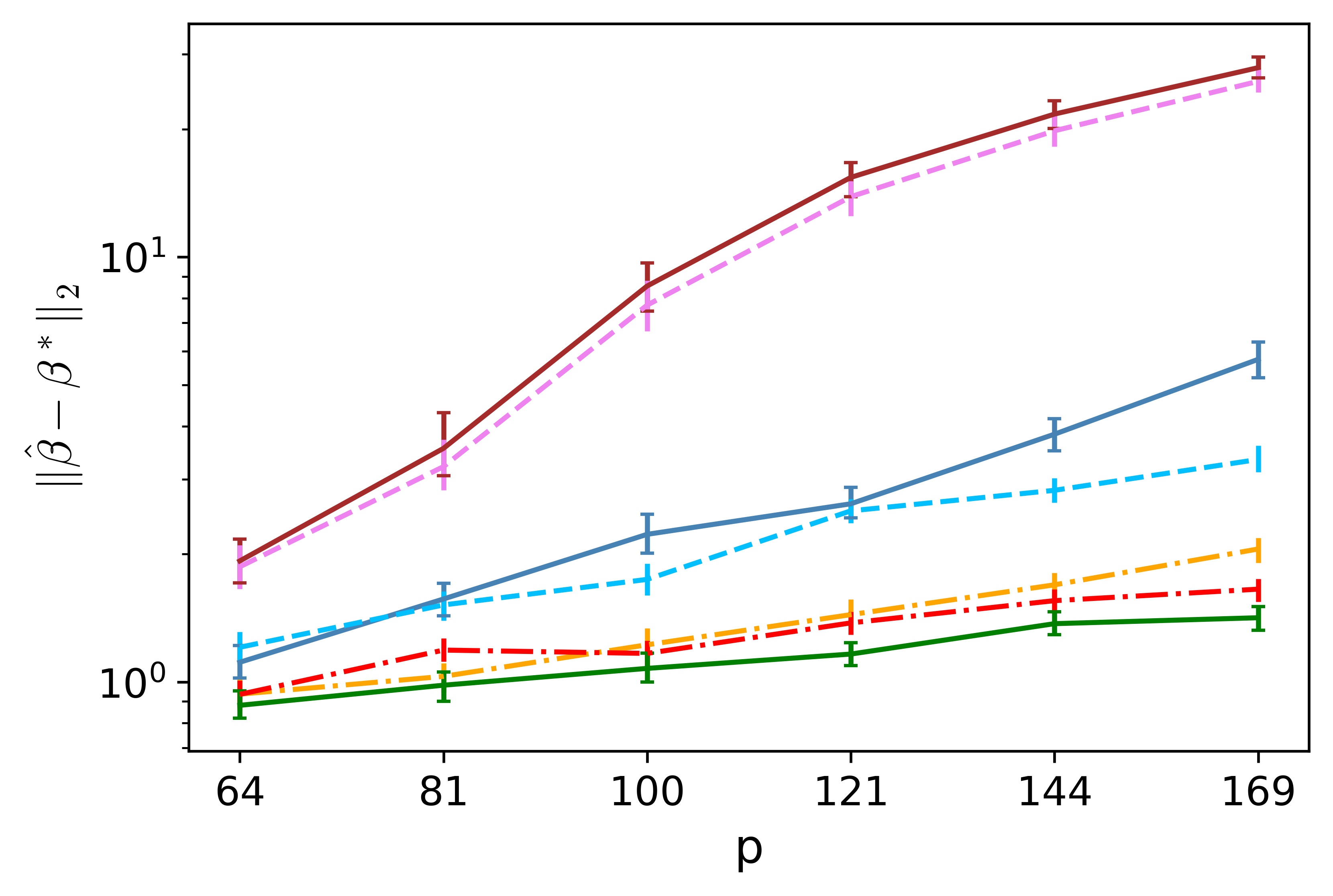}
 \caption{2D grid}
\end{subfigure}
\begin{subfigure}[b]{\textwidth}
    \centering
 \includegraphics[width=0.49\linewidth]{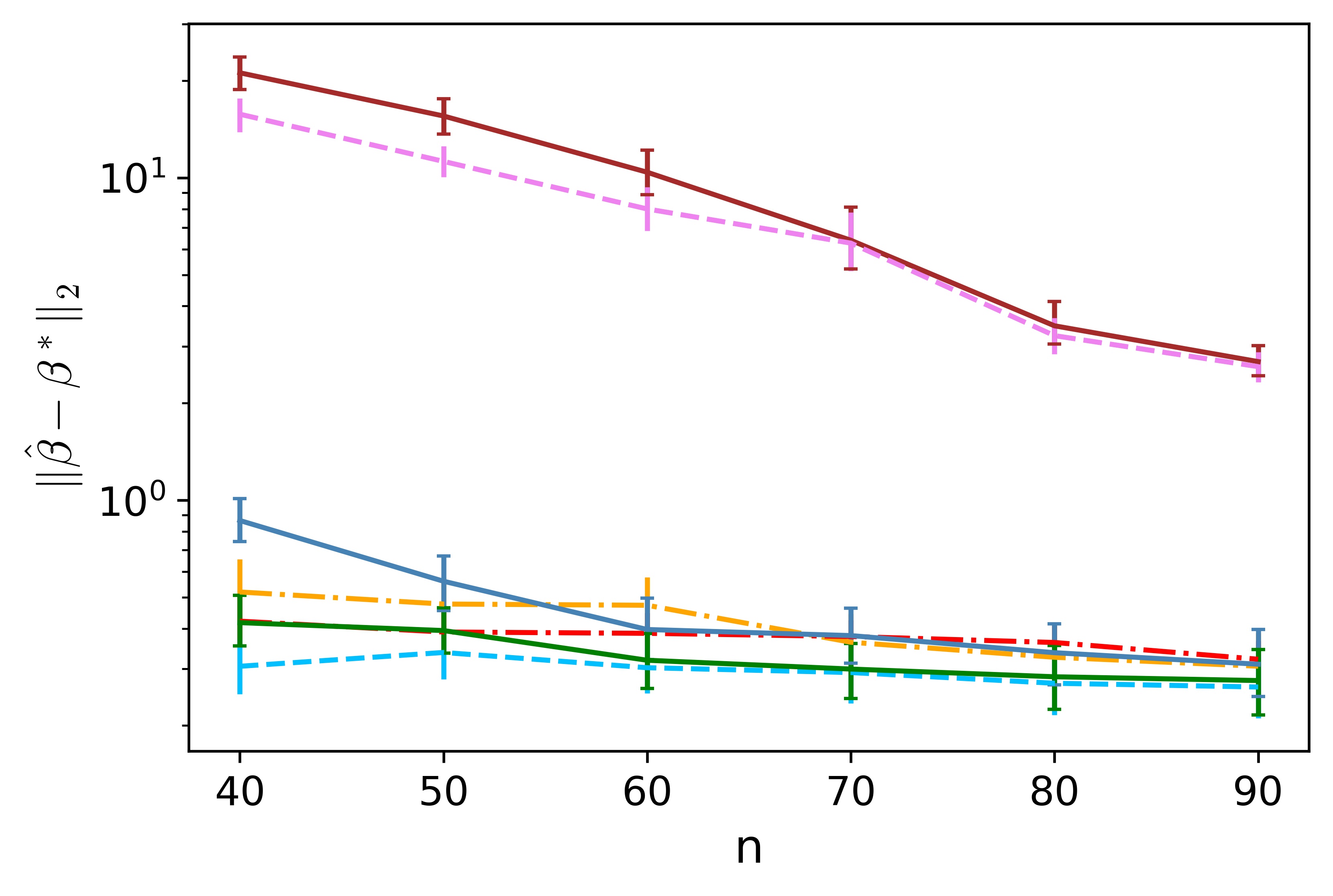}   
   \hfill
 \includegraphics[width=0.49\linewidth]{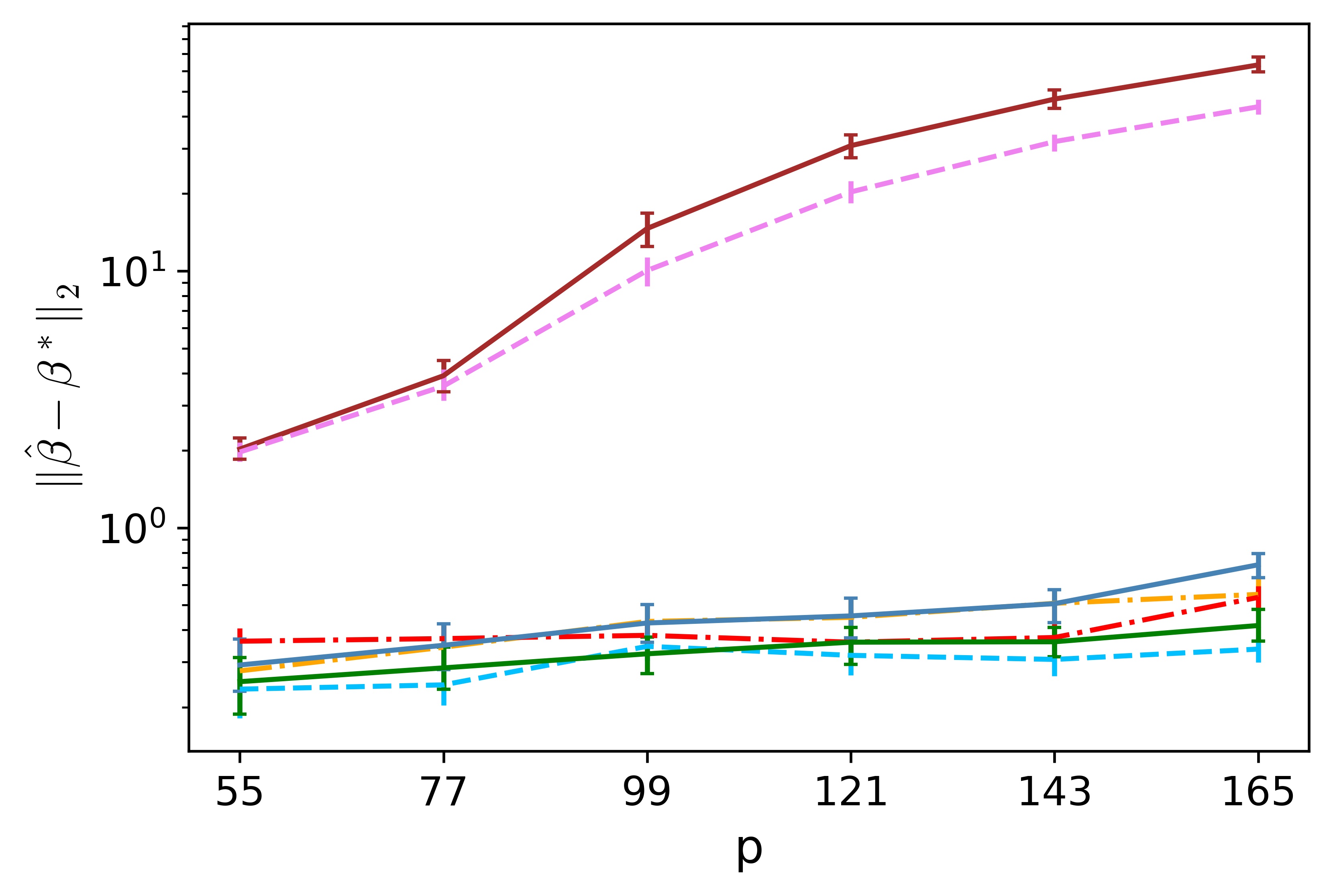} 
 \caption{Barbell}
\end{subfigure}
    \caption{\textit{
    Estimation errors (reported on the log scale) based on 500 resamplings for all estimators as $p$ is fixed ($p = 110$ for chain graph, $p=121$ for 2D grid, $p=66$ for barbell graph) but $n$ increases (left), and as $n=90$ is fixed but $p$ increases (right). $\sigma = 1$ is fixed, and in each plot $\|\Gamma\beta^*\|_\infty$ is kept roughly constant. CV yields $\lambda_L$ identically equal to zero for the Lasso estimator, and thus its performance coincides with that of OLS.}}
    \label{fig:consistency}
\end{figure}

As we can see from Figure \ref{fig:consistency}, GEN consistently has the best performance in terms of estimation errors (except for the barbell graph when GTV-oracle performs slightly better for some values of $(n,p)$). With regard to graph-independent methods, OLS and the Lasso clearly fail to perform well in our setting, and EN only provides limited improvements in terms of estimation errors. FL and SL, whose penalties do take into account the graph $G$, perform significantly better than the previous three methods but never better than GEN. The performance of GTV, whose penalty depends on the covariance estimate $\hat{\Sigma}$, is not consistent; it can be reasonable for the barbell graph but, for the other two graphs, is not very different from OLS, EN and the Lasso for certain values of $(n,p)$. Interestingly, the performance of GTV-oracle can surpass that of GEN for the barbell graph; this can be attributed to the fact that $\Sigma$ is constructed to reflect the graph structure and hence is a good estimate for the graph $G$ itself. The divergence between the estimation errors of GTV and GTV-oracle therefore suggests that the covariance estimation error is not negligible, especially when $p$ is small relative to $n$. Note that GTV requires much more time than other methods for hyperparameter selection and model training, as we have discussed in Section \ref{sec: tuning}.

\subsubsection{Performance comparisons when $\beta^*$ is both sparse and smooth over $G$}\label{sec:contiguous_zeros} In Section \ref{sec:adaptivity} and Section \ref{sec:consistency}, we have compared the performances of various estimators when $\beta^*$ is dense. We now consider the case when $G$ is the chain graph, and $\beta^*$ is sparse and has small variations in its successive entries, as illustrated in the left plot of Figure \ref{fig:zero_smooth}. Such a signal structure should be more favorable to either FL or SL, and we expect at least one of them to outperform GEN in this case. 

However, as can be seen in Table \ref{tab:zeros}, GEN still has the best performance compared to all other estimators. FL and SL perform better than the Lasso estimator, which in turn is better than OLS as expected. Certainly, such a strong performance relative to the other estimators may depend on the smoothness and sparsity levels of the true signal. Nonetheless, this example clearly demonstrates that effectively leveraging the true signal's smoothness over $G$ can be more important in reducing prediction and estimation errors than exploiting its sparsity structure. 

\begin{figure}[H]
\centering
\begin{subfigure}{0.49\textwidth}
\centering
\includegraphics[width = \textwidth]{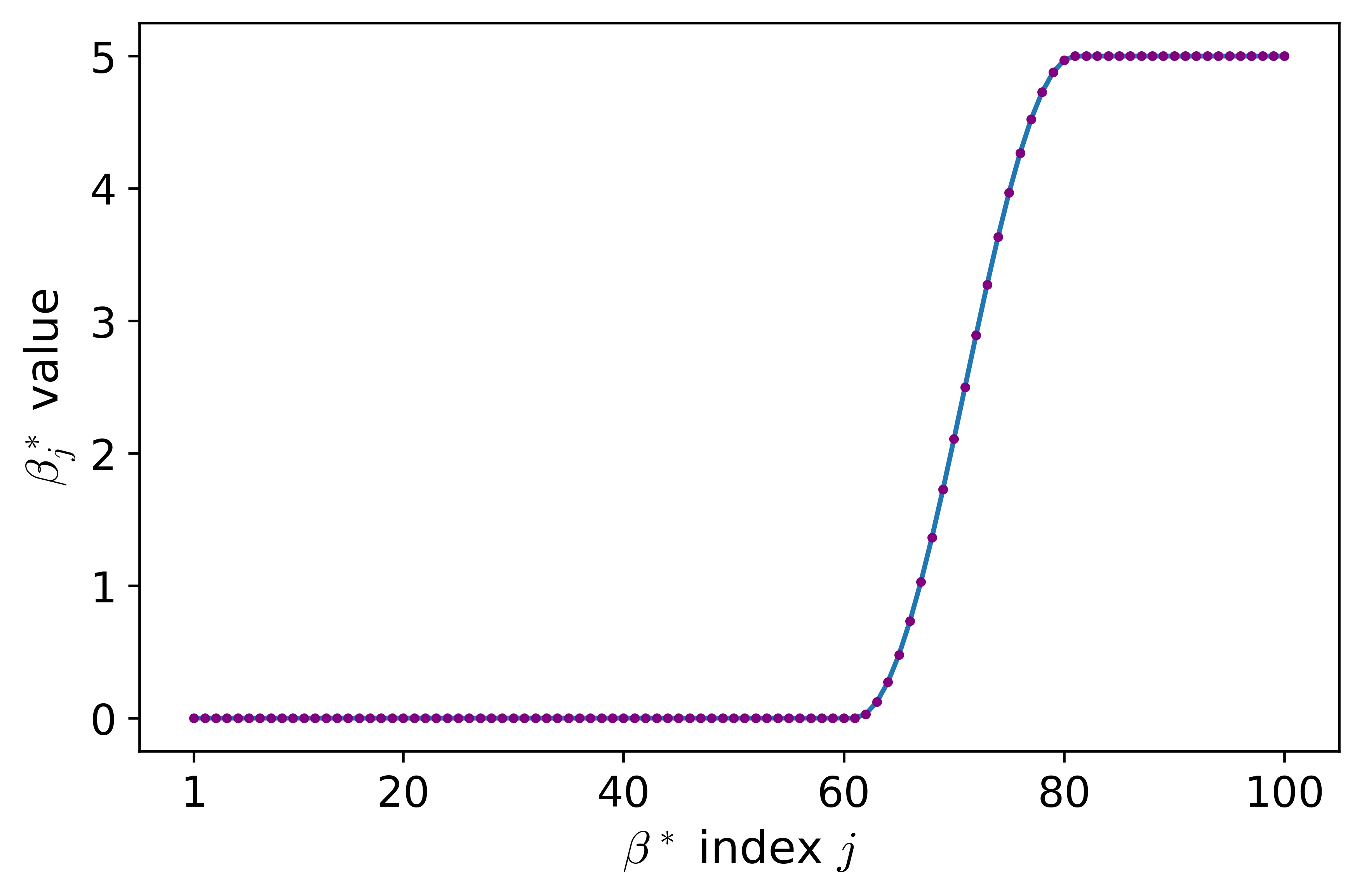}
\end{subfigure}
\hfill
\begin{subfigure}{0.49\textwidth}
\centering
\includegraphics[width = \textwidth]{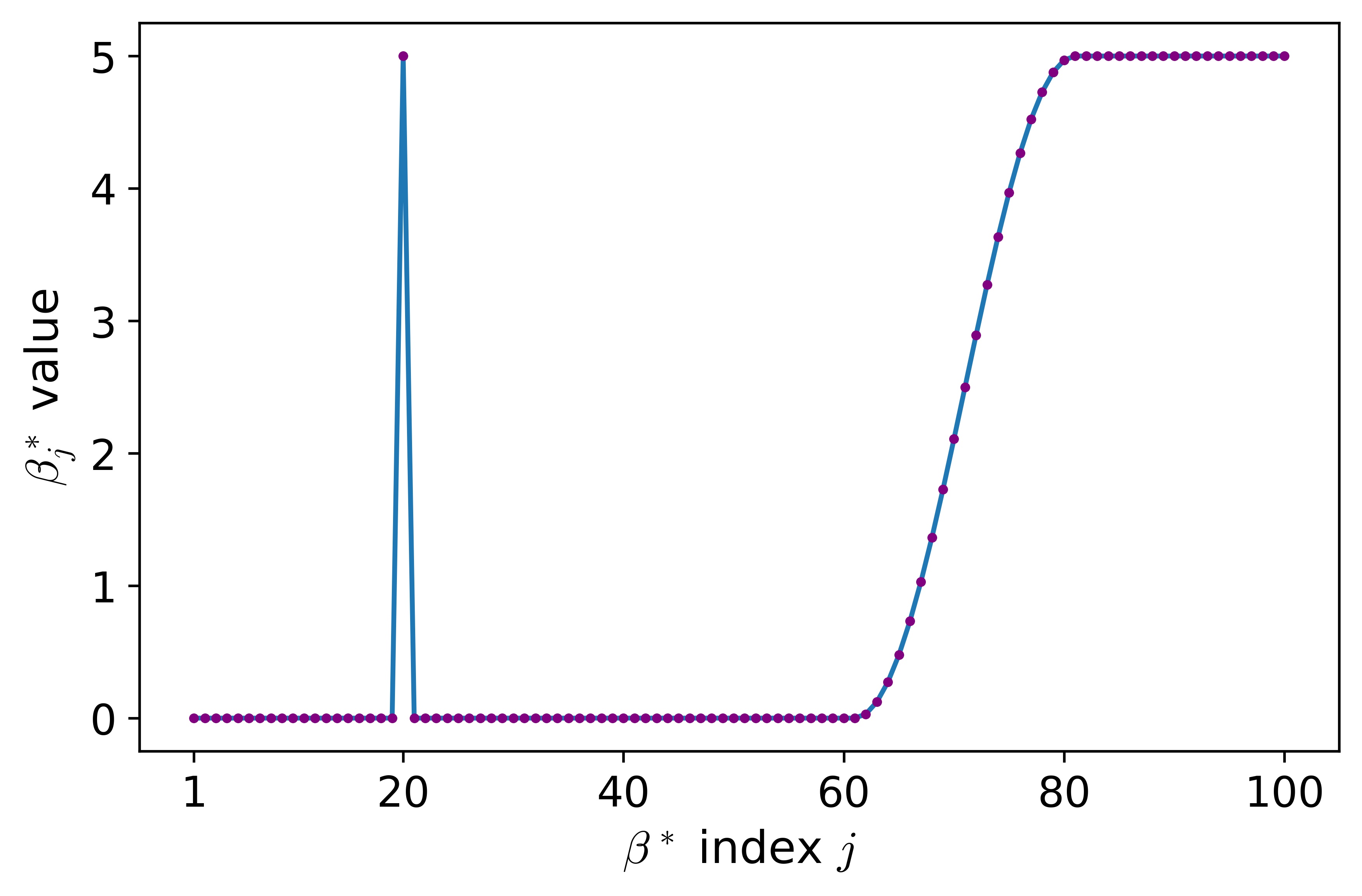}
\end{subfigure}
\caption{\textit{Left: Sparse and smooth signal with $p = 100$, $\|\beta^*\|_0 = 40$, $\|\Gamma\beta^*\|_\infty = 0.39$. Right: The left signal is modified to include a spike, so that $\|\Gamma\beta^*\|_\infty$ increases to 5. We use $\sigma = 1, n = 80$ and the Toeplitz covariance matrix with $\rho = 0.5$ for $\Sigma$ in this section.}}
\label{fig:zero_smooth}
\end{figure}

We also consider a slight modification to the previous example, so that we have a sharp spike in the zero region of the signal. Adding a single spike should not significantly change the radius $R_q$ of the $\ell_q$-ball to which $\Gamma\beta^*$ belongs. However, since $\|\Gamma\beta^*\|_\infty$ is now much larger, CV yields $\lambda_2$ identically equal to zero, and GEN degenerates into FL with $\lambda_L = 0$. As a result, FL performs better than GEN (and so does GTV-oracle), although the deterioration of GEN's performance is not drastic and GEN still performs better than EN, SL, GTV and the Lasso. It is therefore a question of interest for future research whether we can replace the $\ell_2$ component of GEN with another penalty that is more robust to signal spikes, while retaining the benefits of having the $\ell_2$ component as discussed in Section \ref{sec:adaptivity}.

\begin{table}[H]
    \caption{Prediction and estimation errors for the true signals in Figure \ref{fig:zero_smooth}; `L' and `R' denote errors for the left and right true signals respectively. The mean and standard deviation of the errors based on 500 resamplings are shown below. Errors better than GEN's errors are shown in orange.}
    \centering
        \begin{tabularx}{\textwidth}{|>{\centering\arraybackslash}X||>{\centering\arraybackslash}X|>{\centering\arraybackslash}X|>{\centering\arraybackslash}X|>{\centering\arraybackslash}X|>{\centering\arraybackslash}X|>{\centering\arraybackslash}X|>{\centering\arraybackslash}X|>{\centering\arraybackslash}X|}
\hline
 & \textbf{OLS} & \textbf{L} & \textbf{EN} & \textbf{FL} & \textbf{SL} & \textbf{GTV} & \textbf{GTV-oracle} & \textbf{GEN} \\
\hline \hline
L: Est. errors & $6.92 \pm 1.05$ & $1.98 \pm 0.36$ &$1.98 \pm 0.36$ & $0.56 \pm 0.08$ & $0.40 \pm 0.06$& $0.87 \pm 0.13$ & $0.38 \pm 0.06$ & $0.27 \pm 0.05$ \\
\hline
L: Pred. errors & $38.12 \pm 15.33$& $2.94 \pm 1.35$ & $2.94\pm 1.35$ & $0.28 \pm 0.12$ & $0.33 \pm 0.14$ & $0.72 \pm 0.26$ & $0.21\pm0.11 $& $0.15\pm 0.08$ \\
\hline \hline
R: Est. errors & $7.25 \pm 1.10$&$2.95\pm0.59$ &$2.95\pm0.59$ &\textcolor{orange}{$0.88\pm0.17$} &$1.75\pm0.30$ &$1.31\pm0.22$ &\textcolor{orange}{$0.75\pm 0.12$} &$0.97\pm0.18$ \\
\hline
R: Pred. errors & $41.75 \pm 15.95$ & $6.19 \pm 2.93$ & $6.19 \pm 2.93$ & \textcolor{orange}{$0.67\pm 0.26$} & $2.37 \pm 0.93$&$1.44 \pm 0.50$ &\textcolor{orange}{$0.49 \pm 0.19$} & $0.87 \pm 0.32$ \\
\hline
\end{tabularx}
    \label{tab:zeros}
\end{table}

\subsection{Empirical study of the quantity \texorpdfstring{$\gmin\left(\frac{1}{64}\Sigma+\lambda_2L\right)$}{gmin}} \label{sec:min_eig} In Section \ref{sec: theory}, if $\Sigma$ is ill-conditioned and we cannot assume $\gmin(\Sigma)$ is bounded away from zero, then we assume that $\gmin\left(\frac{1}{64}\Sigma+\lambda_2L\right)$ may be greater than $c\lambda_2$ or $c\sqrt{\lambda_2}$. These assumptions lead to the bounds \eqref{eq:lambda2} and \eqref{eq:lambda2-est}, which may allow for consistency in prediction and estimation respectively.  

\begin{figure}[H]
\centering
\begin{subfigure}{0.32\textwidth}
\centering
\includegraphics[width = \textwidth]{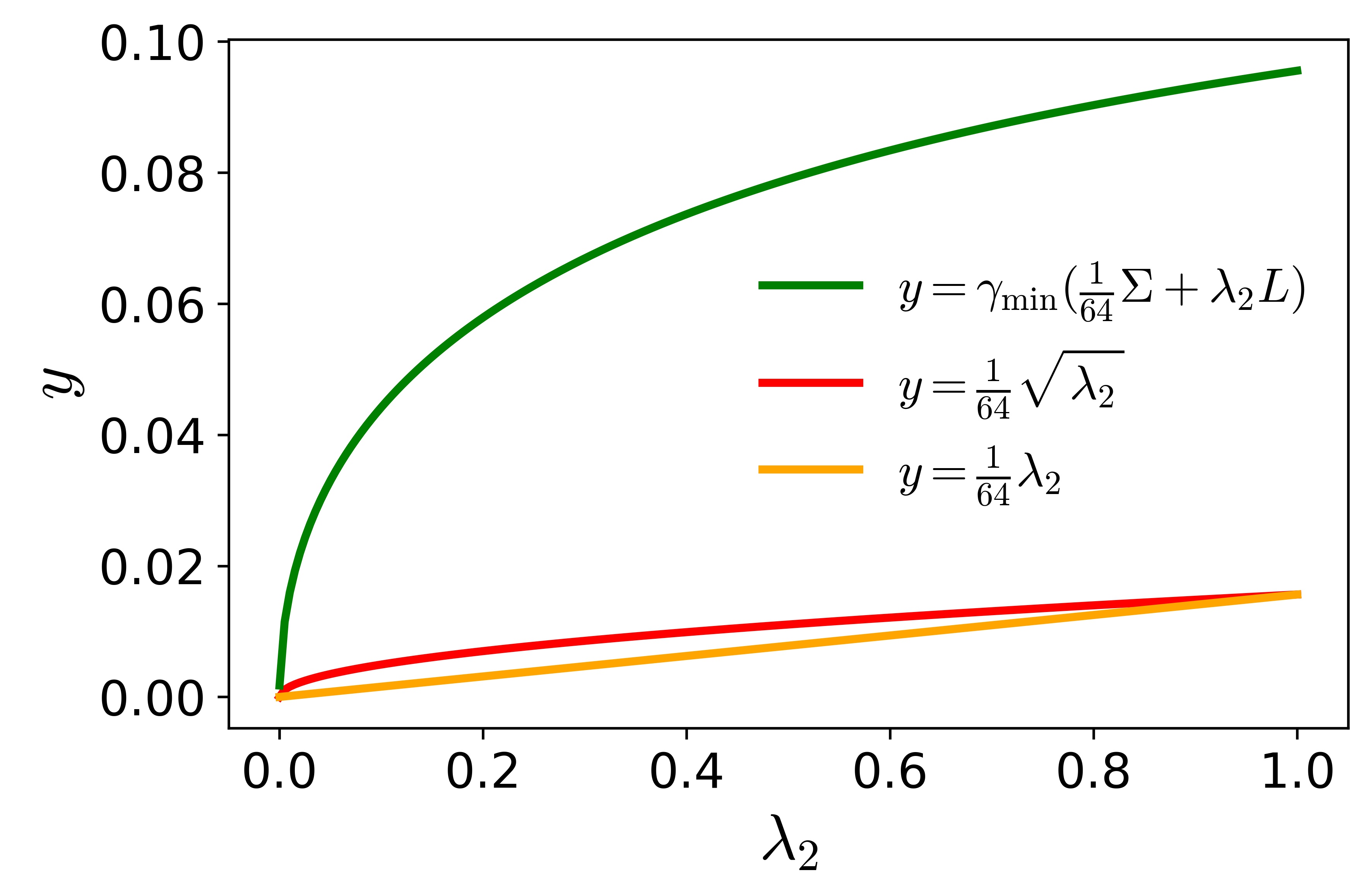}
\caption{Chain graph; $\rho = 0.8$}
\end{subfigure}
\hfill
\begin{subfigure}{0.32\textwidth}
\centering
\includegraphics[width = \textwidth]{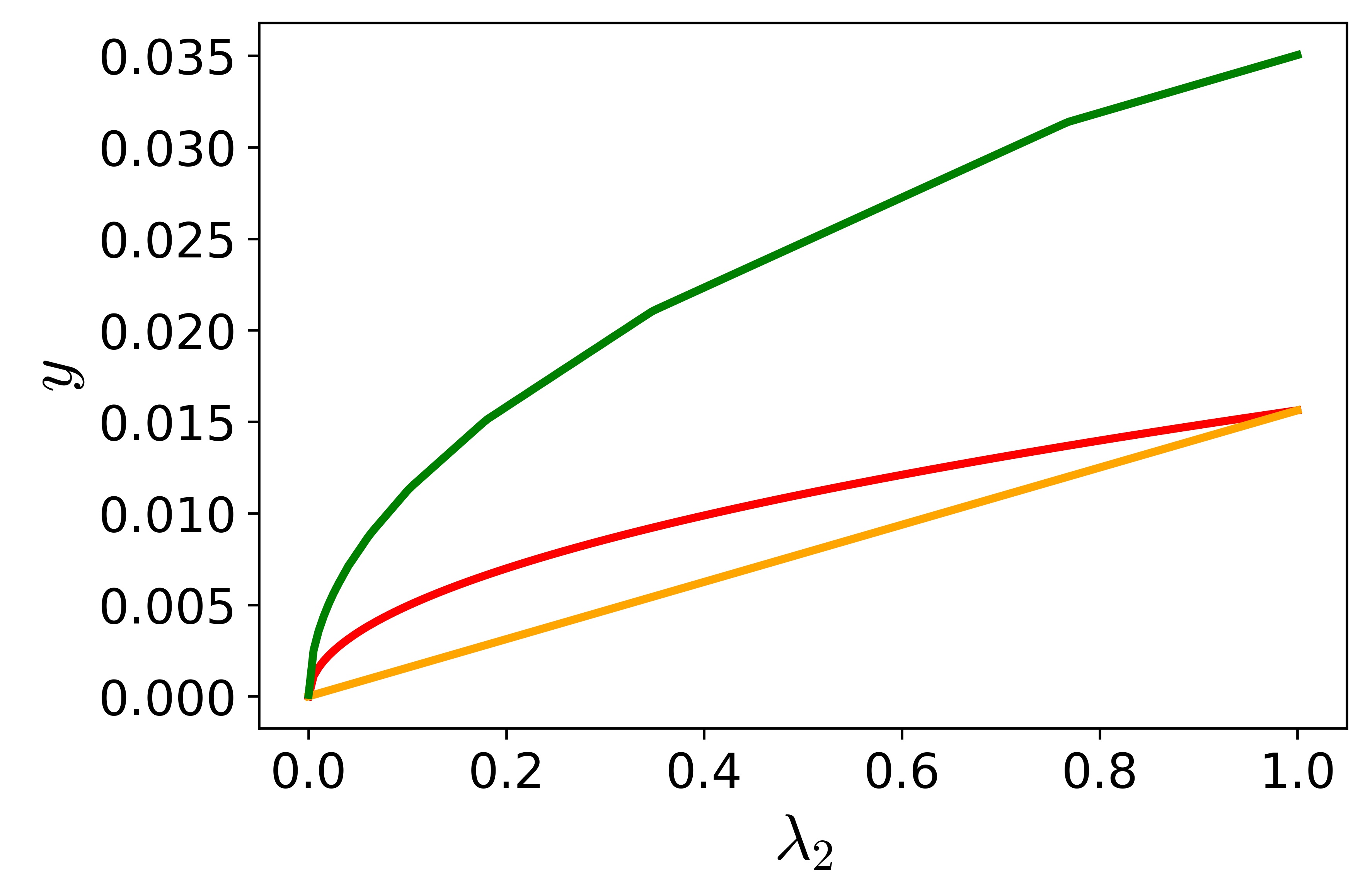}
\caption{Chain graph; $\rho = 0.99$}
\end{subfigure}
\hfill
\begin{subfigure}{0.32\textwidth}
\centering
\includegraphics[width = \textwidth]{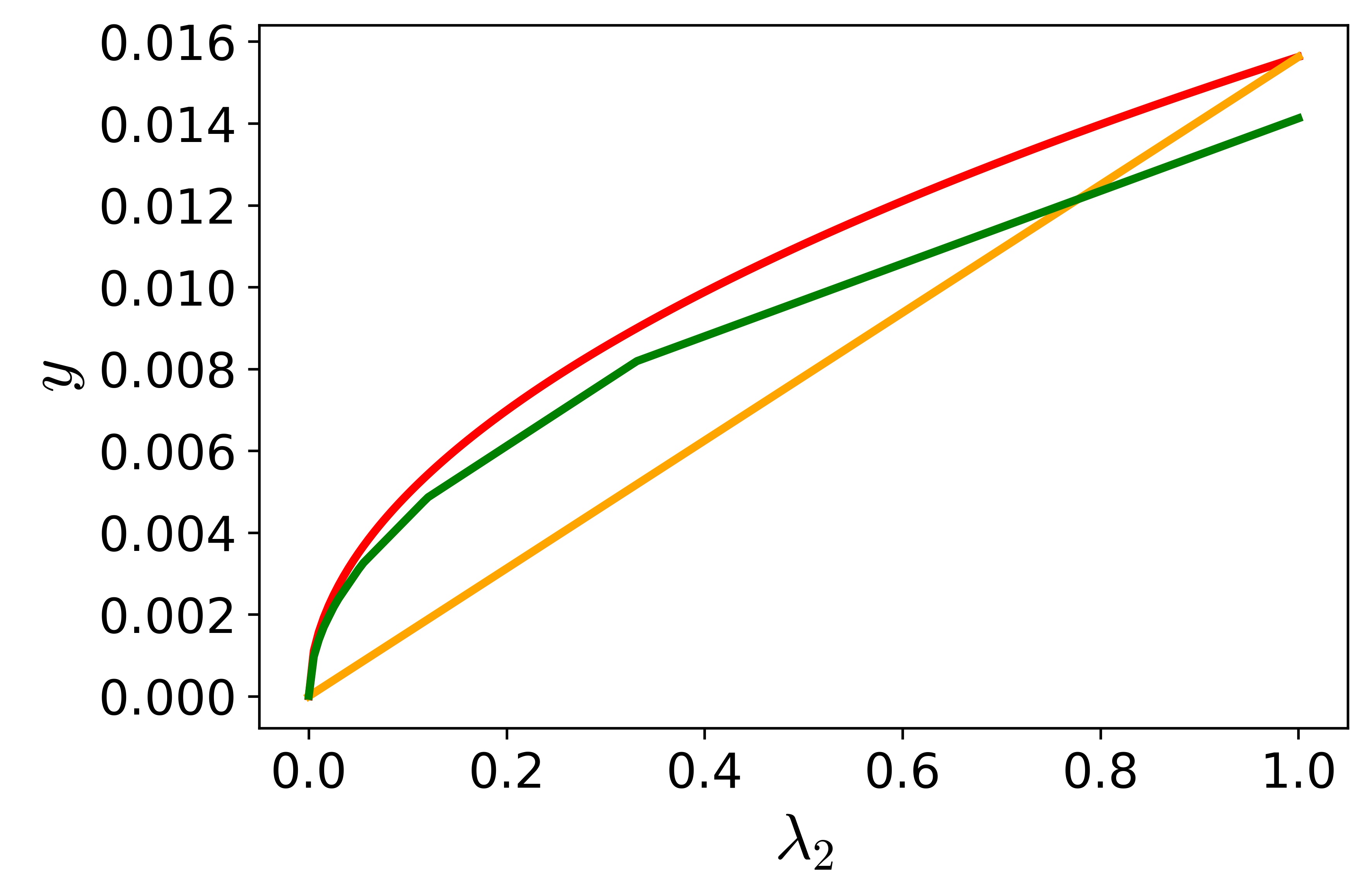}
\caption{Chain graph; $\rho = 0.9985$}
\end{subfigure}

\begin{subfigure}{0.32\textwidth}
\centering
\includegraphics[width = \textwidth]{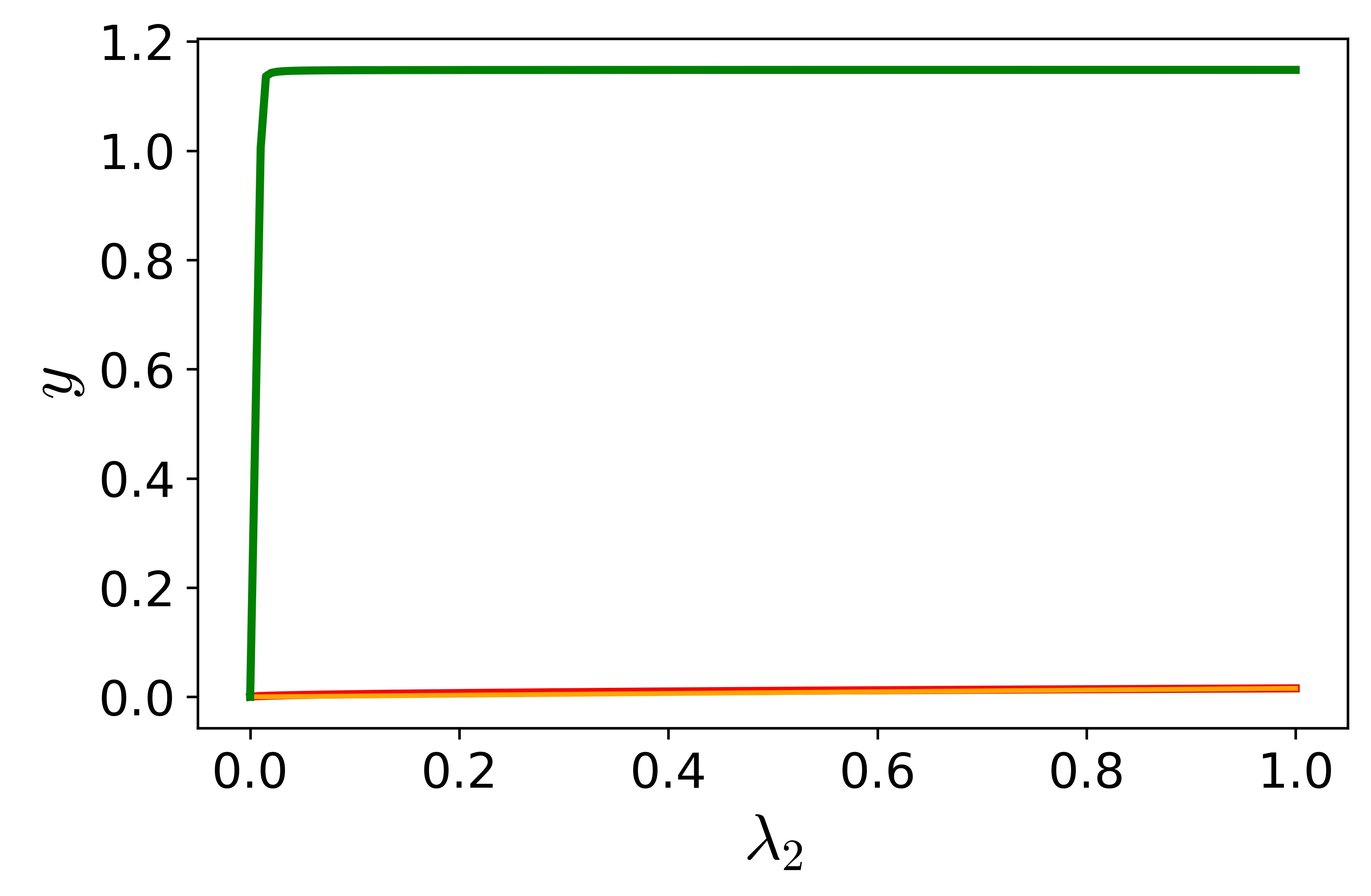}
\caption{Complete graph; $\rho = 0.99$}
\end{subfigure}
\hfill
\begin{subfigure}{0.32\textwidth}
\centering
\includegraphics[width = \textwidth]{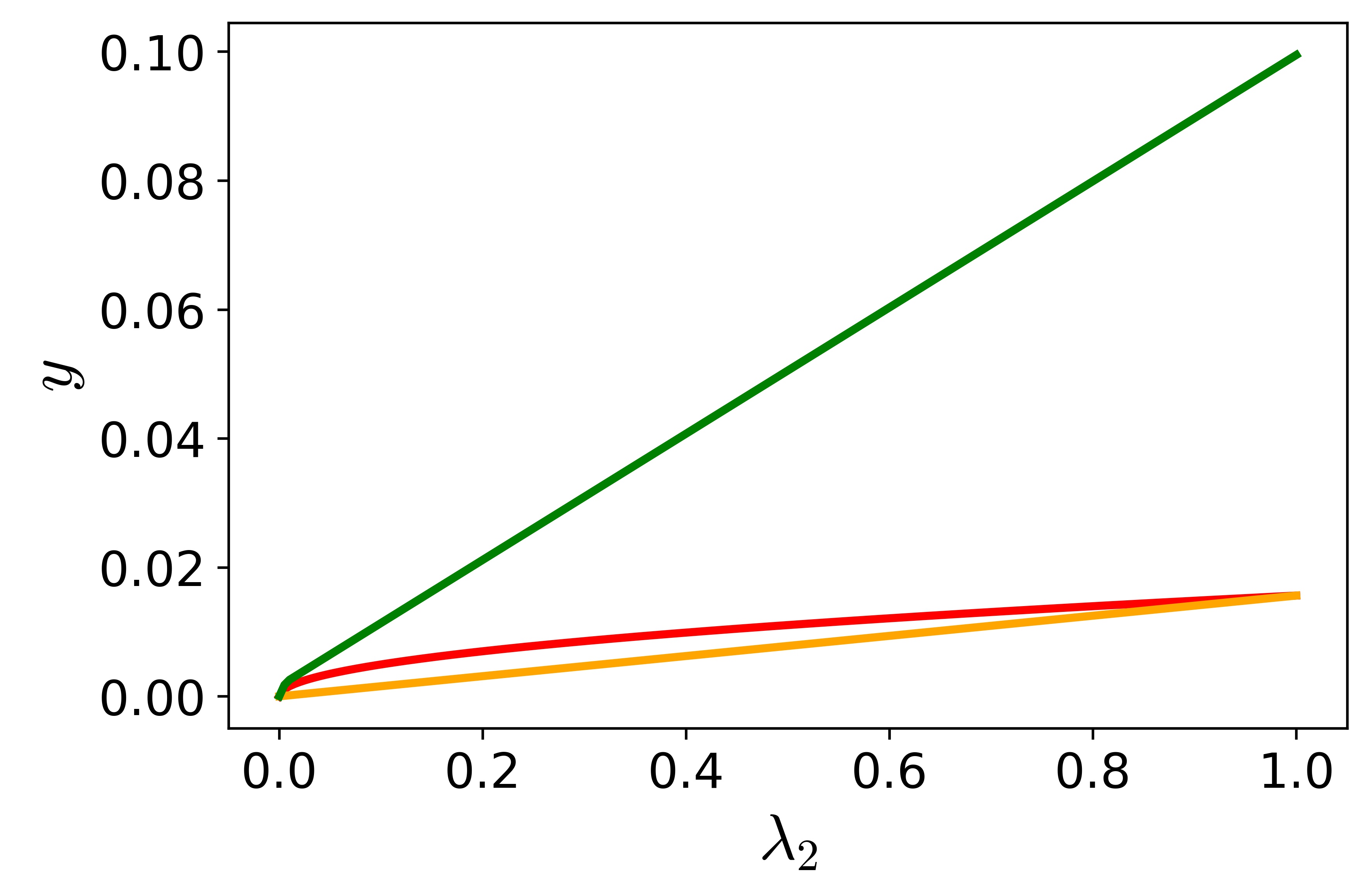}
\caption{2D grid; $\Sigma$ from $L + 10^{-4} I_p$}
\end{subfigure}
\hfill
\begin{subfigure}{0.32\textwidth}
\centering
\includegraphics[width = \textwidth]{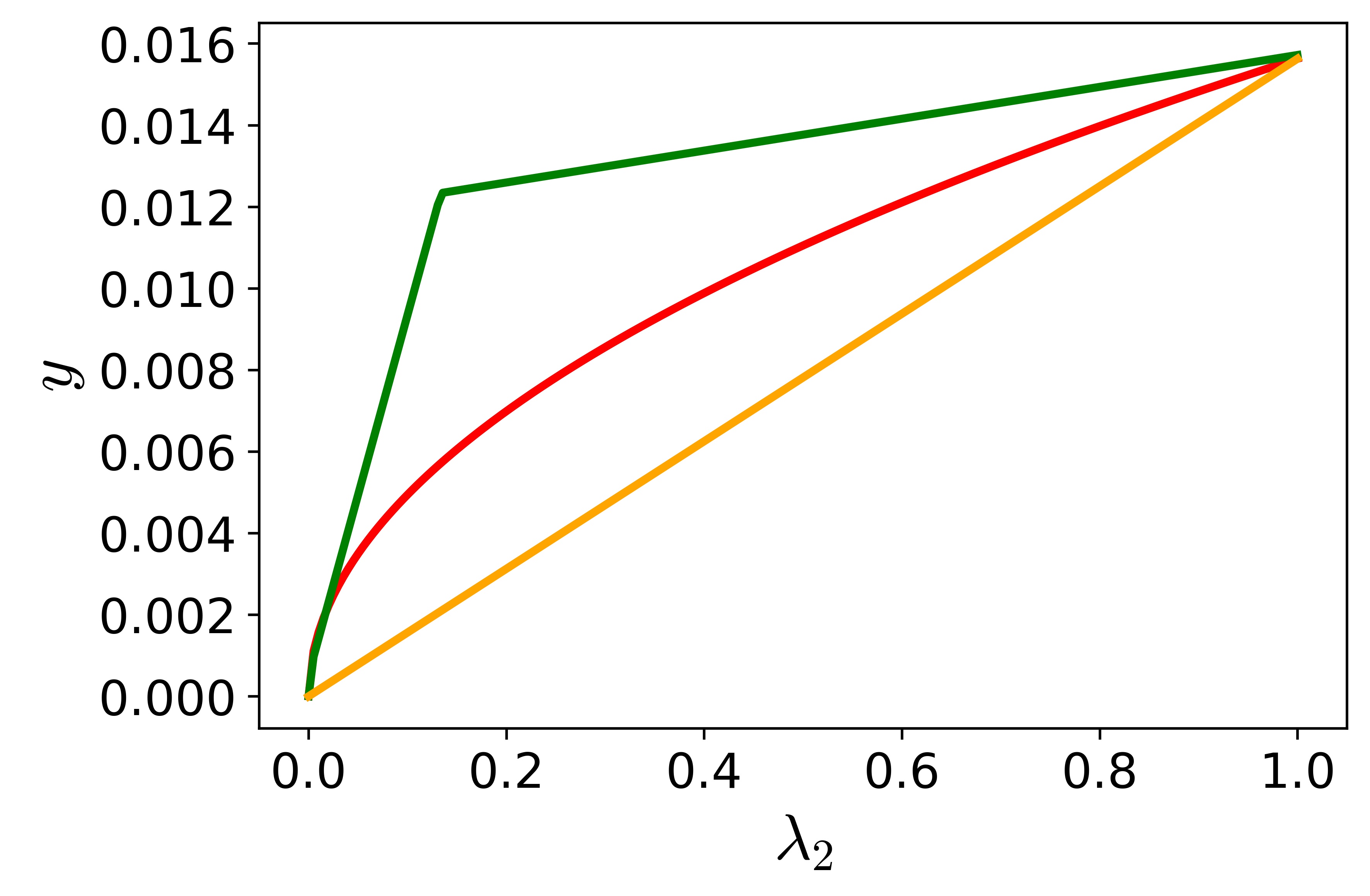}
\caption{Barbell; $\Sigma$ from $L + 3\times 10^{-4} I_p$}
\end{subfigure}
    \caption{\textit{Growth of $\gamma_{\min}(\frac{1}{64}\Sigma + \lambda_2 L)$ as a function of $\lambda_2$ for various choices of $\Sigma$ and $G$ ($p=100$ for all plots). In (a), (b) and (c), $G$ is the chain graph and $\Sigma$ has Toeplitz structure with varying $\rho$. In (d), $G$ is the complete graph and $\Sigma$ has Toeplitz structure. In (e) and (f), we use the 2D grid and barbell graph, with corresponding (and highly correlated) covariance structures as in Section \ref{sec:params_choice}.}} \label{fig:min_eig}
\end{figure}

We conjecture that $\gmin\left(\frac{1}{64}\Sigma+\lambda_2L\right) \geq \frac{1}{64}\lambda_2$ holds for all $\lambda_2 \in [0,1]$ under reasonable assumptions about $(\Sigma, L)$; this implies $\gmin\left(\frac{1}{64}\Sigma+\lambda_2L\right) \geq\frac{1}{64}\min(\lambda_2, 1)$. Figure \ref{fig:min_eig} shows the growth of the quantity $\gmin\left(\frac{1}{64}\Sigma+\lambda_2L\right)$ as a function of $\lambda_2$, for the various types of graphs and covariance matrices we have considered in Section \ref{sec: experiments}. When $G$ is the chain graph and $\Sigma$ has the Toeplitz structure, we generally have $\gmin\left(\frac{1}{64}\Sigma+\lambda_2L\right) \geq \frac{1}{64}\sqrt{\lambda_2}$ for all $\lambda_2\in [0,1]$ unless $\rho > 0.99$. When $G$ is the 2D grid or barbell graph and $\Sigma$ is constructed accordingly as in Section \ref{sec:params_choice}, we can also observe the same trends. Overall, when $\Sigma$ is ill-conditioned and $\lambda_2$ can be chosen to be sufficiently large, the $\ell_2$ component of the GEN penalty can significantly improve our error upper bounds.

\subsection{Real data analysis} In this section, we apply the GEN penalty to a number of real datasets. 

\subsubsection{COVID-19 trend prediction} 

We consider the problem of predicting the number of COVID-19 cases 14 days in advance for a given county in California, using a New York Times-curated \href{https://github.com/nytimes/covid-19-data}{COVID-19 dataset}. This problem may be of importance for hospitals and local authorities, as they may wish to anticipate potential spikes in COVID-19 cases based on current, local data. It is reasonable to assume that the number of cases $Y_{tc}$ on day $t$ in county $c$ is Poisson-distributed, as in \cite{agosto2020poisson}, \cite{bu2021likelihood} and \cite{cori2013new}. In this case, we can apply the variance-stabilizing Anscombe transform $x \mapsto 2 \sqrt{x + \frac{3}{8}}$ to form $\tilde{Y}_{tc}$. Following the modeling approach in \cite{cori2013new}, we may then consider the Gaussian model 
\begin{equation}\label{eq:covid-model1}
    \tilde{Y}_{tc} = \sum_{s=14}^{21}\alpha_s Y_{t-s,c} + \epsilon_{tc}
\end{equation}
where $\epsilon_{tc}\sim N(0, \sigma^2)$. In order to reduce temporal correlation between observations, the days $t$ are sampled such that consecutive time points are at least 7 days apart. We restrict our analysis (a) to the period from June 2020 to July 2021 to avoid non-stationary effects in the evolution of the pandemic due to the appearance of new virus strains, and (b) to the 25 densest counties in California where linear models are typically a better fit. As in \cite{ngonghala2022unraveling}, we use cross-validation to evaluate the accuracy of our model; 6/7 of our data is used for fitting and the remaining data is for performance evaluation (i.e. 2 months of data). Fitting an OLS model based on \eqref{eq:covid-model1} usually results in a satisfactory fit with an $R^2$ score above 0.8 (see the Appendix).

We hypothesize that for densely populated counties, rising cases in neighboring counties may further explain a significant fraction of the remaining variance in the data due to population movements between counties. To test this hypothesis, we consider a model that incorporates the number of cases from nearby counties within a two-hop radius of the given county $c$: 
\begin{equation}\label{eq:covid-model2}
    \tilde{Y}_{tc} = \sum_{k\in N_2(c)}\sum_{s=14}^{21}\alpha_{sk}Y_{t-s, k} + \epsilon_{tc}
\end{equation}

Fitting model \eqref{eq:covid-model2} is a high-dimensional problem, where the number of parameters $p$ can be up to 3 times the number of observations $n$, depending on the county. OLS therefore is not a suitable method for model \eqref{eq:covid-model2}. Consequently, in this experiment we fit the penalty-based methods (except GTV) based on \eqref{eq:covid-model2} and compare with the performance of OLS computed based on \eqref{eq:covid-model1}. The graph $G$ we consider here is such that two feature vectors are connected if they are indexed by the same day $t$ and by two adjacent counties, or if they are indexed by the same county $k$ and two consecutive time points.

We perform this prediction task for each of 25 most densely populated counties in California. For each county, we report the median root mean square error (RMSE) computed on the test set in Table \ref{tab:covid}. As expected, incorporating the numbers of cases in neighboring counties allows us to outperform OLS based on \eqref{eq:covid-model1} in 21 out of 25 counties. The graph-dependent methods FL, SL and GEN perform better than other methods in 19 out of 25 counties. Among these 19 counties, GEN has the best performance in 7 of them and is therefore a competitive candidate for this prediction task. The improvement it yields can be quite substantial; for the county Sutter in particular, GEN reduces the RMSE by 50\% compared to OLS and at least 25\% compared to FL and SL.

\begin{table}[ht]
\caption{Median RMSE achieved by various methods for 25 counties. OLS is fitted based on model \eqref{eq:covid-model1}, and all other methods are based on \eqref{eq:covid-model2}. The best performances for each county are highlighted in bold.}
\centering
\begin{tabular}{lcccccc}
  \hline
 {\bf County} & {\bf OLS} & {\bf L} &  {\bf EN} & {\bf FL} & {\bf SL} & {\bf GEN} \\ 
  \hline
    \hline
 Alameda & 1.08 & 1.18 & 1.14 & 0.96 & {\bf 0.85} & { 0.99} \\ 
 Butte & 3.20 & 1.81 & 1.83 & {\bf 1.46} & 1.88 & 2.10 \\ 
Contra Costa & {\bf 1.21} & 1.79 & 1.69 & 3.47 & 2.49 & 3.35 \\ 
Fresno & 8.22 & 7.06 & 9.77 & {\bf 5.25} & 7.71 & 5.95 \\
 Los Angeles & 4.92 & 6.92 & 7.34 & 5.28 & 6.06 & {\bf 4.56} \\ 
Marin & 6.11 & 3.92 & 5.49 & 5.47 & {\bf 3.38} & 4.18 \\ 
Merced & 7.94 & 9.75 & 9.03 & 9.26 & 9.68 & {\bf 6.08} \\
 Napa & 3.93 & 3.97 & 5.53 & 3.92 & 3.90 & {\bf 2.96} \\ 
Orange & {\bf  1.84} & 4.44 & 3.97 & 3.34 & 2.08 & 2.56 \\ 
 Placer & 2.20 & 1.35 & 1.89 & {\bf 1.27} & 1.55 & 1.53 \\ 
 Riverside & 3.46 & 3.32 & 3.62 & 3.21 & {\bf  2.79} & 3.73 \\ 
Sacramento & 2.23 & 3.11 & 2.53 & 3.61 & 2.31 & {\bf 1.66} \\ 
San Diego & 1.43 & 1.67 & 1.00 & 0.98 & {\bf 0.88} & 0.91 \\ 
San Francisco & 1.41 & 2.23 & {\bf 1.05} & 1.23 & 1.33 & 1.39 \\ 
San Joaquin & 3.64 &  {\bf 3.43} & {\bf 3.43} & 3.93 & 5.24 & 5.41 \\ 
San Mateo &  {\bf 1.44} & 2.34 & 2.45 & 1.68 & 1.56 & 1.75 \\ 
 Santa Barbara & 2.84 & 2.02 & 2.02 & 2.02 &{\bf  2.01} & 3.71 \\
 Santa Clara & 1.14 & 2.04 & 1.85 & 1.05 & 1.10 & {\bf 1.03} \\ 
 Santa Cruz & 6.56 & 3.86 & 4.62 & {\bf 3.55} & 4.17 & 4.59 \\ 
 Solano & 2.10 & 3.86 & 3.72 & {\bf 2.03} & 2.93 & 2.62 \\ 
 Sonoma & {\bf 1.74} & 3.47 & 3.60 & 2.78 & 2.62 & 2.67 \\ 
 Stanislaus & 9.29 & 6.41 & 9.17 & {\bf 4.55} & 4.76 & 4.88 \\
 Sutter & 4.07 & 7.51 & 7.51 & 4.33 & 2.58 & {\bf 1.94} \\
 Ventura & 2.02 & 1.22 & 1.20 & 1.23 & {\bf 1.16} & 1.36 \\ 
Yolo & 3.43 & 5.13 & 4.54 & 1.93 & 2.45 & {\bf 1.79} \\ 
   \hline
\end{tabular}
\label{tab:covid}
\end{table}

\subsubsection{Detection of Alzheimer's disease}

We test the GEN penalty's performance in detecting Alzheimer's disease, using an \href{https://www.kaggle.com/datasets/sachinkumar413/alzheimer-mri-dataset}{MRI dataset} available on Kaggle. The task is to classify whether the MRI images in the dataset show signs of dementia. Since the responses are binary, we need to consider the logistic extension \eqref{eq:6} of our method as well as that of all other methods. The original dataset has images labeled with moderate, mild, very mild and no dementia, but we exclude the moderate cases due to the small number of training samples. We also exclude the very mild cases since the images may be too similar to those with no dementia, thus leading to lower prediction accuracy for all methods. 

Since the features are 2D MRI images, it is natural to use the 2D grid graph as our graph $G$, which is of size $p = 32 \times 32 = 1024$ (we compress the original images to this size for computational convenience). We use the first 800 images with no dementia and 400 images with mild dementia in the original dataset. Out of these 1200 images, $n=480$ images are used as training data (note that $n < p$), 480 images are use for hyperparameter tuning, and the other 240 images constitute our testing data. For computation, we use ECOS for all methods. Since GTV requires at least a 3D grid search for hyperparameter tuning, it is too slow to be considered for this experiment. All other methods take at most 5 seconds of training time. 

The classification accuracies for all methods except GTV are reported in Table \ref{tab:real}. As expected, GEN shows better prediction performance than all other methods in consideration.  

\begin{table}[H]
    \caption{Prediction accuracies for classification of Alzhemer's disease status. Here, OLS is replaced by logistic regression (LR), and the logistic extensions of all penalty-based methods (except GTV) are used.}
    \centering
    \begin{tabularx}{\textwidth}{X X X X X X X}
    \hline
        & \textbf{LR} & \textbf{L} & \textbf{EN} & \textbf{FL}& \textbf{SL} & \textbf{GEN}\\[1mm]
    \hline \hline 
    \textbf{Accuracy} & 82.08\% & 90.0 \% & 92.50\% & 91.25\% & 92.08\% & 92.92\% \\[1mm]
    \hline
    \end{tabularx}
    \label{tab:real}
\end{table}

\subsubsection{Estimation of crime patterns in Chicago} 
Consider the task of uncovering crime trends over time across the 77 communities of Chicago (which we denote by the set $\mathcal{C}$). Statistics on the number of crimes per community between 2004 and 2022 are available on the city's \href{https://data.cityofchicago.org/}{data portal}. The monthly crime rates (which are defined here as the number of crimes per 100,000 inhabitants) vary over the years and across the communities, and they are also subject to significant seasonal effects. Additional details on the nature of the data and preprocessing are provided in the Appendix.

Let $Y_{my}^{(c)}$ denote the crime rate for community $c\in \mathcal{C}$, month $m$ and year $y$. Since we are working with count data, it is reasonable to pre-process the data by applying the Anscombe transform to $Y_{my}^{(c)}$ to form $\tilde{Y}_{my}^{(c)}$. We then consider the following additive Gaussian model 
\begin{equation}\label{eq:crime}
 \tilde{Y}_{my}^{(c)} = \sum_{i=1}^{12} \alpha_{i} \mathbf{1}[m=i]  +\sum_{j=2004}^{2022} \beta_{j}  \mathbf{1}[y=j]   +\sum_{c \in \mathcal{C}} \gamma_{k}\mathbf{1}[k=c] + \epsilon_{my}^{(c)}  
\end{equation}
where $\epsilon_{my}^{(c)} \sim N(0,\sigma^2)$. While our design matrix here is not equal to identity as in the trend filtering case, note that it contains ``one-hot'' encodings rather than i.i.d. rows from some distribution $\mathbb{P}$. The parameters $(\alpha, \beta)$ naturally exhibit \textit{temporal smoothness}, since we expect them to vary smoothly over time. The community offset parameter $\gamma$, on the other hand, should exhibit \textit{spatial smoothness}, as we expect neighboring communities to have similar offsets. To define our GEN penalty, we encode these prior beliefs in a regularizing graph $G$ with 3 disconnected components: one chain graph reinforcing the temporal smoothness of the month coefficients, another chain graph for that of the years, and a third component encoding neighborhood adjacency. We compare our method's performance with all other methods except GTV. 

Figure \ref{fig:rmse_chicago} compares the prediction performance (reported using RMSE computed on held-out data across 40 independent trials) for all methods. Here, performance is assessed for different data regimes: while the original dataset contains 17,094 observations, we use a fraction $\alpha \in \{0.5\%, 1\%, 2\%, 5\%\}$ of data for estimation of the $p=108$ parameters in our model ($\alpha = 0.5\%$ and $\alpha = 1\%$ correspond to $p > n$ and $p \approx n$ respectively). As shown in Figure \ref{fig:rmse_chicago}, GEN performs consistently better than all other methods, especially in the data-sparse regime. While in this example we are more interested in the estimation of crime patterns rather than prediction (note that the model \eqref{eq:crime} cannot be used to predict crime rates beyond 2022), Figure \ref{fig:rmse_chicago} provides evidence for GEN's superior performance and can be of interest if we are given a dataset with many missing values that require data imputation.

\begin{figure}[H]
    \centering
    \includegraphics[width=\textwidth]{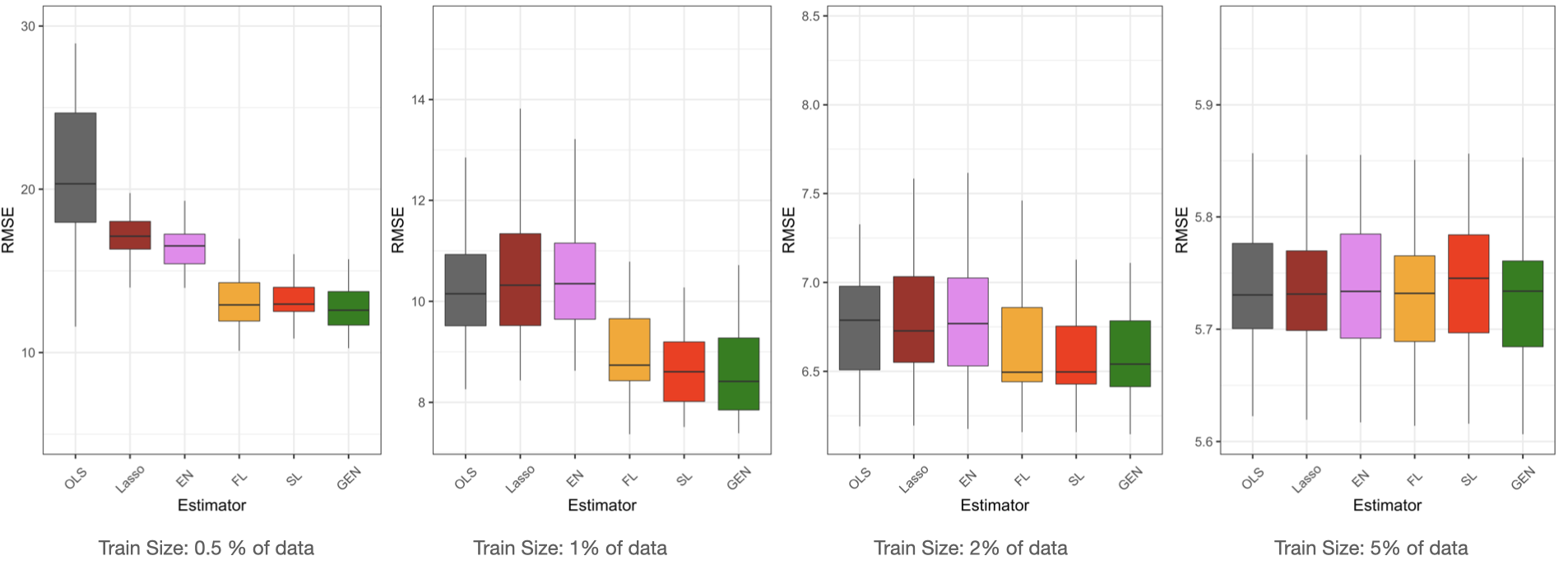}
    \caption{\textit{RMSE achieved by different estimators as the proportion $\alpha$ of data used for training varies.}}
    \label{fig:rmse_chicago}
\end{figure}

\begin{figure}[H]
            \includegraphics[width=\textwidth]{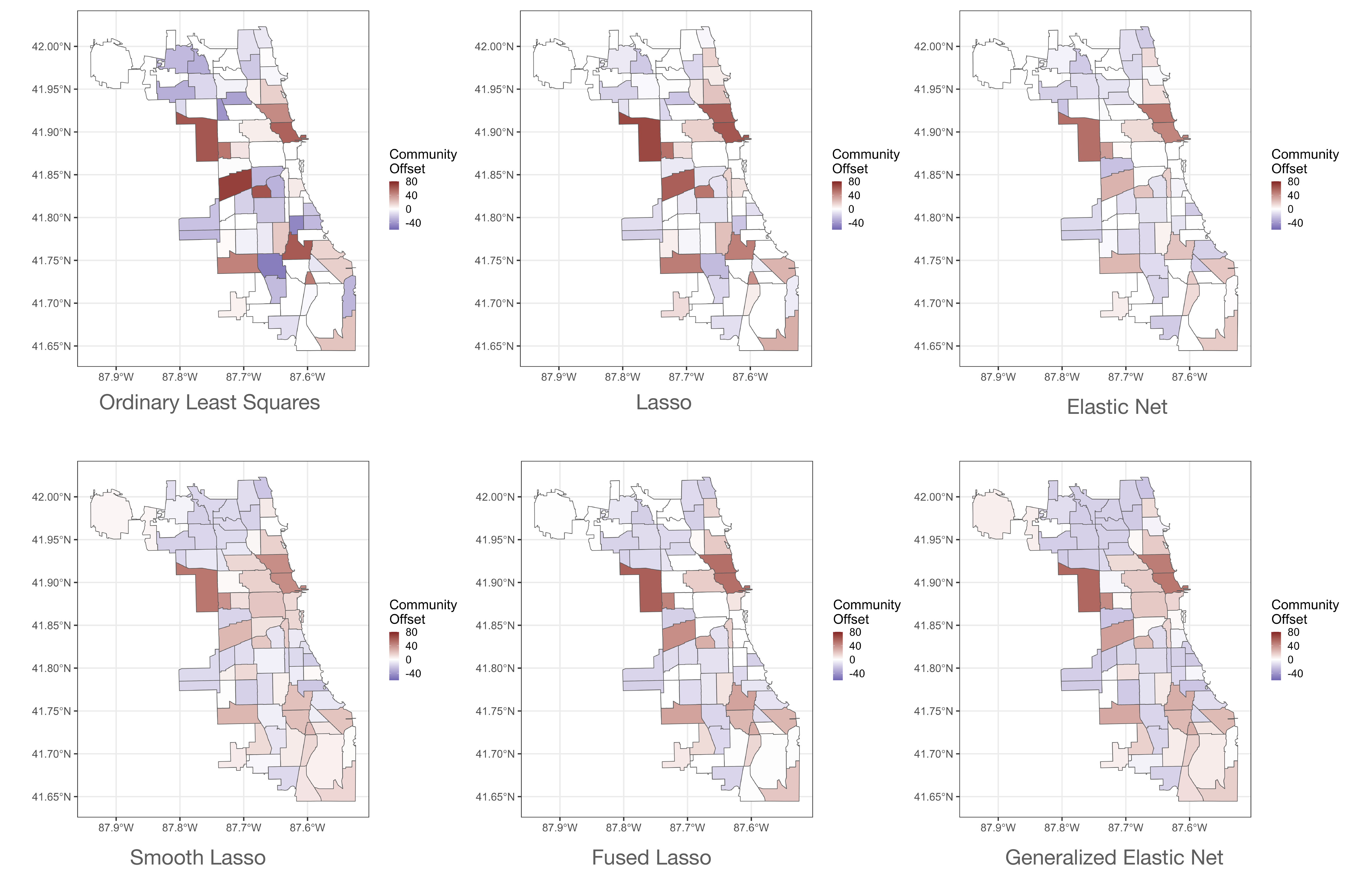}
        \caption{\textit{Visualization of the community offsets $\gamma$ produced by different estimators. Note that GEN produces smoother estimates with greater magnitudes.}}
        \label{fig:coeff_map}
\end{figure}

Figure \ref{fig:coeff_map} and Figure \ref{fig:temporal_component_chicagp} visualize the estimates of the community offsets $\gamma$ and the temporal parameters $(\alpha, \beta)$ respectively. Note that the estimate of $\gamma$ obtained from GEN contains fewer zero entries and is significantly smoother when compared with other methods. We can clearly see that GEN divides the communities into clusters with similar community offsets. Only Smooth Lasso provides an estimate of $\gamma$ that is close to GEN's, but interestingly the estimates obtained by GEN tend to be greater in magnitude. From Figure \ref{fig:temporal_component_chicagp}, we can see that GEN, FL and SL produce smooth estimates to show that crime rates tend to decrease in the colder months and that there is a general reduction in crime rates between 2002 and 2022.

\begin{figure}[H]
    \begin{subfigure}[r]{0.49\textwidth}
            \centering
          \includegraphics[width=\textwidth]{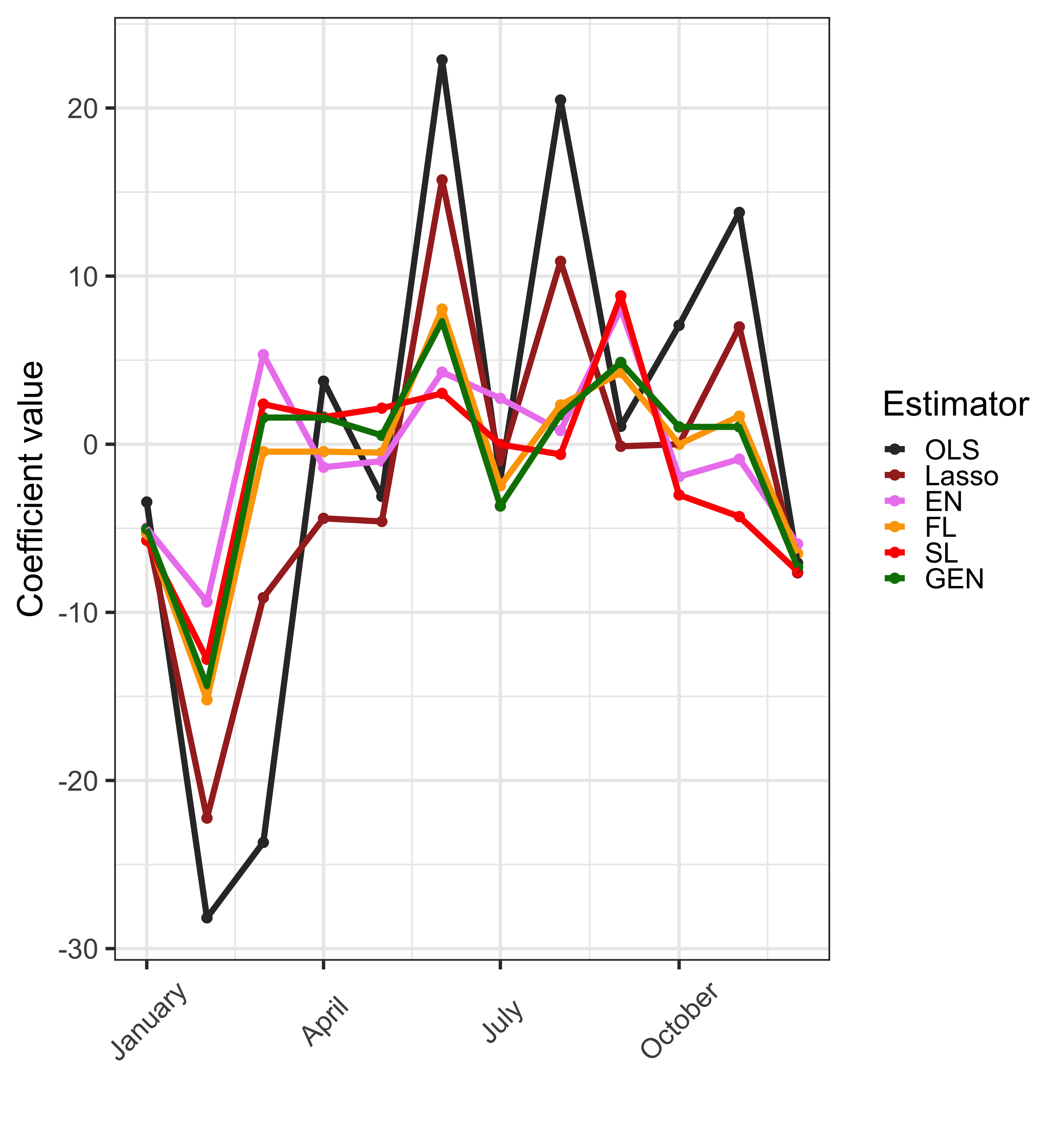}
    \caption{\textit{Estimates of $\alpha$}}      
    \end{subfigure}
    \begin{subfigure}[r]{0.49\textwidth}
            \centering
          \includegraphics[width=\textwidth]{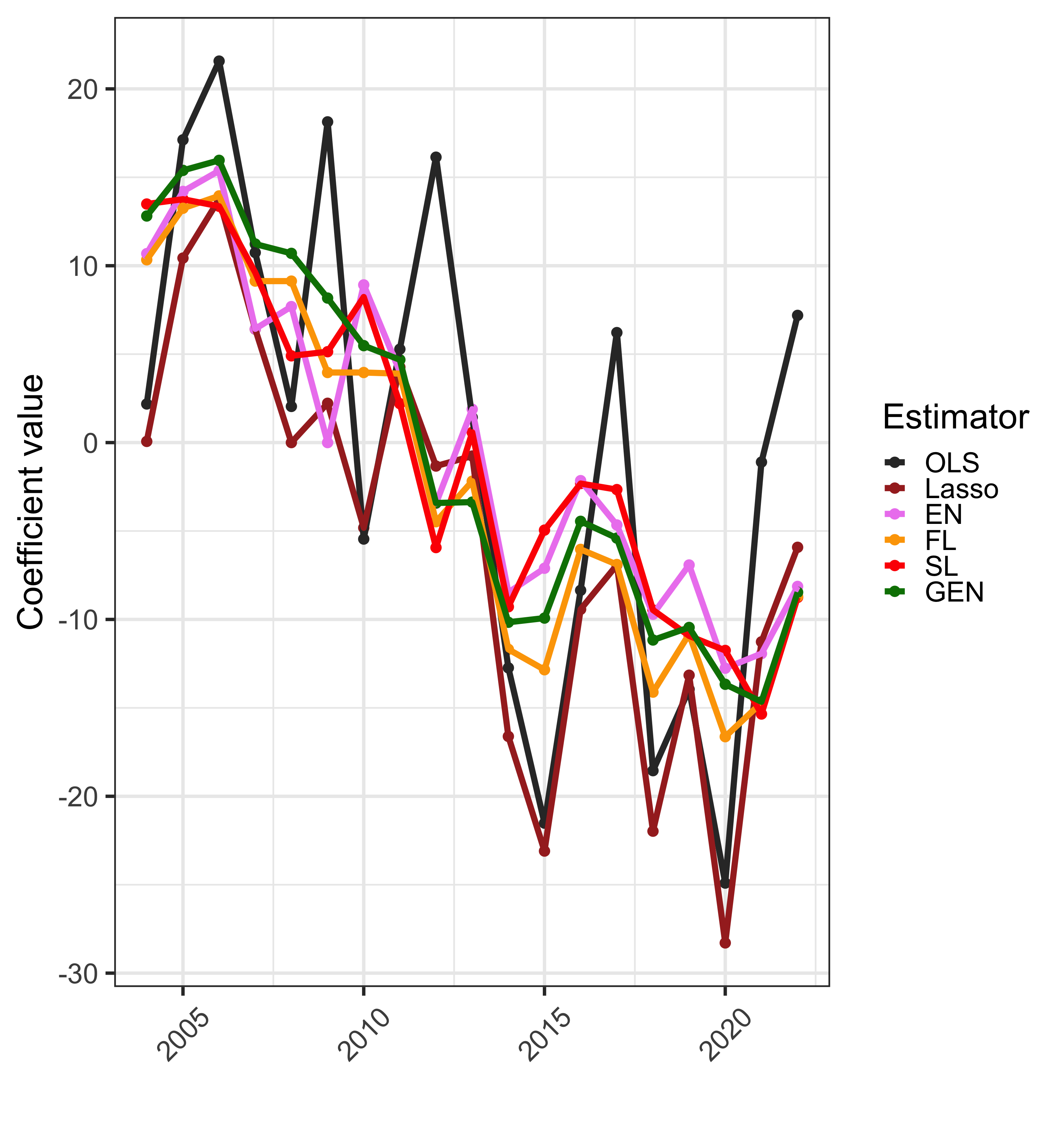}
          \caption{\textit{Estimates of $\beta$} }
    \end{subfigure}
        \caption{\textit{Visualization of the estimates of the monthly parameters $\alpha$ and the yearly parameters $\beta$'s produced by different methods. GEN, FL and SL produce smoother estimates relative to graph-independent methods.}}   
    \label{fig:temporal_component_chicagp}
\end{figure}

%% file: conclusion.tex
In this paper, we have introduced an $\ell_1+\ell_2$ penalty in which both the $\ell_1$ and $\ell_2$ components are defined with respect to a given graph $G$. We have shown that the assumption that $\beta^*$ is piecewise constant or smooth with respect to $G$ can be sufficient for our estimator to be consistent in prediction and estimation, even in the high-dimensional setting. The $\ell_1$ component of our penalty improves estimation of the piecewise constant regions of the true signal, whereas the $\ell_2$ component improves estimation of the true signal's smooth regions and alleviates the negative effects of a highly correlated design. We also show that our estimator can be computed efficiently, even for large-scale problems where choosing hyperparameters via cross-validation may be challenging. 

%% file: appendix_proofs.tex
We restate our theorems in this appendix for convenience. 
\begin{theorem}[Theorem \ref{theorem:1}]\label{theorem1}
Fix $\delta > 0$ and choose $\lambda_1 = 32\sigma \rho(\Gamma)\sqrt{\frac{\gmax(\Sigma)\log p}{n}}$, $\lambda_2 \leq \frac{\lambda_1}{8\|\Gamma\beta^*\|_\infty}$. Given any set $S$ satisfying both
\begin{equation}\label{eq:appendix-regularity}
\frac{144\gmax(\Sigma)(\sqrt{n_c}+\delta)^2}{n} +\frac{36\lambda_1^2|S|k_S^{-2}}{\sigma^2} \leq \frac{1}{2}\gmin\left(\frac{1}{64}\Sigma + \lambda_2 L\right) 
\end{equation}
and 
\begin{equation}\label{eq:appendix-regularity2}
    \lambda_1\|(\Gamma\beta^*)_{-S}\|_1 \leq \frac{\sigma^2}{18}
\end{equation}
with probability at least $1-c_1\exp(-nc_2) - \frac{2}{m} - e^{-\delta^2/2}$ we have 
\begin{equation*}
\|\Sigma^{1/2}(\hat{\beta}-\beta^*)\|_2^2 \lesssim \frac{
\sigma^2\gmax(\Sigma)\frac{n_c+\delta^2}{n} + \lambda_1^2|S|k_S^{-2}
}{\gmin\left(\frac{1}{64}\Sigma+\lambda_2L\right)} + \lambda_1\|(\Gamma\beta^*)_{-S}\|_1
\end{equation*}
and
\begin{equation*}
\|\hat{\beta}-\beta^*\|_2^2 \lesssim \frac{
\sigma^2\gmax(\Sigma)\frac{n_c+\delta^2}{n} + \lambda_1^2|S|k_S^{-2}
}{\gmin^2\left(\frac{1}{64}\Sigma+\lambda_2L\right)} + \frac{\lambda_1\|(\Gamma\beta^*)_{-S}\|_1}{\gmin\left(\frac{1}{64}\Sigma+\lambda_2L\right)}
\end{equation*}
\end{theorem}

\begin{proof}
By definition, 
$$\hat{\beta} = \arg\min_{\beta} \frac{1}{n}\|Y-X\beta\|_2^2 + \lambda_1 \|\Gamma\beta\|_1 + \lambda_2 \|\Gamma\beta\|_2^2 $$

We can also rewrite our estimator as:

$$\hat{\beta} = \arg\min_{\beta} \frac{1}{n}\|\tilde{Y} - \tilde{X}\beta\|_2^2 + \lambda_1\|\Gamma\beta\|_1 $$

Using subdifferential calculus, we can see that $\hat{\beta}$ must satisfy
$$\frac{2\tilde{X}^T(\tilde{Y}-\tilde{X}\hat{\beta})}{n} = \lambda_1 \Gamma^T\text{sign}(\Gamma\hat{\beta}) $$
where 
$$
[\sign(x)]_i = \left\{
\begin{array}{rl}
1 & \text{if}\ x_i > 0,\\
 \text{any value in }\left[-1 , 1\right] & \text{if}\ x_i = 0,\\
-1 & \text{if}\ x_i < 0\,.
\end{array}
\right.
$$

Hence, we obtain 
$$\frac{2}{n}\hat{\beta}^T\tilde{X}^T(\tilde{Y} - \tilde{X}\hat{\beta}) = \lambda_1\hat{\beta}^T\Gamma^T\sign(\Gamma\hat{\beta}) = \lambda_1 \|\Gamma\hat{\beta}\|_1 $$
and for any $\beta\in \mathbb{R}^p$,
$$ \frac{2}{n}\beta^T\tilde{X}^T(\tilde{Y}-\tilde{X}\hat{\beta}) = \lambda_1\beta^T\Gamma^T\sign(\Gamma\hat{\beta}) \leq \lambda_1\|\Gamma\beta\|_1$$

By subtracting the previous equality from the inequality above, for any $\beta \in \mathbb{R}^p$ we have
$$\frac{2}{n}(\beta-\hat{\beta})^T\tilde{X}^T(\tilde{Y}-\tilde{X}\hat{\beta}) \leq \lambda_1(\|\Gamma\beta\|_1 - \|\Gamma\hat{\beta}\|_1) $$

Since $\tilde{Y} = \tilde{X}\beta^* + \tilde{\epsilon}$, 
\begin{align}
&\frac{2}{n}(\hat{\beta}-\beta)^T\tilde{X}^T\tilde{X}(\hat{\beta}-\beta^*)\nonumber \\
&\leq \frac{2}{n}\tilde{\epsilon}^T\tilde{X}(\hat{\beta}-\beta) + \lambda_1(\|\Gamma\beta\|_1 - \|\Gamma\hat{\beta}\|_1)\nonumber\\
&= \frac{2}{n}\epsilon^TX(\hat{\beta}-\beta) - 2\lambda_2(\beta^*)^T\Gamma^T\Gamma(\hat{\beta}-\beta) + \lambda_1(\|\Gamma\beta\|_1-\|\Gamma\hat{\beta}\|_1)\nonumber \\ 
&\leq \frac{2}{n}\epsilon^TX(\hat{\beta}-\beta) + 2\lambda_2\|\Gamma\beta^*\|_\infty \|\Gamma(\hat{\beta}-\beta^*)\|_1 + \lambda_1(\|\Gamma\beta\|_1-\|\Gamma\hat{\beta}\|_1) \nonumber\\ 
&\leq \frac{2}{n}\epsilon^TX(\hat{\beta}-\beta) + \frac{\lambda_1}{4}\|\Gamma(\hat{\beta}-\beta)\|_1 + \lambda_1(\|\Gamma\beta\|_1-\|\Gamma\hat{\beta}\|_1)\nonumber
\end{align}
where the last inequality follows if we choose $\lambda_2 \leq \frac{\lambda_1}{8\|\Gamma\beta^*\|_\infty}$.

We wish to bound $\frac{2}{n}\epsilon^TX(\hat{\beta}-\beta)$. As $\Pi \in \mathbb{R}^{p\times p}$ denotes the projection matrix onto the kernel of $\Gamma$,  we have $I_p = \Pi + \Gamma^\dagger \Gamma$. Hence,
\begin{align}
\frac{2}{n}\epsilon^TX(\hat{\beta}-\beta) &= \frac{2}{n}\epsilon^TX\Pi (\hat{\beta}-\beta) + \frac{2}{n}\epsilon^TX\Gamma^\dagger\Gamma (\hat{\beta}-\beta) \label{eq:stopping-point}\\ 
&\leq \frac{2}{n}\|\Pi X^T\epsilon\|_2\|\hat{\beta}-\beta\|_2 + \frac{2}{n}\|(\Gamma^\dagger)^TX^T\epsilon\|_\infty \|\Gamma(\hat{\beta}-\beta)\|_1 \nonumber\\ 
&\leq \frac{2}{n}\|\Pi X^T\epsilon\|_2\|\hat{\beta}-\beta\|_2 + \frac{\lambda_1}{4} \|\Gamma(\hat{\beta}-\beta)\|_1\nonumber
\end{align}
where the last inequality follows if we choose $\lambda_1 \geq \frac{8}{n}\|(\Gamma^\dagger)^TX^T\epsilon\|_\infty$ (with high probability). 

We obtain the bound: 
\begin{equation}\label{eq:bound-obtained}
\begin{split}
\frac{2}{n}(\hat{\beta}-\beta)^T\tilde{X}^T\tilde{X}(\hat{\beta}-\beta^*) \leq  &\frac{2}{n}\|\Pi X^T\epsilon\|_2\|\hat{\beta}-\beta\|_2 +\frac{\lambda_1}{2}\|\Gamma(\hat{\beta}-\beta)\|_1\\
&+\lambda_1\|\Gamma\beta\|_1-\lambda_1\|\Gamma\hat{\beta}\|_1
\end{split}
\end{equation}
$$\frac{2}{n}(\hat{\beta}-\beta)^T\tilde{X}^T\tilde{X}(\hat{\beta}-\beta^*) \leq  \frac{2}{n}\|\Pi X^T\epsilon\|_2\|\hat{\beta}-\beta\|_2 + \frac{\lambda_1}{2}\|\Gamma(\hat{\beta}-\beta)\|_1+\lambda_1\|\Gamma\beta\|_1-\lambda_1\|\Gamma\hat{\beta}\|_1$$

For any $S \subseteq [m]$: 
\begin{align*}
&\frac{\lambda_1}{2}\|\Gamma(\hat{\beta}-\beta)\|_1+\lambda_1\|\Gamma\beta\|_1-\lambda_1\|\Gamma\hat{\beta}\|_1 \\
&\leq \frac{\lambda_1}{2}\|(\Gamma\hat{\beta}-\Gamma\beta)_S\|_1 + \frac{\lambda_1}{2}\|(\Gamma\hat{\beta})_{-S}\|_1 + \frac{\lambda_1}{2}\|(\Gamma\beta)_{-S}\|_1 + \lambda_1\|\Gamma\beta\|_1 - \lambda_1 \|\Gamma\hat{\beta}\|_1 \\ 
&\leq \frac{3\lambda_1}{2}\|(\Gamma\hat{\beta}-\Gamma\beta)_S\|_1 + \frac{3\lambda_1}{2}\|(\Gamma\beta)_{-S}\|_1 - \frac{\lambda_1}{2}\|(\Gamma\hat{\beta})_{-S}\|_1 \\ 
&\leq \frac{3\lambda_1}{2}\|(\Gamma\hat{\beta}-\Gamma\beta)_S\|_1 + 2\lambda_1\|(\Gamma\beta)_{-S}\|_1 - \frac{\lambda_1}{2}\|(\Gamma\hat{\beta}-\Gamma\beta)_{-S}\|_1\\ 
&\leq 2\lambda_1\|(\Gamma\hat{\beta}-\Gamma\beta)_S\|_1 + 2\lambda_1\|(\Gamma\beta)_{-S}\|_1 - \frac{\lambda_1}{2}\|\Gamma\hat{\beta}-\Gamma\beta\|_1
\end{align*}
and so we have 
\begin{align*}
&\frac{2}{n}(\hat{\beta}-\beta)^T\tilde{X}^T\tilde{X}(\hat{\beta}-\beta^*) + \frac{\lambda_1}{2}\|\Gamma\hat{\beta}-\Gamma\beta\|_1\\
&\leq  \frac{2}{n}\|\Pi X^T\epsilon\|_2\|\hat{\beta}-\beta\|_2 + 2\lambda_1\|(\Gamma\hat{\beta}-\Gamma\beta)_S\|_1 + 2\lambda_1\|(\Gamma\beta)_{-S}\|_1 \\ 
&\leq 2\left(\frac{1}{n}\|\Pi X^T\epsilon\|_2 + \frac{\lambda_1\sqrt{|S|}}{k_S}\right)\|\hat{\beta}-\beta\|_2 + 2\lambda_1\|(\Gamma\beta)_{-S}\|_1 \\ 
&\leq 2\left(\sqrt{2\sigma^2\gmax(\Sigma)}\frac{\sqrt{n_c}+\delta}{\sqrt{n}}+\frac{\lambda_1\sqrt{|S|}}{k_S}\right)\|\hat{\beta}-\beta\|_2 + 2\lambda_1\|(\Gamma\beta)_{-S}\|_1
\end{align*}
with high probability, where we used the definition of $k_S$ and Lemma \ref{lem:projection}. If we set $\beta = \beta^*$, we obtain
\begin{equation}\label{eq:geer2.2}
\begin{split}
    \frac{1}{n}\|\tilde{X}(\hat{\beta}-\beta^*)\|_2^2 + \frac{\lambda_1}{4}\|\Gamma\hat{\beta}-\Gamma\beta^*\|_1 \leq &\left(\sqrt{2\sigma^2\gmax(\Sigma)}\frac{\sqrt{n_c}+\delta}{\sqrt{n}}+\frac{\lambda_1\sqrt{|S|}}{k_S}\right)\|\hat{\beta}-\beta^*\|_2\\ 
    &+\lambda_1\|(\Gamma\beta^*)_{-S}\|_1
\end{split}
\end{equation}
which implies 
$$\lambda_1\|\Gamma\hat{\beta}-\Gamma\beta^*\|_1 \leq 4\left(\sqrt{2\sigma^2\gmax(\Sigma)}\frac{\sqrt{n_c}+\delta}{\sqrt{n}}+\frac{\lambda_1\sqrt{|S|}}{k_S}\right)\|\hat{\beta}-\beta^*\|_2 + 4\lambda_1\|(\Gamma\beta^*)_{-S}\|_1 $$
or that
\begin{align}
&576\rho^2(\Gamma)\frac{\gmax(\Sigma)\log p}{n}\|\Gamma\hat{\beta}-\Gamma\beta^*\|_1^2 = \frac{576}{1024}\frac{\lambda_1^2}{\sigma^2}\|\Gamma\hat{\beta}-\Gamma\beta^*\|_1^2\nonumber\\
&\leq 18\frac{\lambda_1^2}{\sigma^2}\left(\sqrt{2\sigma^2\gmax(\Sigma)}\frac{\sqrt{n_c}+\delta}{\lambda_1\sqrt{n}}+\frac{\sqrt{|S|}}{k_S}\right)^2\|\hat{\beta}-\beta^*\|_2^2+ 18\frac{\lambda_1^2}{\sigma^2}\|(\Gamma\beta^*)_{-S}\|_1^2\nonumber\\
&\leq \left(72\gmax(\Sigma)\frac{(\sqrt{n_c}+\delta)^2}{n} + 36\frac{\lambda_1^2|S|k_S^{-2}}{\sigma^2}\right)\|\hat{\beta}-\beta^*\|_2^2 + 18\frac{\lambda_1^2}{\sigma^2}\|(\Gamma\beta^*)_{-S}\|_1^2\nonumber\\
&\leq \left(72\gmax(\Sigma)\frac{(\sqrt{n_c}+\delta)^2}{n} + 36\frac{\lambda_1^2|S|k_S^{-2}}{\sigma^2}\right)\|\hat{\beta}-\beta^*\|_2^2 + \lambda_1\|(\Gamma\beta^*)_{-S}\|_1\label{eq:cone}
\end{align}
where we used the condition \eqref{eq:appendix-regularity2}. Now if we apply Corollary \ref{cor:re} to \eqref{eq:geer2.2}, we have 
\begin{equation*}
\begin{split}
&(\hat{\beta}-\beta^*)^T\left(\frac{1}{64}\Sigma+\lambda_2L\right)(\hat{\beta}-\beta^*)\leq \left(\sqrt{2\sigma^2\gmax(\Sigma)}\frac{\sqrt{n_c}+\delta}{\sqrt{n}}+\frac{\lambda_1\sqrt{|S|}}{k_S}\right)\|\hat{\beta}-\beta^*\|_2 \\
&+ \lambda_1\|(\Gamma\beta^*)_{-S}\|_1+\frac{72\gmax(\Sigma)n_c}{n}\|\hat{\beta}-\beta^*\|_2^2 + 576\rho^2(\Gamma)\frac{\gmax(\Sigma)\log p}{n}\|\Gamma\hat{\beta}-\Gamma\beta^*\|_1^2
\end{split}
\end{equation*}
which, by \eqref{eq:cone} and the inequality $\frac{n_c}{n}\leq \frac{(\sqrt{n_c}+\delta)^2}{n}$, implies
\begin{equation*}
\begin{split}
&(\hat{\beta}-\beta^*)^T\left(\frac{1}{64}\Sigma+\lambda_2L\right)(\hat{\beta}-\beta^*)\leq \left(\sqrt{2\sigma^2\gmax(\Sigma)}\frac{\sqrt{n_c}+\delta}{\sqrt{n}}+\frac{\lambda_1\sqrt{|S|}}{k_S}\right)\|\hat{\beta}-\beta^*\|_2 \\
&+ 2\lambda_1\|(\Gamma\beta^*)_{-S}\|_1+\left(\frac{144\gmax(\Sigma)(\sqrt{n_c}+\delta)^2}{n}+\frac{36\lambda_1^2|S|k_S^{-2}}{\sigma^2}\right)\|\hat{\beta}-\beta^*\|_2^2
\end{split}
\end{equation*}
If we now apply the condition \eqref{eq:appendix-regularity}, we obtain 
\begin{equation*}
\begin{split}
(\hat{\beta}-\beta^*)^T\left(\frac{1}{64}\Sigma+\lambda_2L\right)&(\hat{\beta}-\beta^*)\leq \left(\sqrt{2\sigma^2\gmax(\Sigma)}\frac{\sqrt{n_c}+\delta}{\sqrt{n}}+\frac{\lambda_1\sqrt{|S|}}{k_S}\right)\|\hat{\beta}-\beta^*\|_2 \\
&+ 2\lambda_1\|(\Gamma\beta^*)_{-S}\|_1+\frac{1}{2}\gmin\left(\frac{1}{64}\Sigma+\lambda_2L\right)\|\hat{\beta}-\beta^*\|_2^2
\end{split}
\end{equation*}
which, by using $\gmin\left(\frac{1}{64}\Sigma+\lambda_2L\right)\|\hat{\beta}-\beta^*\|_2^2 \leq (\hat{\beta}-\beta^*)^T\left(\frac{1}{64}\Sigma+\lambda_2L\right)(\hat{\beta}-\beta^*)$, implies both
\begin{equation}\label{eq:quadratic1}
\begin{split}
    \gmin\left(\frac{1}{64}\Sigma+\lambda_2L\right)\|\hat{\beta}-\beta^*\|_2^2 \leq &2\left(\sqrt{2\sigma^2\gmax(\Sigma)}\frac{\sqrt{n_c}+\delta}{\sqrt{n}}+\frac{\lambda_1\sqrt{|S|}}{k_S}\right)\|\hat{\beta}-\beta^*\|_2 \\
    &+4\lambda_1\|(\Gamma\beta^*)_{-S}\|_1
\end{split}
\end{equation}
and 
\begin{equation}\label{eq:quadratic2}
\begin{split}
    (\hat{\beta}-\beta^*)^T&\left(\frac{1}{64}\Sigma+\lambda_2L\right)(\hat{\beta}-\beta^*) \leq 4\lambda_1\|(\Gamma\beta^*)_{-S}\|_1 \\
    &+ 2\frac{\sqrt{2\sigma^2\gmax(\Sigma)}\frac{\sqrt{n_c}+\delta}{\sqrt{n}}+\frac{\lambda_1\sqrt{|S|}}{k_S}}{\sqrt{\gmin\left(\frac{1}{64}\Sigma+\lambda_2L\right)}}\sqrt{(\hat{\beta}-\beta^*)^T\left(\frac{1}{64}\Sigma+\lambda_2L\right)(\hat{\beta}-\beta^*)}
\end{split}
\end{equation}
The error bounds follow from \eqref{eq:quadratic1} and \eqref{eq:quadratic2} if we note that $x^2 - bx - c \leq 0$ implies $x^2 \leq 4\max(b^2, c) \leq 4(b^2+c)$, for $b, c>0$.
\end{proof}
\begin{theorem}[Theorem \ref{theorem:chain}]
Let $\Gamma$ be the incidence matrix of the $p$-vertex chain graph, and fix $\delta > 0$. With an appropriate choice of $\lambda_1$ and $\lambda_2 \leq \frac{\lambda_1}{8\|\Gamma\beta^*\|_\infty}$, with high probability we have 
\begin{equation}\label{eq:appendix-chain-pred}
    \|\Sigma^{1/2}(\hat{\beta}-\beta^*)\|_2^2 \lesssim \frac{\sigma^2\gmax(\Sigma)}{\gmin(\frac{1}{64}\Sigma+\lambda_2L)}\frac{1+\delta^2}{n} + \frac{(\sigma^2\gmax(\Sigma)\|\Gamma\beta^*\|_1)^{2/3}}{\gmin^{1/3}\left(\frac{1}{64}\Sigma+\lambda_2L\right)}\frac{(p\log p)^{1/3}}{n^{2/3}}
\end{equation}
\begin{equation}
    \|\hat{\beta}-\beta^*\|_2^2 \lesssim \frac{\sigma^2\gmax(\Sigma)}{\gmin^2(\frac{1}{64}\Sigma+\lambda_2L)}\frac{1+\delta^2}{n} + \frac{(\sigma^2\gmax(\Sigma)\|\Gamma\beta^*\|_1)^{2/3}}{\gmin^{4/3}\left(\frac{1}{64}\Sigma+\lambda_2L\right)}\frac{(p\log p)^{1/3}}{n^{2/3}}
\end{equation}
provided that the RHS of \eqref{eq:appendix-chain-pred} is smaller than $C\sigma^2$.
\end{theorem}
\begin{proof}
The proof is identical to that of Theorem \ref{theorem1} up to \eqref{eq:stopping-point}. However, we need to bound $\frac{2}{n}\epsilon^TX\Gamma^\dagger\Gamma(\hat{\beta}-\beta)$ differently. 

Let $\Gamma = U\Xi V^T$ be the singular value decomposition of $\Gamma$, and let $\xi_1, \dots, \xi_{p-1}$ be the nonzero singular values of $\Gamma$. Let $u_1, \dots, u_m$ and $v_1, \dots, v_p$ denote the columns of $U$ and $V$. Denote by  $V_{[k]}\in \mathbb{R}^{p\times k}$ the matrix containing the first $k$ columns of $V$ ($k$ is to be specified later) and $V_{-[k]}\in \mathbb{R}^{p\times (p-k)}$ the matrix containing the other $p-k$ columns of $V$. Define the projection matrix $P_{[k]}:= V_{[k]}V_{[k]}^T\in \mathbb{R}^{p\times p}$.

Noting that $\Gamma^\dagger \Gamma$ is a projection matrix, we have:
\begin{align}
&\frac{2}{n}\epsilon^TX\Gamma^\dagger\Gamma(\hat{\beta}-\beta) \nonumber\\
&= \frac{2}{n}\epsilon^T XP_{[k]}\Gamma^\dagger\Gamma(\hat{\beta}-\beta)+ \frac{2}{n}\epsilon^T X(I_p-P_{[k]})\Gamma^\dagger\Gamma(\hat{\beta}-\beta)\nonumber\\
&\leq \frac{2}{n}\|P_{[k]}X^T\epsilon\|_2\|\Gamma^\dagger\Gamma(\hat{\beta}-\beta)\|_2 + \frac{2}{n}\|(\Gamma^\dagger)^T(I_p-P_{[k]})X^T\epsilon\|_\infty \|\Gamma(\hat{\beta}-\beta)\|_1\nonumber\\
&\leq \frac{2}{n}\|P_{[k]}X^T\epsilon\|_2\|\hat{\beta}-\beta\|_2 + \frac{\lambda_1}{2}\|\Gamma(\hat{\beta}-\beta)\|_1 \label{eq:chain2}
\end{align}
if we choose $\lambda_1 \geq \frac{8}{n}\|(\Gamma^\dagger)^T(I_p-P_{[k]})X^T\epsilon\|_\infty$ with high probability. 

In order to choose $k$, we need to bound $\frac{8}{n}\|(\Gamma^\dagger)^T(I_p-P_{[k]})X^T\epsilon\|_\infty$. Let $s_1', \dots, s_m'$ be the columns of $(I_p - P_{[k]})\Gamma^\dagger$. Let $e_j$, $j\in [m]$, denote the $j^\text{th}$ canonical basis element. As in the proof of Theorem 6 of \cite{wang2015trend}, we have: 
\begin{align*}
    \|s_j'\|_2^2 &= \|(I_p - V_{[k]}V_{[k]}^T)V\Xi^\dagger U^Te_j\|_2^2\\ 
    &= \|\begin{bmatrix}0 & V_{-[k]}\end{bmatrix}\Xi^\dagger U^Te_j\|_2^2 = \left\|\sum_{i=k+1}^{p-1}\xi_i^{-1}\langle u_i, e_j\rangle v_i\right\|_2^2 \\
    &= \sum_{i=k+1}^{p-1}\xi_i^{-2}\langle u_i, e_j\rangle^2 \leq \frac{2}{p}\sum_{i=k+1}^{p-1}\xi_i^{-2}
\end{align*}
where we made use of the fact that the left singular vectors $\{u_i\}_{i=1}^m$ of $\Gamma$, when $\Gamma$ is the incidence matrix of the chain graph with $p$ vertices, satisfy $\forall i \in [m]: \|u_i\|_\infty \leq \sqrt{\frac{2}{p}}$. 

For the chain graph, the nonzero singular values $\xi_i$ are such that
$$\xi_i^2 = 4\sin^2\left(\frac{\pi i}{2p}\right) = 2 - 2\cos\left(\frac{\pi i}{p}\right), \text{ for }i = 1, \dots, p-1$$
Hence, as in \cite{wang2015trend},
\begin{align*}
    \max_{j\in [m]}\|s_j'\|_2^2\leq\frac{2}{p}\sum_{i=k+1}^{p-1} \xi_i^{-2} &= \frac{1}{2p} \sum_{i=k+1}^{p-1}\sin^{-2}\left(\frac{\pi i}{2p}\right)\\
    &\leq \frac{1}{2p}\int_k^p\sin^{-2}\left(\frac{\pi x}{2p}\right)dx
    = \frac{\cos\left(\frac{\pi k}{2p}\right)}{\pi\sin\left(\frac{\pi k}{2p}\right)} \leq \frac{4p}{\pi^2k}
\end{align*}
where we used $\sin(x) \geq x/2$ and $\cos(x) \leq 1$ for $x \in [0,\pi/2]$.

Similar to Lemma \ref{lem:lambda1}, we can then select $\lambda_1 = \frac{64}{\pi}\sigma\sqrt{p/k}\sqrt{\frac{\gmax(\Sigma)\log p}{n}}$. We also have $\frac{2}{n}\|P_{[k]}X^T\epsilon\|_2 \leq 4\sqrt{\frac{2\sigma^2\gmax(\Sigma)k}{n}}$ with probability at least $1-e^{-n/8}-e^{-k^2/2}$, as in Lemma \ref{lem:projection}. The rest of the proof is again identical to that of Theorem \ref{theorem:1}, and we obtain for any $S$ that 
\begin{equation*}
    \|\Sigma^{1/2}(\hat{\beta}-\beta^*)\|_2^2 \lesssim \frac{\sigma^2\gmax(\Sigma)\left[\frac{1+\delta^2}{n}+\frac{k}{n}\right]+\lambda_1^2|S|k_{S}^{-2}}{\gmin\left(\frac{1}{64}\Sigma+\lambda_2L\right)} + \lambda_1\|(\Gamma\beta^*)_{-S}\|_1
\end{equation*}
with high probability. By setting $S = \emptyset$ and choosing $k$ such that 
$$ \frac{\sigma^2\gmax(\Sigma)k}{\gmin\left(\frac{1}{64}\Sigma+\lambda_2L\right) n} \asymp \lambda_1 \|\Gamma\beta^*\|_1 \asymp  \sigma\sqrt{\frac{p\gmax(\Sigma)\log p}{kn}}\|\Gamma\beta^*\|_1$$
we obtain 
$$k \asymp \left(\frac{p \log p \|\Gamma\beta^*\|_1^2\gmin^2\left(\frac{1}{64}\Sigma+\lambda_2L\right)n}{\sigma^2\gmax(\Sigma)}\right)^{1/3} $$
and with this choice of $k$
$$\frac{\sigma^2\gmax(\Sigma)k}{\gmin\left(\frac{1}{64}\Sigma+\lambda_2L\right) n} \asymp \frac{(\sigma^2\gmax(\Sigma)\|\Gamma\beta^*\|_1)^{2/3}}{\gmin^{1/3}\left(\frac{1}{64}\Sigma+\lambda_2L\right)}\frac{(p\log p)^{1/3}}{n^{2/3}} $$
\end{proof}

We will often use the following lemma to compare probabilities involving two Gaussian vectors.

\begin{lemma}[Anderson's Gaussian comparison inequality \cite{anderson1955integral}]\label{lem:anderson} Let $X$ and $Y$ be two zero-mean Gaussian vectors with covariance $\Sigma_X$ and $\Sigma_Y$ respectively. If $\Sigma_Y-\Sigma_X$ is positive semi-definite, then for any convex set $C$ satisfying $C = -C$, 
$$\bP(X\in C) \geq \bP(Y\in C) $$
\end{lemma}

\begin{lemma}[Lemma \ref{lem:re}]
\label{lem:re-app}
If $X\in \mathbb{R}^{n\times p}$ has i.i.d. $N(0, \Sigma)$ rows and $m \geq 2$, $n \geq 10$, then the event 
$$\bigg\{\forall v \in \mathbb{R}^p: \frac{\|Xv\|_2}{\sqrt{n}} \geq \frac{1}{4} \|\Sigma^{1/2}v\|_2 - 3\sqrt{\frac{\gmax(\Sigma)n_c}{n}}\|v\|_2 - 6\sqrt{2}\rho(\Gamma)\sqrt{\frac{\gmax(\Sigma)\log p}{n}}\|\Gamma v\|_1\bigg\}$$
holds with probability at least $1 - c_1\exp(-nc_2)$, for some universal constants $c_1, c_2 > 0$.
\end{lemma}

\begin{proof}
We follow the proof outline of \cite{raskutti2010restricted}. First note that we can restrict our attention to $v \in \mathbb{R}^p$ satisfying $\|\Sigma^{1/2} v\|_2 = 1$, as the inequality that defines the event above is invariant to scaling of $v$. Define 
$$V(r,s) := \{v \in \mathbb{R}^p: \|\Sigma^{1/2} v\|_2 = 1, \|\Gamma v\|_1 \leq r, \|v\|_2 \leq s\} 
$$
and 
$$M(r,s,X) := \sup_{v \in V(r,s)}\left(1-\frac{\|Xv\|_2}{\sqrt{n}}\right) $$
\textit{Bounding the expectation $\bE(M(r,s,X))$}: By an application of Gordon's inequality (see Section 4.2 of \cite{raskutti2010restricted} for the details), 
\begin{align*}
    \bE\left(\sup_{v\in V(r,s)}(-\|Xv\|_2)\right) = \bE\left(\sup_{v\in V(r,s)}\inf_{u\in S^{n-1}}u^TXv\right) &\leq \bE\left(\sup_{v\in V(r,s)}\inf_{u\in S^{n-1}}g^Tu + h^T\Sigma^{1/2}v\right)\\
    &= -\bE\|g\|_2 + \bE\left(\sup_{v\in V(r,s)}h^T\Sigma^{1/2}v\right)
\end{align*}
where $g \sim N(0,I_n)$ independent of $h \sim N(0,I_p)$. We know that $\bE\|g\|_2 \geq \frac{3}{4}\sqrt{n}$ when $n \geq 10$, so we just need to upper bound $\bE\left(\sup_{v\in V(r,s)}h^T\Sigma^{1/2}v\right)$.\\ 

Since $\Pi + \Gamma^\dagger \Gamma = I_p$, 
$$h^T\Sigma^{1/2}v = h^T\Sigma^{1/2}(\Pi + \Gamma^\dagger \Gamma)v \leq \|\Pi\Sigma^{1/2}h\|_2 \|v\|_2 + \|(\Gamma^\dagger)^T\Sigma^{1/2}h\|_\infty \|\Gamma v\|_1 $$
and by definition of $V(r,s)$ we have $\|v\|_2 \leq s$ and $\|\Gamma v\|_1\leq r$ for all $v \in V(r,s)$, so we obtain
$$ \bE\left(\sup_{v\in V(r,s)}h^T\Sigma^{1/2}v\right)\leq s\bE\|\Pi\Sigma^{1/2}h\|_2 + r \bE \|(\Gamma^\dagger)^T\Sigma^{1/2}h\|_\infty$$

Note that the spectral decomposition of $\Pi = U\Lambda U^T$, where $U$ is an orthogonal matrix, is such that $\Lambda$ is a diagonal matrix with $n_c$ ones and $p-n_c$ zeros on the diagonal. Since $\gmax(\Sigma)\Pi - \Pi\Sigma\Pi$ is positive semi-definite, by Lemma \ref{lem:anderson} we know that $\|\sqrt{\gmax(\Sigma)}\Pi h\|_2$ stochastically dominates $\|\Pi\Sigma^{1/2}h\|_2$, and hence
\begin{align*}
    \bE\|\Pi\Sigma^{1/2}h\|_2 &\leq \sqrt{\gmax(\Sigma)}\bE\|\Pi h\|_2 \\
    &= \sqrt{\gmax(\Sigma)}\bE\|U\Lambda U^T h\|_2 \\
    &= \sqrt{\gmax(\Sigma)}\bE\|\Lambda h\|_2 \\
    &= \sqrt{\gmax(\Sigma)} \bE\sqrt{h_1^2 + \cdots + h_{n_c}^2}\\
    &\leq \sqrt{\gmax(\Sigma)n_c}
\end{align*}
where we have used Jensen's inequality and the rotational invariance of the standard Gaussian distribution in the above derivations.\\

By Exercise 2.12 b) of \cite{wainwright2019high}, we also have for $m \geq 2$:
\begin{align*}
    \bE\|(\Gamma^\dagger)^T\Sigma^{1/2}h\|_\infty &= \bE \max_{j \in [m]}|\langle s_j, \Sigma^{1/2}h\rangle| \\
    &\leq 2\sqrt{\gmax(\Sigma)}\rho(\Gamma) \sqrt{\log m} \leq 2\sqrt{2}\sqrt{\gmax(\Sigma)}\rho(\Gamma)\sqrt{\log p}
\end{align*}
since \{$\langle s_j, \Sigma^{1/2}h\rangle: j=1,\cdots, m\}$ is a collection of $m$ zero-mean Gaussian variables with variance at most $\gmax(\Sigma)\max_{j\in [m]}\|s_j\|_2^2 = \gmax(\Sigma)\rho(\Gamma)^2$ (and in the last inequality we used $m \leq p^2$).\\

We can therefore conclude
$$\bE\left(-\inf_{v\in V(r,s)}\|Xv\|_2\right) \leq -\frac{3}{4}\sqrt{n} + s\sqrt{\gmax(\Sigma)n_c} + 2\sqrt{2}r\sqrt{\gmax(\Sigma)}\rho(\Gamma)\sqrt{\log p} $$

Dividing by $\sqrt{n}$ and adding 1 on both sides, we obtain 
$$\bE\left(M(r,s,X)\right)\leq \frac{1}{4} + s\sqrt{\frac{\gmax(\Sigma)n_c}{n}} + 2\sqrt{2}r\sqrt{\gmax(\Sigma)}\rho(\Gamma)\sqrt{\frac{\log p}{n}}$$
\textit{Concentration around the mean for $M(r,s,X)$}: As $M(r,s,X)$ is a Lipschitz function of a Gaussian vector (see Section 4.3 of \cite{raskutti2010restricted} for details), for all $t>0$ we have:
$$\bP(|M(r,s,X) - \bE M(r,s,X)| \geq t/2) \leq 2\exp(-nt^2/8)$$
Substituting $t = t(r,s) := \frac{1}{4} + s\sqrt{\frac{\gmax(\Sigma)n_c}{n}} + 2\sqrt{2}r\sqrt{\gmax(\Sigma)}\rho(\Gamma)\sqrt{\frac{\log p}{n}}$, we obtain
$$\bP\left(M(r,s,X) \geq \frac{3t(r,s)}{2}\right) \leq 2\exp(-nt(r,s)^2/8) $$
\textit{Peeling}: This part is adapted from Section 4.4 of \cite{raskutti2010restricted}. We have shown that
\begin{align*}
   &\bP\left(\sup_{\substack{\|v\|_2\leq s, \|\Gamma v\|_1 \leq r \\ \|\Sigma^{1/2}v\|_2=1}}\left(1-\frac{\|Xv\|_2}{\sqrt{n}}\right) \geq \frac{3}{8} + \frac{3}{2}\sqrt{\frac{\gmax(\Sigma)n_c}{n}}s + 3\sqrt{2}\rho(\Gamma)\sqrt{\frac{\gmax(\Sigma)\log p}{n}}r\right) \\
   &\leq 2\exp\left(-\frac{n}{18}\left(\frac{3}{8} + \frac{3}{2}s\sqrt{\frac{\gmax(\Sigma)n_c}{n}}+ 3\sqrt{2}\rho(\Gamma)r\sqrt{\frac{\gmax(\Sigma)\log p}{n}}\right)^2\right) 
\end{align*}

Let $g_1(s) := \frac{3}{16} + \frac{3}{2}\sqrt{\frac{\gmax(\Sigma)n_c}{n}}s$ and $g_2(r):= \frac{3}{16} + 3\sqrt{2}\rho(\Gamma)\sqrt{\frac{\gmax(\Sigma)\log p}{n}}r$. We can rewrite the above as 
$$\bP\left(\sup_{\substack{\|v\|_2\leq s, \|\Gamma v\|_1 \leq r \\ \|\Sigma^{1/2}v\|_2=1}}\left(1-\frac{\|Xv\|_2}{\sqrt{n}}\right)\geq g_1(s) + g_2(r)\right) \leq 2\exp\left(-\frac{n}{18}[g_1(s) + g_2(r)]^2\right)$$

Note that $g_1 \geq \mu$ and $g_2 \geq \mu$ where $\mu := \frac{3}{16}$. For $i= 1, 2, \cdots, $ and $j = 1, 2, \cdots$, we define the sets
$$A_{ij} := \{v\in \mathbb{R}^p: \|\Sigma^{1/2}v\|_2=1, 2^{i-1}\mu \leq g_1(\|v\|_2) < 2^i\mu, 2^{j-1}\mu \leq g_2(\|\Gamma v\|_1) < 2^j\mu\} $$

Also, we define the events 
$$ \mathcal{E}_{ij} := \bigg\{\exists v \in A_{ij}: 1-\frac{\|Xv\|_2}{\sqrt{n}} \geq 2[g_1(\|v\|_2) + g_2(\|\Gamma v\|_1)]\bigg\}$$
as well as the event 
$$\mathcal{E} := \bigg\{\exists v \in \mathbb{R}^p: \|\Sigma^{1/2}v\|_2=1 \text{ and }1-\frac{\|Xv\|_2}{\sqrt{n}}\geq 2\left[g_1(\|v\|_2) + g_2(\|\Gamma v\|_1)\right]\bigg\}$$

Note that $\mathcal{E} = \cup_{i=1}^\infty \cup_{j=1}^\infty \mathcal{E}_{ij}$. Our goal is to prove that $\bP(\mathcal{E}) \leq c_1\exp(-nc_2)$, from which the lemma follows.\\

If we have $v \in A_{ij}$ such that $1-\frac{\|Xv\|_2}{\sqrt{n}} \geq 2[g_1(\|v\|_2) + g_2(\|\Gamma v\|_1)]$ holds, then by definition of $A_{ij}$,
$$1-\frac{\|Xv\|_2}{\sqrt{n}} \geq 2(2^{i-1}\mu + 2^{j-1}\mu) = 2^i\mu + 2^j \mu = g_1(g_1^{-1}(2^i\mu)) + g_2(g_2^{-1}(2^j \mu)) $$
Again by definition of $A_{ij}$, $g_{1}(\|v\|_2)\leq 2^i \mu$ and $g_2(\|\Gamma v\|_1) \leq 2^j \mu$, and so 
$$\|v\|_2 \leq g_1^{-1}(2^i\mu) \qquad \text{and}\qquad \|\Gamma v\|_1\leq g_2^{-1}(2^j\mu)$$

Therefore, we must have 
\begin{align*}
    \bP(\mathcal{E}_{ij}) &\leq 2\exp\left(-\frac{n}{18}[g_1(g_1^{-1}(2^i\mu)) + g_2(g_2^{-1}(2^j\mu))]^2\right)\\
    &=2\exp\left(-\frac{n}{18}(2^i+2^j)^2\mu^2\right)\\
    &\leq 2\exp\left(-\frac{n}{18}2^{2i}\mu^2\right)\exp\left(-\frac{n}{18}2^{2j}\mu^2\right)
\end{align*}
Hence, 
\begin{align*}
    \bP(\mathcal{E}) &\leq 2\sum_{i=1}^\infty \sum_{j=1}^\infty \exp\left(-\frac{n}{18}2^{2i}\mu^2\right)\exp\left(-\frac{n}{18}2^{2j}\mu^2\right) \\ 
    &= 2\left(\sum_{i=1}^\infty \exp\left(-\frac{n}{18}2^{2i}\mu^2\right)\right)^2\\
    &\leq 2\left(\sum_{i=1}^\infty\exp\left(-\frac{ni}{18}\mu^2\right)\right)^2 \\
    &= 2\left(\frac{\exp\left(-\frac{n}{18}\mu^2\right)}{1-\exp\left(-\frac{n}{18}\mu^2\right)}\right)^2\leq c_1\exp(-nc_2)
\end{align*}
\end{proof}

\begin{corollary}
\label{cor:re}
Under the settings of Lemma \ref{lem:re-app}, 
$$\frac{\|\tilde{X}v\|_2^2}{n} \geq v^T\left(\frac{1}{64}\Sigma+\lambda_2L\right)v - \frac{72\gmax(\Sigma)n_c}{n}\|v\|_2^2 - 576\rho(\Gamma)^2\frac{\gmax(\Sigma)\log p}{n}\|\Gamma v\|_1^2$$
holds for all $v\in \mathbb{R}^p$ with probability at least $1-c_1\exp(-nc_2)$
\end{corollary}
\begin{proof}
We argue in a manner similar to the proof of Theorem 7.16 in \cite{wainwright2019high}. For any real numbers $a, b, c$ such that $c \geq \max(a-b, 0)$, we claim that $c^2 \geq (1-\delta)^2a^2 - \frac{b^2}{\delta^2}$ for any $\delta \in (0,1)$. This is because if $b \geq a\delta$, then $(1-\delta)^2a^2-\frac{b^2}{\delta^2}\leq a^2[(1-\delta)^2-1] \leq 0 \leq c$, and if $b < a \delta$, then since $c\geq a-b$, we have $c\geq a-a\delta = (1-\delta)a$.\\

Letting $c = \frac{\|Xv\|_2}{\sqrt{n}}$, $a = \frac{1}{4}\|\Sigma^{1/2}v\|_2$, $b = 3\sqrt{\frac{\gmax(\Sigma)n_c}{n}}\|v\|_2 + 6\sqrt{2}\rho(\Gamma)\sqrt{\frac{\gmax(\Sigma)\log p}{n}}\|\Gamma v\|_1$ and $\delta = \frac{1}{2}$, we obtain for all $v \in \mathbb{R}^p$ with probability at least $1 - c_1\exp(-nc_2)$:

\begin{align*}
 \frac{\|Xv\|_2^2}{n} &\geq \frac{1}{64}\|\Sigma^{1/2}v\|_2^2 - 36\left(\sqrt{\frac{\gmax(\Sigma)n_c}{n}}\|v\|_2 + 2\sqrt{2}\rho(\Gamma)\sqrt{\frac{\gmax(\Sigma)\log p}{n}}\|\Gamma v\|_1\right)^2   \\
 &\geq \frac{1}{64}\|\Sigma^{1/2}v\|_2^2 - 72\frac{\gmax(\Sigma)n_c}{n}\|v\|_2^2 - 576\rho(\Gamma)^2\frac{\gmax(\Sigma)\log p}{n}\|\Gamma v\|_1^2
\end{align*}

By adding $\lambda_2 v^T L v$ to both sides, we obtain what we need to prove.
\end{proof}

\begin{lemma}[High-probability bound on $\|\Pi X^T\epsilon\|_2$]
\label{lem:projection}

For any $\delta > 0$, 
$$\|\Pi X^T\epsilon\|_2 \leq \sqrt{2\sigma^2n\gmax(\Sigma)}(\sqrt{n_c}+\delta) $$
with probability at least $1-e^{-n/8}-e^{-\delta^2/2}$.
\end{lemma}
\begin{proof}
We make use of the fact that $X$ and $\epsilon$ are independent. Note that $X\Pi$ has i.i.d. $N(0, \Pi\Sigma\Pi)$ rows, which we denote by $\tilde{x}_1, \cdots, \tilde{x}_n$. Then 
$$\|\Pi X^T\epsilon\|_2 = \left\|\sum_{i=1}^n\epsilon_i \tilde{x}_i\right\|_2 = \|\epsilon\|_2\left\|\frac{1}{\|\epsilon\|_2}\sum_{i=1}^n\epsilon_i \tilde{x}_i\right\|_2 $$
which has the same distribution as $\|\epsilon\|_2\|\tilde{x}\|_2$, where $\tilde{x} \sim N(0, \Pi \Sigma\Pi)$ is independent of $\epsilon$. Since $\gmax(\Sigma) \Pi - \Pi\Sigma\Pi$ is positive semi-definite, by Lemma \ref{lem:anderson}, $\|\tilde{x}\|_2$ is stochastically dominated by $\sqrt{\gmax(\Sigma)}\|\Pi h\|_2$ (where $h \sim N(0, I_p)$), which in turn has the same distribution as $\sqrt{\gmax(\Sigma)}\|h'\|_2$ where $h' \sim N(0, I_{n_c})$.\\

By an application of a concentration inequality for Lipschitz functions of Gaussian vectors, we have for any $\delta > 0$ (see Example 2.28 of \cite{wainwright2019high}): 
$$ \bP(\|h'\|_2 \geq \sqrt{n_c} + \delta) \leq e^{-\delta^2/2}$$
and we also have $\|\epsilon\|_2 \leq \sigma\sqrt{2n}$ with probability at least $1-e^{-n/8}$ (see Example 2.11 of \cite{wainwright2019high}). Combining all the pieces yields the result. 
\end{proof}

\begin{lemma}[Choice of $\lambda_1$]\label{lem:lambda1}
With probability at least $1-\frac{2}{m}-e^{-n/8}$, we have
$$\|(\Gamma^\dagger)^TX^T\epsilon\|_\infty\leq 4\sigma\rho(\Gamma)\sqrt{\gmax(\Sigma)n\log p} $$
and hence $\lambda_1$ should be chosen such that $\lambda_1 \geq 32\sigma\rho(\Gamma)\sqrt{\frac{\gmax(\Sigma)\log p}{n}}$
\end{lemma}

\begin{proof}
Recall that the columns of $\Gamma^\dagger \in \mathbb{R}^{p\times m}$ are denoted as $s_1, \cdots, s_m$, and let the rows of $X$ be $x_1, \cdots, x_n$, which by assumption are i.i.d. $N(0,\Sigma)$ vectors.  \\

For any $t>0$:
\begin{align*}
    &\bP(\|(\Gamma^\dagger)^TX^T\epsilon\|_\infty \geq t) \\
    &= \bP\left(\max_{j\in [m]} \bigg|\bigg\langle s_j, \sum_{i=1}^n\epsilon_i x_i \bigg\rangle\bigg| \geq t\right)\\
    &= \bP\left(\|\epsilon\|_2\max_{j\in [m]} \bigg|\bigg\langle s_j, \frac{1}{\|\epsilon\|_2}\sum_{i=1}^n\epsilon_i x_i \bigg\rangle\bigg|\geq t\right) \\
    &\leq \bP\left(\sqrt{2n}\sigma\max_{j\in [m]} \bigg|\bigg\langle s_j, \frac{1}{\|\epsilon\|_2}\sum_{i=1}^n\epsilon_i x_i \bigg\rangle\bigg|\geq t\right) + \bP(\|\epsilon\|_2 > \sigma\sqrt{2n})
\end{align*}

Using the same trick as in Lemma \ref{lem:projection}, $x:= \frac{1}{\|\epsilon\|_2}\sum_{i=1}^n\epsilon_i x_i \sim N(0,\Sigma)$ independent of $\epsilon$. Also, we note again that $\bP(\|\epsilon\|_2 > \sigma\sqrt{2n}) \leq e^{-n/8}$. Hence, $\bP(\|(\Gamma^\dagger)^TX^T\epsilon\|_\infty \geq t)$ is bounded above by 
$$
    \bP\left(\sqrt{2n}\sigma\max_{j\in [m]} |\langle s_j, x\rangle|\geq t\right) + e^{-n/8}
    \leq \bP\left(\sqrt{2n\gmax(\Sigma)}\sigma \max_{j\in [m]}|\langle s_j, g\rangle|\geq t\right) + e^{-n/8}
$$
where $g \sim N(0,I_p)$ and we used Lemma \ref{lem:anderson} in the last inequality. Since $\{\langle s_j, g\rangle: j\in [m]\}$ are Normal variables with variance at most $\rho(\Gamma)^2$, by applying the union bound, the expression above can be bounded above by

$$2\exp\left(-\frac{t^2}{4\gmax(\Sigma)n\sigma^2\rho(\Gamma)^2} + \log m\right) + e^{-n/8} $$

If $t$ is chosen such that $t^2 = 8\log (m) \gmax(\Sigma)n\sigma^2\rho(\Gamma)^2$, we can conclude that 
$$\|(\Gamma^\dagger)^TX^T\epsilon\|_\infty \leq 2\sqrt{2}\sigma\rho(\Gamma)\sqrt{\gmax(\Sigma)n\log m} \leq 4\sigma\rho(\Gamma)\sqrt{\gmax(\Sigma)n\log p}$$
with probability at least $1-\frac{2}{m}-e^{-n/8}$, where we used $m \leq p^2$.
\end{proof}

\begin{lemma}[Lemma 3 of \cite{hutter2016optimal}]\label{lem:rigollet} If $\Gamma$ is the incidence matrix of a graph $G = (V, E)$ with maximum degree $d$ and $\emptyset \neq S \subseteq E$, then 
$$k_S^{-2} \leq 4\min(d, |S|) $$
\end{lemma}

\begin{lemma}[Lower bound in Lemma \ref{lem:2-dim-rho}]
If $\Gamma$ is the incidence matrix of the 2D grid, then $\rho(\Gamma) \gtrsim 1$.
\end{lemma}
\begin{proof}
Let $N := \sqrt{p}$. In the proof of Proposition 4 of \cite{hutter2016optimal}, it was shown that $\Gamma^\dagger$ has $2N(N-1)$ columns $((s_{i,j}^{(1)})_{\substack{i\in [N-1]\\ j\in [N]}}, ((s_{i,j}^{(2)})_{\substack{j\in [N-1]\\ i\in [N]}})$, each of which has column norm such that 
$$\|s_{i,j}^{(\diamond)}\|_2^2 = \sum_{k=0}^{N-1}\sum_{l=1}^{N-1}\frac{1}{(\lambda_k+\lambda_l)^2}\langle v_l, d_i\rangle^2\langle v_k, e_j \rangle^2 $$
where $\diamond \in \{1,2\}$, $\lambda_k = 2-2\cos \frac{k\pi}{N}$ for $0\leq k \leq N-1$ (these are the eigenvalues of the Laplacian of the one-dimensional chain graph with $N$ vertices), $d_i$ is the $i^\text{th}$ column of $D_1^T$ where $D_1$ is the incidence matrix of the chain graph with $N$ vertices, $e_1, \dots, e_n$ are the canonical basis vectors of $\mathbb{R}^N$, and $v_k\in \mathbb{R}^N$ ($0\leq k \leq N-1$) are the orthonormal eigenvectors of the Laplacian of the one-dimensional chain graph with $N$ vertices: 
$$(v_0)_j = \frac{1}{\sqrt{N}} $$
$$(v_k)_j = \sqrt{\frac{2}{N}}\cos\left(\frac{(j+1/2)k\pi}{N}\right) \text{ for } 0 \leq j \leq N-1, 1 \leq k \leq N-1 $$

Since $\rho(\Gamma)$ is defined as the maximum column norm of $\Gamma^\dagger$, we can bound it below by the column norm of $s_{i,j}^{(\diamond)}$ where $i=j=\diamond = 1$. We have:
$$\|s_{1,1}^{(1)}\|_2^2 = \sum_{k=0}^{N-1}\sum_{l=1}^{N-1}\frac{1}{\left(4-2\cos\frac{k\pi}{N}-2\cos\frac{l\pi}{N}\right)^2}\langle v_l, d_1\rangle^2\langle v_k, e_1\rangle^2$$

Using the inequality $2-2\cos(x) \leq x^2$, we have $4-2\cos \frac{k\pi}{N}-2\cos \frac{l\pi}{N}\leq \frac{\pi^2}{N^2}(k^2+l^2)$. Furthermore, note that 
\begin{align*}
    \langle v_l, d_1\rangle^2 &= \frac{2}{N}\left(\cos \frac{(5/2)l\pi}{N}-\cos \frac{(3/2)l\pi}{N}\right)^2 \\
    &= \frac{2l^2\pi^2}{N^3}\sin^2(x')
\end{align*}
for some $x' \in \left[\frac{(3/2)l\pi}{N}, \frac{(5/2)l\pi}{N}\right]$, by the mean value theorem. Given the inequality $\sin(x)\geq x/2$ for $x\in [0, \pi/2]$, we can conclude that $\sin^2(x') \geq (x')^2/4 \geq \frac{9}{16}\frac{l^2\pi^2}{N^2}$ if we assume $l \leq \frac{N}{5}$, and so 

$$\langle v_l, d_1\rangle^2 \geq \frac{9\pi^4}{8}\frac{l^4}{N^5} $$
if $l\leq N/5$. Moreover, $\langle v_k, e_1 \rangle^2 = \frac{1}{N}$ if $k = 0$, and if we assume $ k \leq \frac{2}{3\pi}N$, then $1-\frac{9}{8}\frac{k^2\pi^2}{N^2} \geq \frac{1}{2} $ and an application of $1-\cos(x)\leq \frac{x^2}{2}$ gives 
\begin{align*}
    \langle v_k, e_1 \rangle^2 = \frac{2}{N}\left[\cos\left(\frac{3k\pi}{2N}\right)\right]^2 &= \frac{2}{N}\left[1-\left(1-\cos \frac{3k\pi}{2N}\right)\right]^2\\ 
    &\geq \frac{2}{N}\left[1-\frac{9k^2\pi^2}{8N^2}\right]^2 \geq \frac{1}{2N}
\end{align*}

Hence, if $k \leq \frac{2}{3\pi}N$, we have $\langle v_k, e_1\rangle^2 \geq \frac{1}{2N}$. Let $c = \min\left(\frac{2}{3\pi}, \frac{1}{5}\right) = \frac{1}{5}$. Now, 
\begin{align*}
    \|s_{1,1}^{(1)}\|_2^2 &\geq \sum_{k=0}^{\lfloor cN\rfloor}\sum_{l=1}^{\lfloor cN\rfloor}\frac{1}{\left(4-2\cos\frac{k\pi}{N}-2\cos\frac{l\pi}{N}\right)^2}\langle v_l, d_1\rangle^2\langle v_k, e_1\rangle^2 \\ 
    &\geq \sum_{k=0}^{\lfloor cN\rfloor}\sum_{l=1}^{\lfloor cN\rfloor} \left(\frac{N^4}{\pi^4(k^2+l^2)^2}\right)\left(\frac{9\pi^4}{8}\frac{l^4}{N^5}\right)\left(\frac{1}{2N}\right)\\
    &= \frac{9}{16N^2}\sum_{l=1}^{\lfloor cN\rfloor}\sum_{k=0}^{\lfloor cN\rfloor} \frac{l^4}{(k^2+l^2)^2}
\end{align*}

Since $\frac{1}{(k^2+l^2)^2}$ is a decreasing function of $k$, 
\begin{align*}
     \|s_{1,1}^{(1)}\|_2^2 &\geq \frac{9}{16N^2}\sum_{l=1}^{\lfloor cN\rfloor}l^4\int_0^{ cN} \frac{1}{(x^2+l^2)^2}dx \\ 
     &= \frac{9}{16N^2}\sum_{l=1}^{\lfloor cN\rfloor}l^4\frac{(l^2+c^2N^2)\arctan(cN/l) + cNl}{2l^3(l^2+c^2N^2)}\\
     &= \frac{9}{32N^2}\sum_{l=1}^{\lfloor cN\rfloor}l\frac{(l^2+c^2N^2)\arctan(cN/l)+cNl}{l^2+c^2N^2} \\ 
     &\geq \frac{9}{32N^2}\sum_{l=1}^{\lfloor cN\rfloor}l\arctan(c) \\ 
     &= \frac{9\arctan(c)}{32N^2}\frac{\lfloor cN\rfloor (\lfloor cN\rfloor+1)}{2} \gtrsim 1
\end{align*}
\end{proof}

\begin{lemma}[Lower bound in Lemma \ref{lem:r-dim-rho}]
If $\Gamma$ is the incidence matrix of the r-dimensional grid for $r\geq 3$, then $\rho(\Gamma) \geq c(r)$, where $c(r)$ is a constant depending only on $r$.
\end{lemma}
\begin{proof}
    Note that the $\gtrsim$ sign is used in this proof to omit constant multipliers that may depend on $r$. Similarly to the previous lemma, it is sufficient to lower bound
    $$\|s_{\boldsymbol{1}}^{(1)}\|_2^2 = \sum_{l=1}^{N-1}\sum_{k_1=0}^{N-1}\dots \sum_{k_{r-1}=0}^{N-1}\frac{\langle v_l, d_1\rangle^2\prod_{j=1}^{r-1}\langle v_{k_j}, e_1\rangle^2}{(\lambda_l + \sum_{j=1}^{r-1}\lambda_{k_j})^2} $$
where $d_1, e_1$ as well as $\lambda_0, \dots \lambda_{N-1}$ and $v_0, \dots, v_{N-1}$ are defined in relation to the chain graph with $N$ vertices as in the previous lemma. 

By applying the inequality $2-2\cos(x) \leq x^2$, we have 
$$\left(\lambda_l + \sum_{j=1}^{r-1}\lambda_{k_j}\right)^2 \leq \left(2-2\cos \frac{l\pi}{N} + \sum_{j=1}^{r-1}\left(2-2\cos\frac{k_j\pi}{N}\right)\right)^2 \leq \frac{\pi^4}{N^4}\left(l^2+\sum_{j=1}^{r-1}k_j^2\right)^2 $$

Also, $k\leq \frac{2}{3\pi}N$ implies $\langle v_k, e_1\rangle^2 \geq \frac{1}{2N}$, and $l\leq N/5$ implies $\langle v_l, d_1\rangle^2 \geq \frac{9\pi^4}{8}\frac{l^4}{N^5}$. Hence, if we define $c =\min(\frac{2}{3\pi},\frac{1}{5})= \frac{1}{5}$ as in the previous lemma,
\begin{align*}
    \|s_{\boldsymbol{1}}^{(1)}\|_2^2 &\gtrsim \frac{1}{N^r}\sum_{l=1}^{\lfloor cN\rfloor}\sum_{k_1=0}^{\lfloor cN\rfloor}\dots \sum_{k_{r-1}=0}^{\lfloor cN\rfloor}\frac{l^4}{(l^2+\sum_{j=1}^{r-1}k_j^2)^2} \\ 
    &\geq \frac{1}{N^r}\sum_{l=1}^{\lfloor cN\rfloor}\int_{0\leq x_j \leq cN, j=1,\dots,r-1 } \frac{l^4}{(l^2+\|x\|_2^2)^2}dx \\
    &\geq \frac{1}{N^r}\sum_{l=1}^{\lfloor cN\rfloor}\int_{\|x\|_2\leq cN} \frac{l^4}{(l^2+\|x\|_2^2)^2}dx\\
    &= \frac{1}{N^r}\sum_{l=1}^{\lfloor cN\rfloor}\int_0^{cN}\int_{S_{r-2}}\frac{l^4R^{r-2}}{(l^2+R^2)^2}d\sigma_{r-2}(u)dR
\end{align*}
where we changed to polar coordinates in the last equality; here, $S_{r-2}$ is the unit sphere in $\mathbb{R}^{r-1}$, and $\sigma_{r-2}$ is a measure on $S_{r-2}$ such that, if $A\subseteq S_{r-2}$ is a Borel set and $\tilde{A}$ is the set of all points $ru$ with $0<r<1$ and $u\in A$, then $\sigma_{r-2}(A) = (r-1)m_{r-1}(\tilde{A})$, where $m_{r-1}$ is the Lebesgue measure on $\mathbb{R}^{r-1}$ (see Exercise 6, Chapter 8 of \cite{rudin1974real}). We continue: 
\begin{align*}
    \|s_{\boldsymbol{1}}^{(1)}\|_2^2 &\gtrsim \frac{1}{N^r}\sum_{l=1}^{\lfloor cN\rfloor}\int_0^{cN}\frac{l^4R^{r-2}}{(l^2+R^2)^2}dR \\ 
    &\geq \frac{1}{N^r}\sum_{l=1}^{\lfloor cN\rfloor}\int_{cN/2}^{cN}\frac{l^4R^{r-2}}{(l^2+R^2)^2}dR \\ 
    &\gtrsim \frac{1}{N^3}\sum_{l=1}^{\lfloor cN\rfloor}\int_{cN/2}^{cN}\frac{l^4R}{(l^2+R^2)^2}dR
\end{align*}
where we used the fact that $r\geq 3$. Note that $\int_a^b\frac{R}{(l^2+R^2)^2}dR = \frac{b^2-a^2}{2(b^2+l^2)(a^2+l^2)}$ and hence 
\begin{align*}
     \|s_{\boldsymbol{1}}^{(1)}\|_2^2 &\gtrsim\frac{1}{N^3}\sum_{l=1}^{\lfloor cN\rfloor}\frac{l^4N^2}{(l^2+N^2)^2} \geq \frac{1}{N}\int_{0}^{N/10}\frac{l^2}{(l^2+N^2)^2}dl
\end{align*}
where we used the fact that $\frac{l^2}{(l^2+N^2)^2}$ is increasing in $l$. Since
$$\int_{0}^{N/10}\frac{l^2}{(l^2+N^2)^2}dl = \frac{151-1515\arctan(1/10)}{1010}N $$
and $151-1515\arctan(1/10) > 0$, the proof is complete. 
\end{proof}

%% file: appendix_algo.tex
For the special case where the design matrix is the identity and $\lambda_2 = 0$, \cite{kim2009ell_1} applies the interior point method on the dual objective. Similarly, we can apply interior point method to solve our more general dual objective
\begin{equation}\label{eq:A29}
    \hat{\beta}:= \arg\min_{\beta\in\mathbb{R}^p}\frac{1}{2}\|\tilde{Y}-\tilde{X}\beta\|_2^2 + \lambda_1\|\Gamma\beta\|_1
\end{equation}

We will specify the update directions, update step size, and measure of suboptimality. Let $\bm{1}_m$ denote the vector in $\mathbb{R}^m$ with all entries equal to 1. The dual problem only has inequality constraints $f_1(u) = u - \lambda_1 \bm{1}_m \leq 0$ and $f_2(u) = -u - \lambda_1 \bm{1}_m \leq 0$. Let $\mu_1, \mu_2$ be the dual variables corresponding to $f_1, f_2$. We apply the standard Newton's updates on the perturbed KKT conditions (by parameter $t$) of this dual problem. That is, the directions of the updates $(\Delta u, \Delta \mu_1, \Delta \mu_2)$ are solutions of the following linear system: 

The residuals are:
\begin{equation}
\begin{split}
    r_t(u, \mu_1,\mu_2) = 
    \begin{bmatrix}
    \check{\Gamma}\check{\Gamma}^Tu -    \check{\Gamma}\check{Y} + \mu_1 -\mu_2\\
    -\mathrm{diag}(\mu_1)f_1(u) - \frac{1}{t}\bm{1}_m\\
    -\mathrm{diag}(\mu_2)f_2(u) - \frac{1}{t}\bm{1}_m
    \end{bmatrix}
\end{split}
\end{equation}

By Newton's method, we need to solve:
\begin{equation}
    \nabla r_t(u, \mu_1,\mu_2) \begin{bmatrix}
    \Delta u\\
    \Delta \mu_1\\
    \Delta \mu_2
    \end{bmatrix} =  -r_t(u, \mu_1,\mu_2) 
\end{equation}

That simplifies to 3 linear equations below, where divisions between vectors are element-wise:
\begin{equation} \label{eq:IP_updates_u}
 \left[ \check{\Gamma}\check{\Gamma}^T - \mathrm{diag}(\mu_1 /f_1(u)) - \mathrm{diag}(\mu_2 /f_2(u)) \right] \Delta u  = -\left[\check{\Gamma}\check{\Gamma}^Tu -    \check{\Gamma}\check{Y} - \frac{\bm{1}_m}{tf_1(u) } + \frac{\bm{1}_m}{tf_2(u) }\right]
 \end{equation}
 
\begin{equation} \label{eq:IP_updates_mu1}
 \Delta \mu_1  = - \left[\mathrm{diag}(\mu_1 /f_1(u))\Delta u + \mu_1 + \frac{\bm{1}_m}{tf_1(u) }\right] 
\end{equation}

\begin{equation} \label{eq:IP_updates_mu2}
 \Delta \mu_2  = - \left[-\mathrm{diag}(\mu_2 /f_2(u))\Delta u  + \mu_2 + \frac{\bm{1}_m}{tf_2(u) }\right]
\end{equation}
The step size for each update are computed in a standard way as Section 11.7.3 of \cite{boyd2004convex}. We apply standard backtracking line search to find the step size $s$ for the updates. Choose parameters $\alpha, \gamma \in (0,1)$ for backtracking. Denote the updates as  $(u^+, \mu_1^+,\mu_2^+)$. For example, $u^+ = u + s\Delta u$. To ensure the updates to be feasible, we first make sure that $\mu_1^+, \mu_2^+ \geq 0$. That is, we set $s_{\max} = \min\{1, \min\{-\mu_{1i}/\Delta\mu_{1i} \:|\: \Delta\mu_{1i} < 0\}, \min\{-\mu_{2i}/\Delta\mu_{2i} \:|\: \Delta\mu_{2i} < 0\}\} $.
Next, continuously set $s = \gamma s$ until $f_1(u^+), f_2(u^+) < 0$. Finally, set $s = \gamma s$ until $ \|r_t(u^+, \mu_1^+,\mu_2^+)\|_2\leq (1 - \alpha s)\|r_t(u, \mu_1,\mu_2)\|_2$.

As a standard measure of suboptimality, the surrogate duality gap (see Section 11.7.2 of \cite{boyd2004convex} for details) at the $k^\text{th}$ iteration is: 
\begin{equation}\label{eq:surrogate_gap}
\eta^{(k)} = -f_1(u^{(k)})^T\mu^{(k)}_1 -f_2(u^{(k)})^T\mu^{(k)}_2
\end{equation}
And the residual at the $k^\text{th}$ iteration is:
\begin{equation}
r^{(k)} = r_t(u^{(k)}, \mu_1^{(k)},\mu_2^{(k)})
\end{equation}

Our interior point algorithm is presented below.\\
\begin{algorithm}[H]\label{alg:2}
\DontPrintSemicolon
  \KwInput{$\lambda_1, \lambda_2, \Gamma, Y, X$, tolerance $\epsilon$}
  \KwOutput{$\hat{\beta}$ as defined in \eqref{eq:A29}}
  Initialize $u^{(0)} = 0$, $\mu_1^{(0)}, \mu_2^{(0)} >0$, $\tau > 1$ \\
  \While{$r^{(k)} > \epsilon$ or $\eta^{(k)} > \epsilon$}
   {
   		Set $t = 2\tau m/ \eta^{(k)}$\\
        Compute update direction $(\Delta u, \Delta \mu_1, \Delta \mu_2)$ as in \eqref{eq:IP_updates_u},  \eqref{eq:IP_updates_mu1},  \eqref{eq:IP_updates_mu2}\\
        Determine step size $s$ by using $\alpha, \gamma$ backtracking line search\\
        Update:
        \begin{align*}
            u^{(k+1)} &= u^{(k)} + s\Delta u\\
            \mu_1^{(k+1)} &=\mu_1^{(k)} + s\Delta\mu_1\\
            \mu_2^{(k+1)} &=\mu_2^{(k)} + s \Delta \mu_2
        \end{align*}
   }
   Compute $\hat{\beta} \leftarrow \tilde{X}^\dagger(\check{Y} - \check{\Gamma}^T u)$\\
   Return $\hat{\beta}$
   
\caption{Interior point method on the dual objective}
\end{algorithm}

%% file: appendix_data.tex
\xhdr{Chicago Crime Data} 
As per the main text, statistics on the number of crimes per community between 2001 and 2022 are available on the city's data portal. For the purpose of our analysis, we consider the data between 2004 and 2022, since by preliminary inspection of the data, the first years of collection seem to have more missing data (see Figure~\ref{fig:seasonality1}). We define the monthly crime rates as the number of crimes per 100,000 inhabitants. The latter are computed from the raw crime data by aggregating crime counts over neighborhoods and dividing by neighborhood population estimates found at the following \hyperlink{https://en.wikipedia.org/wiki/Community_areas_in_Chicago}{ link}. These crime rates are usually modeled by Poisson distributions \cite{osgood2000poisson}, which we transform here into a normal distribution through the use of an Anscombe transform. Examples of the resulting estimates are displayed on Figure~\ref{fig:seasonality2}.  We note that the crime rates vary substantially over the years and across the communities, and are also subject to significant seasonal effects.

\begin{figure}[h!]
     \centering
     \begin{subfigure}[t]{0.49\textwidth}
     \centering
    \includegraphics[width=\textwidth]{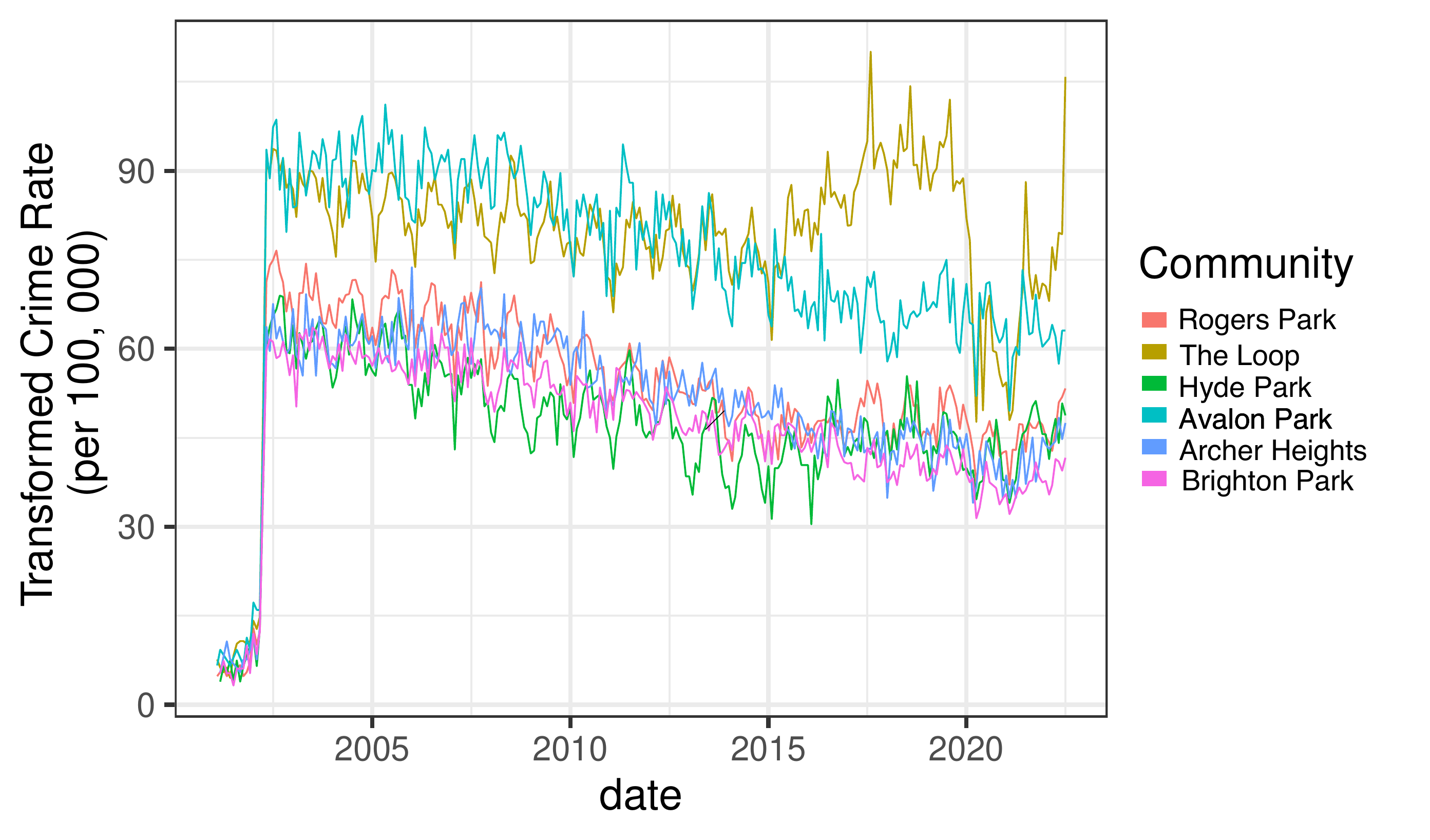}
    \caption{Temporal evolution of the Anscombe transformed crime rate (per 100,000) as a function of time across 6 specific neighborhoods.}
    \label{fig:seasonality1}
     \end{subfigure}
     \hfill
     \begin{subfigure}[t]{0.5\textwidth}
     \centering
    \includegraphics[width=\textwidth]{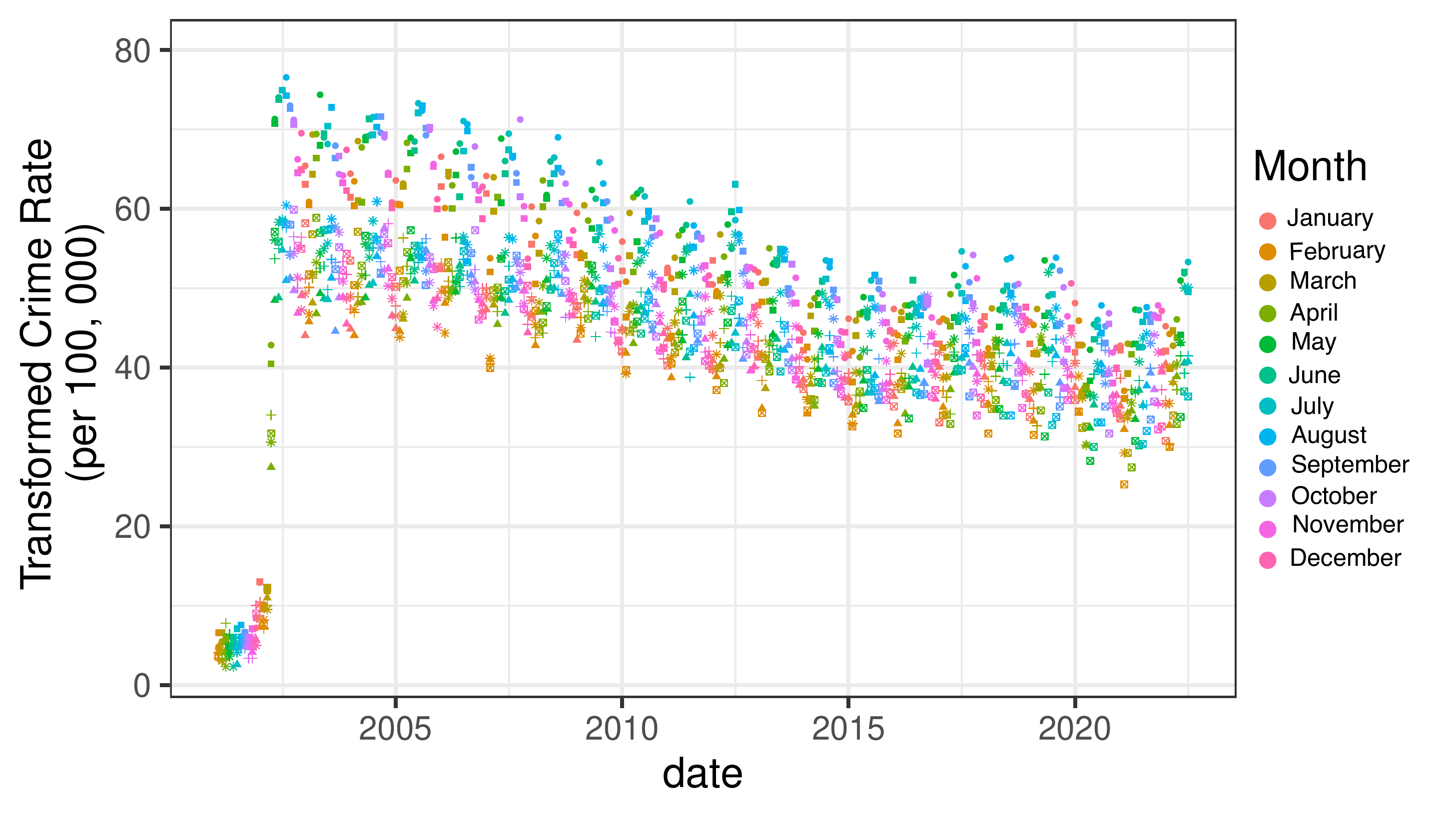}
        \caption{Crime rate across neighborhoods, coloured by month of the year to highlight seasonal effects. }
    \label{fig:seasonality2}
     \end{subfigure}
    \caption{Anscombe transform of the number of crimes per month per 100,000 inhabitants in a few neighborhoods of Chicago. Note the seasonal effect in the crime rate and the  consistent drop across neighborhoods during the coldest months of the year.}
\end{figure}
\vspace{0.2cm}


\xhdr{COVID Data} We consider the problem of predicting the number of COVID-19 cases 14 days in advance for a given county in California. As decribed in the main text, this could be an interesting use case for local public health decisions, such as for instance, trying to plan 2 weeks in advance appropriate resources at a local clinic. To this end, we used the New York Times-curated \href{https://github.com/nytimes/covid-19-data}{COVID database}. The NYT COVID data provides a description of the total number of cases across all US counties, from January 2020 to October 2022 (time of writing).  For the purpose of our analysis, we focus more specifically on analyzing new cases in the 25 densest California counties using data from June $1^{\text{st}}$, 2020 to July $1^{\text{st}}$, 2021. This time window was selected to provide more consistency in the epidemics dynamics: by June 2020, all counties in California had non zero daily incidence data. On the other hand, restricting the analysis to before July 2021 allows selecting a more cohesive window of time where the epidemic propagation was not dominated by (other unobserved) covariates, such as the advent of new contagious strains of the virus (Delta in Summer 2021, and subsequently Omicron in Winter 2022). We pre-process the data and make it amenable to data analysis through the following steps:
\begin{enumerate}
    \item {\bf Conversion of cumulative case counts to incidence data}
    \item {\bf  Correction of aberrations and smoothing}: we fix data aberrations (e.g. negative incidences, due to small errors in the reporting) by imposing the lower bound on the number of new cases to be 0. We further transform the incidence data using a seven-day rolling average so as to get rid of known spurious phenomena (e.g. the ``weekend effect'', by which the number of new cases is lower over the weekend but typically followed by a spike on the following Monday).
    \item {\bf Anscombe transform.} We apply a variance stabilizing transform to transform incidence data (here modeled as a Poisson process, as per \cite{agosto2020poisson, bu2021likelihood,cori2013new,toharudin2020bayesian}): $\tilde{x} \leftarrow 2 \sqrt{x + \frac{3}{8}}$.
\end{enumerate}

\begin{figure}
    \centering
    \includegraphics[width=0.8\textwidth]{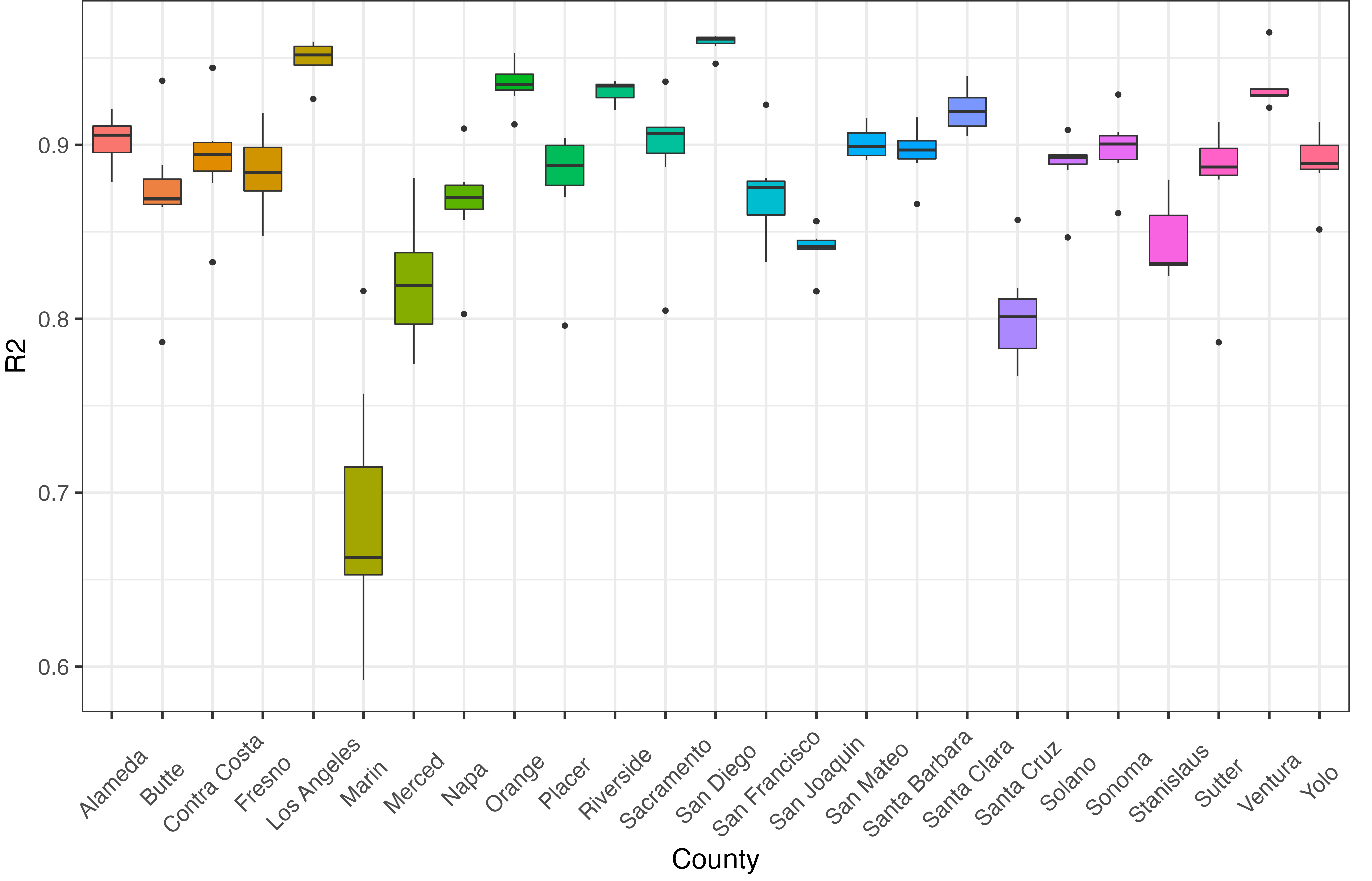}
    \caption{$R^2$ for the simple autoregressive model of Equation \ref{eq:covid-model1} on the seven different folds (see main text). Note that most models have $R^2$ of over 0.8, thus indicating the validity of the model.}
    \label{fig:my_label}
\end{figure}